\documentclass[runningheads]{LNCS/llncs}

\usepackage[
  paperwidth=152mm,
  paperheight=235mm,
  textwidth=122mm,
  textheight=193mm,
  centering
]{geometry}

\usepackage{xcolor}
\usepackage{hyperref}
\usepackage{listings}
\usepackage{amssymb}
\usepackage{amsmath}

\spnewtheorem{axiom}[theorem]{Axiom}{\bfseries}{\itshape}

\usepackage[utf8]{inputenc}
\usepackage[T1]{fontenc}
\usepackage{fontspec} 

\usepackage[bottom]{footmisc}

\usepackage{colorblind}
\usepackage{wrapfig}
\usepackage{graphicx}
\usepackage{caption}
\usepackage{subcaption}
\usepackage{semantic}
\usepackage{simplebnf}
\usepackage{tabularray}
\usepackage{ninecolors}
\usepackage{stmaryrd} 
\usepackage{tikz}\usetikzlibrary{intersections}\usetikzlibrary{matrix}\usetikzlibrary{shapes.arrows,positioning,arrows.meta,bending,shapes,decorations.markings,fit,backgrounds,calc,decorations.pathmorphing,decorations.pathreplacing}
\usetikzlibrary{automata,positioning}

\def\fadeline#1#2#3{%
  \path let
    \p1 = ($(#2)-(#1)$),
    \n1 = {veclen(\p1)},
    \n2 = {atan2(\y1,\x1)} 
  in (#1) -- (#2) node[#3, midway, rotate = \n2, shading angle = \n2+90, minimum width=\n1, inner sep=0pt, draw=none] {};
}

\usepackage{pgfplots}

\tikzset{
  -|-/.style={
  to path={
    (\tikztostart) -- ({$(\tikztostart)!#1!(\tikztotarget)$} |- {$(\tikztostart)$}) |- (\tikztotarget)
    \tikztonodes
  }, pos=0.25
  },
  -|-/.default=0.5,
  |-|/.style={
  to path={
    (\tikztostart) -- ({$(\tikztostart)!#1!(\tikztotarget)$} -| {$(\tikztostart)$}) -| (\tikztotarget)
    \tikztonodes
  }, pos=0.25
  },
  |-|/.default=0.5,
  -/.style={
  to path={
    (\tikztostart) -- ({$(\tikztostart)$} -| {$(\tikztotarget)$})
    \tikztonodes
  }
  },
  |/.style={
  to path={
    (\tikztostart) -- ({$(\tikztostart)$} |- {$(\tikztotarget)$})
    \tikztonodes
  }
  }
}
\usepackage{multirow}
\usepackage{soul}
\usepackage{placeins}

\usepackage[noabbrev, nameinlink, capitalize]{cleveref}
\AddToHook{cmd/appendix/before}{%
  \crefalias{section}{appendix}%
  \crefalias{subsection}{appendix}
}
\makeatletter
\if@cref@capitalise
  \crefname{axiom}{Axiom}{Axioms}
\else
  \crefname{axiom}{axiom}{axioms}
\fi
\makeatother
\Crefname{axiom}{Axiom}{Axioms}

\usepackage{xspace}

\newcommand{\coq}{\text{Rocq}\xspace}
\newcommand{\rocq}{\text{Rocq}\xspace}
\newcommand{\ocaml}{\text{OCaml}\xspace}

\newcommand{\ilstRocq}[1]{\lstinline{#1}}
\newcommand{\ilstrocq}[1]{%
  \ifmmode
    \text{{\lstset{basicstyle=\ttfamily,style=coq,mathescape}\lstinline!#1!}}%
  \else
    {\lstset{basicstyle=\ttfamily,style=coq,mathescape}\lstinline!#1!}%
  \fi
}

\newcommand{\koika}{\text{K\^oika}\xspace}
\newcommand{\kami}{\text{Kami}\xspace}
\newcommand{\fjfj}{\text{Fjfj}\xspace}

\newcommand{\R}{\ensuremath{\mathcal{R}}\xspace}
\newcommand{\cL}{\ensuremath{\mathcal{L}}\xspace}
\newcommand{\Ctx}{\ensuremath{\Gamma}\xspace}
\newcommand\mdoubleplus{\mathbin{+\mkern-10mu+}}
\newcommand{\kread}[2]{\ensuremath{
    \texttt{read}_{#1}\texttt{(}#2\texttt{)}
    }}
\newcommand{\kwrite}[3]{\ensuremath{
    \texttt{write}_{#1}\texttt{(}#2\texttt{,} #3\texttt{)}
    }}

\newcommand{\pre}[1]{\textcolor{orange}{#1}}
\newcommand{\post}[1]{\textcolor{green}{#1}}
\newcommand{\hoare}[3]{\ensuremath{ \pre{\{ #1 \}} ~ #2 ~ \post{\{ #3 \}}  }}
\newcommand{\thoare}[3]{\ensuremath{ \pre{[#1]} ~ #2 ~ \post{[ #3 ]}  }}
\newcommand{\wpt}[3]{\ensuremath{ \texttt{WP}_{\!\texttt{#1}} ~ #2 ~ \post{[ #3 ]}  }}
\newcommand{\wpp}[3]{\ensuremath{ \texttt{WP}_{\!\texttt{#1}} ~ #2 ~ \post{\{ #3 \}}  }}

\newcommand{\noc}{\text{NoC}\xspace}
\newcommand{\cuttlec}{\texttt{cuttlec}\xspace}

\newcommand{\sys}{\text{NoC-Out}\xspace}
\newcommand{\genoc}{\text{GeNoC}\xspace}
\newcommand{\acl}{\text{ACL2}\xspace}

\newcommand{\oursem}{Read-After-Write semantics\xspace}
\newcommand{\koikasemfun}{\ensuremath{\downarrow_{\R,L}}}
\newcommand{\oursemfun}{\ensuremath{\downarrow_{\R,L}^{\text{raw}}}}

\definecolor{primary}       {RGB}{   0,  48,  94 }
\definecolor{secondary}     {RGB}{   0, 105, 180 }
\definecolor{gray}          {RGB}{ 114, 119, 119 }
\definecolor{tertiary}      {RGB}{   0, 159, 227 }

\definecolor{green1}        {RGB}{ 148, 195,  86 }
\definecolor{teal1}         {RGB}{ 138, 203, 193 }
\definecolor{blue1}         {RGB}{ 132, 207, 237 }
\definecolor{magenta1}      {RGB}{ 238, 123, 174 }
\definecolor{red1}          {RGB}{ 240, 130,  98 }
\definecolor{orange1}       {RGB}{ 247, 169,  65 }

\definecolor{green2}        {RGB}{ 101, 179,  46 }
\definecolor{teal2}         {RGB}{   0, 172, 169 }
\definecolor{blue2}         {RGB}{   0, 161, 217 }
\definecolor{magenta2}      {RGB}{ 224,  49, 138 }
\definecolor{red2}          {RGB}{ 232,  65,  44 }
\definecolor{orange2}       {RGB}{ 239, 125,   0 }

\definecolor{green3}        {RGB}{   0, 137,  58 }
\definecolor{teal3}         {RGB}{   0, 131, 141 }
\definecolor{blue3}         {RGB}{   0, 119, 174 }
\definecolor{magenta3}      {RGB}{ 206,   0, 117 }
\definecolor{red3}          {RGB}{ 205,  23,  25 }
\definecolor{orange3}       {RGB}{ 201,  75,  23 }

\definecolor{c11}{RGB}{  0, 173, 239}
\definecolor{c12}{RGB}{  0, 185, 241}
\definecolor{c13}{RGB}{ 65, 199, 244}

\definecolor{c21}{RGB}{  0, 113, 187} 
\definecolor{c22}{RGB}{ 30, 131, 197}
\definecolor{c23}{RGB}{101, 154, 209}

\definecolor{c31}{RGB}{ 87,  63, 152} 
\definecolor{c32}{RGB}{113,  93, 168}
\definecolor{c33}{RGB}{140, 124, 185}

\definecolor{c41}{RGB}{145,  38, 143} 
\definecolor{c42}{RGB}{160,  83, 160}
\definecolor{c43}{RGB}{178, 121, 180}

\definecolor{c51}{RGB}{  0, 140,  78} 
\definecolor{c52}{RGB}{  6, 155, 105}
\definecolor{c53}{RGB}{ 89, 174, 134}

\definecolor{c61}{RGB}{ 97, 187,  69} 
\definecolor{c62}{RGB}{131, 197, 101}
\definecolor{c63}{RGB}{161, 209, 138}

\colorlet{green}  {green3}
\colorlet{teal}   {teal3}
\colorlet{blue}   {blue3}
\colorlet{magenta}{magenta3}
\colorlet{red}    {red3}
\colorlet{orange} {orange3}

\def\setallcolors{%
\gdef\kwcolor{\color{primary}\bfseries}%
\gdef\idcolor{\color{black}}%
\gdef\commentcolor{\color{gray}}%
\gdef\stringcolor{\color{green3}}
\gdef\namespacecolor{\color{c31}}%
\gdef\classcolor{\color{teal3}}%
\gdef\methodcolor{\color{primary}}%
\gdef\macrocolor{\color{c41}}%
\gdef\numcolor{\color{teal3}}%
}
\def\setallgray{%
\gdef\kwcolor{\color{gray}\bfseries}%
\gdef\idcolor{\color{gray}}%
\gdef\commentcolor{\color{gray}}%
\gdef\stringcolor{\color{gray}}%
\gdef\namespacecolor{\color{gray}}%
\gdef\classcolor{\color{gray}}%
\gdef\methodcolor{\color{gray}}%
\gdef\macrocolor{\color{gray}}%
\gdef\numcolor{\color{gray}}%
}
\setallcolors{}
\lstdefinelanguage{coq}{
  morekeywords={Definition, Fixpoint, Lemma, Theorem,
  Example, Proof, Inductive, Notation, Arguments, Hint,
  Instance, Ltac, Class, Variant, CoInductive, CoFixpoint, Context},
  sensitive=true, 
  morecomment=[s]{(*}{*)}, 
  morestring=[b]{"} 
}
\lstdefinelanguage{verilog}{
  morekeywords={module, endmodule, generate, assert, property, endproperty, endgenerate, always_ff, begin, end, if, else, for, genvar},
  sensitive=true, 
  morecomment=[l]{//}, 
  morecomment=[s]{/*}{*/}, 
  morestring=[b]{"} 
}
\lstdefinestyle{verilog}{
  language=verilog,
  emph      = [1]{logic, parameter, input, int},
  emphstyle = [1]{\macrocolor},
  emph      = [2]{Type, Prop},
  emphstyle = [2]{\classcolor},
  emph      = [3]{}, 
  emphstyle = [3]{\methodcolor},
  emph      = [4]{}, 
  emphstyle = [4]{\namespacecolor},
  moredelim=**[is][\classcolor]{@[}{]@},
  moredelim=**[is][\setallgray\color{gray}\aftergroup\setallcolors]{@}{@},
  moredelim=**[is][\numcolor]{§}{§},
  moredelim=**[is][\namespacecolor]{§4:}{§},
  moredelim=**[is][\methodcolor]{§3:}{§},
  moredelim=**[is][\classcolor]{§2:}{§},
  moredelim=**[is][\macrocolor]{§1:}{§},
}
\lstdefinestyle{coq}{
  language=coq,
  emph      = [1]{auto, match, with, end, fun, fix, assert, let, in, exact, ltac, ltac2, vm_compute, if, then, else, λ},
  emphstyle = [1]{\macrocolor},
  emph      = [2]{Type, Prop},
  emphstyle = [2]{\classcolor},
  emph      = [3]{}, 
  emphstyle = [3]{\methodcolor},
  emph      = [4]{}, 
  emphstyle = [4]{\namespacecolor},
  moredelim=**[is][\classcolor]{@[}{]@},
  moredelim=**[is][\setallgray\color{gray}\aftergroup\setallcolors]{@}{@},
  moredelim=**[is][\numcolor]{§}{§},
  moredelim=**[is][\namespacecolor]{§4:}{§},
  moredelim=**[is][\methodcolor]{§3:}{§},
  moredelim=**[is][\classcolor]{§2:}{§},
  moredelim=**[is][\macrocolor]{§1:}{§},
}
\lstdefinestyle{frame}{
  frame=l,
  framerule=0.8pt,
  framesep=4pt,
  xleftmargin=6.8pt, 
}
\makeatletter
\lst@AddToHook{Init}{\setlength{\lineskip}{0pt}}
\makeatother

\begin{document}

\title{\sys: A Formally-verified Network-on-Chip Library for Rule-based Hardware Designs}
\titlerunning{\sys}

\author{Max Kurze\orcidID{0009-0006-3404-2569} \and
František Farka\orcidID{0000-0001-8177-1322} \and
Sebastian Ertel\orcidID{0009-0000-3953-9810}}

\authorrunning{M. Kurze et al.}

\institute{Barkhausen Institut, Germany\\
\email{\{first\}.\{last\}@barkhauseninstitut.org}}

\maketitle

\begin{abstract}

%
The Network-on-Chip (\noc) is the communication backbone of
any multiprocessor chip.
A failure of the \noc has severe consequences for the whole
system.
Yet, no approach exists that provides formally-verified NoCs
with strong guarantees but without tedious verification effort.
%


%
Any library that generates formally-verified \noc{}s needs to
be parametric in the structure of the \noc.
This requires a hardware description language (HDL) that allows
for parametric and concurrent yet efficient hardware designs as well as the
necessary program logic to reason about them in a modular fashion.
So far, HDLs fall short in both aspects.
%


%
In this paper, we implement \sys, the first library/generator for formally-verified
$k$-dimensional \noc designs.
In order to build \sys, we extended \koika, a rule-based HDL
in the \rocq theorem prover, with support for concurrent yet efficient
\noc designs and a program logic for modular, automated reasoning.
Given a configuration, \sys produces a $k$-dimensional torus \noc in \koika,
which can then be compiled to Verilog.
Each produced \noc is equipped with a proof that it refines our formal
\noc specification;
no additional verification effort is required.
Our specification proves a strong liveness guarantee, which
consequently applies to all generated \noc{}s.
In our evaluation, we find that our verification approach is
even required to synthesize efficient \noc{s} in rule-based HDLs.

\keywords{Formal verification \and Network-on-chip \and Theorem proving \and Program logic}
\end{abstract}

\section{Introduction}
\label{sec:intro}


%
Modern hardware architectures are heterogeneous platforms
that compose full-fledged processors with specialized
accelerators in a \emph{network-on-chip (\noc)}~\cite{bertozzi2004xpipes,10.1145/1132952.1132953,kumar2002network}.
A failure or compromise of the \noc is not only costly but
jeopardizes the trust of the whole hardware platform~\cite{10.1145/3450964}.
Formal verification of strong guarantees for \noc{}s requires
expert knowledge.
%
%
Yet, an approach to \textbf{generate formally-verified \noc designs
with strong guarantees at zero end-user verification cost} does
not exist.
%


%
A network-on-chip consists of a set of \emph{routers}
interconnected via hardware channels to enable the
components of a hardware platform to
communicate with each other via messages.
As such, this network is essential to the platform and
its components.
But every hardware platform needs to meet different
requirements.
Some are optimized for performance.
Others focus on energy-efficiency.
And resistance to side-channel attacks makes security
 a pressing concern for many.
As such, hardware platforms are inherently heterogeneous
and integrate full-fledged (RISC-V) processors, tensor
compute units for machine learning, crypto accelerators,
and other domain-specific circuits.
To optimize the communication in these platforms, very
different \noc architectures and implementations exist.
Existing approaches to the formal verification of \noc designs
need to be carried out for every \noc individually,
and do not connect to actual hardware designs~\cite{genoc}.
Most notably, \genoc, a generic \noc model defined in \acl,
emits local obligations that every specific model needs to prove.
From these local obligations global properties can be inferred automatically.
However, these models are not directly connected to the corresponding
hardware implementations.
As a result, such specifications are often high-level, making them easier
to verify but less likely to match the implementation.
Meanwhile, promising approaches to formally-verified hardware
design exist but lack the necessary infrastructure.
\koika, \kami and \fjfj are hardware description
languages (HDLs) implemented in \rocq~\cite{koika,kami,fjfj}.
All three HDLs are rule-based for which the \noc presents a
particularly challenging use case when it comes to generating
efficient hardware.
All three HDLs come with well-defined and mechanized semantics.
As such, proving that an actual hardware program
refines a formal specification is possible.
Yet, scaling such proofs even to a small \noc remains a
daunting task, as a program logic for modular reasoning
together with corresponding proof automation is still missing.
As such, many \noc{}s remain unverified and thus
not trustworthy.
%


%
In order to provide strong guarantees without end-user verification
efforts, we incorporate the \noc structure into the specification.
To scale this to more than one \noc, we abstract the network structure
to $k$-dimensional tori.
To connect this specification to actual hardware implementations,
we define it in the \rocq theorem prover and introduce the missing
infrastructure into \koika to prove that a \noc implementation follows
this specification.
%


%
In this paper, we design a formally-verified \noc generator
in \koika, called \emph{\sys}.
\sys takes a configuration that specifies the number of
dimensions $k$ and produces a corresponding $k$-dimensional \noc
hardware description.
Every produced \noc is formally-verified against the
specification and requires no additional verification effort.
Furthermore, to demonstrate the soundness of the specification, we
prove a global liveness guarantee on top of it.

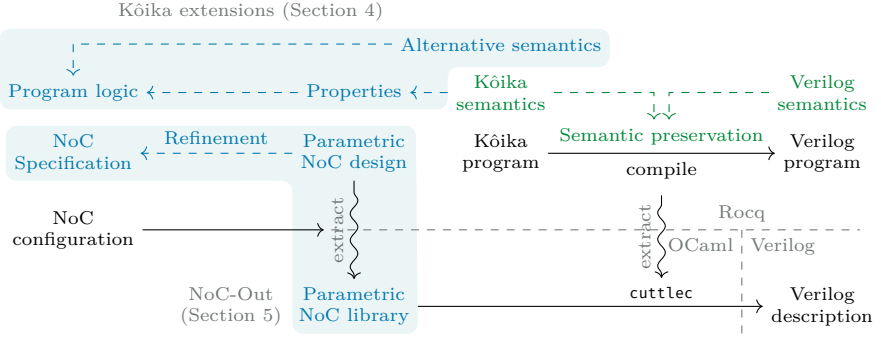
\begin{figure}
  \centering
  \begin{tikzpicture}[
  node distance=2cm,
  every node/.style={align=center},
  color2/.style={green},
  color3/.style={blue},
  colornew/.style={fill=teal!8},
  font=\scriptsize
]
  \node (koika) [color3] {Parametric \\ \noc design};

  \node (koika_prog) [right=0.5cm of koika] {\koika \\ program};
  \node (verilog_prog) [right=3cm of koika_prog] {Verilog \\ program};
  \draw[->] (koika_prog) -- (verilog_prog)
  node (sem_pre) [color2, midway, above] {Semantic preservation}
  node (compile) [midway, below] {compile};

  \node (koika_sem)   [color2, above=0.1cm of koika_prog] {\koika \\ semantics};
  \node (verilog_sem) [color2] at (koika_sem -| verilog_prog) {Verilog \\ semantics};

  \draw[->, dashed, color2] (koika_sem.east) -| ([xshift=-0.1cm]sem_pre.north);
  \draw[->, dashed, color2] (verilog_sem.west) -| ([xshift=0.1cm]sem_pre.north);

  \node[fit=(koika_prog) (verilog_prog) (compile)] (comp) {};

  \node (koika_ocaml) [color3, below=1.3cm of koika] {Parametric \\ \noc library};
  \node[above] (cuttlec) at (comp |- koika_ocaml) { \cuttlec };

  \node (verilog_desc) at (koika_ocaml -| verilog_prog) {Verilog \\ description};

  \node[color3] (prop) at (koika |- koika_sem) {Properties};

  \draw[->, dashed, color3] (koika_sem.west) -- (prop.east);

  \draw[->] (koika_ocaml) -- (verilog_desc);

  \draw[<-, decoration={snake, amplitude=0.5mm,pre length=0.5mm}, decorate] (koika_ocaml.north) -- (koika.south);
  \coordinate (extr1) at ($(koika_ocaml.north)!0.5!(koika.south)$);
  \draw[<-, decoration={snake, amplitude=0.5mm,pre length=0.5mm}, decorate] (cuttlec.north) -- (comp.south);
  \coordinate (extr2) at ($(cuttlec.north)!0.5!(comp.south)$);

  \node[gray, rotate=90, above, colornew, inner sep=1pt] at ([xshift=-1mm]extr1) {extract};
  \node[gray, rotate=90, above, fill=white, inner sep=1pt] at ([xshift=-1mm]extr2) {extract};

  \node (dashy) at (extr1) {};

  \draw[dashed, gray] ($(koika |- dashy) - (0.3cm,0)$) -- ($(verilog_desc |- dashy) + (0.5cm,0)$);
  \draw[dashed, gray] ({$(cuttlec)!0.5!(verilog_desc)$} |- dashy) -- ($(cuttlec)!0.5!(verilog_desc) - (0,0.5cm)$);

  \node[gray, below left] (ocaml) at ({$(cuttlec)!0.5!(verilog_desc)$} |- dashy) {\ocaml};
  \node[gray, above] (coq) at ({$(cuttlec)!0.5!(verilog_desc)$} |- dashy) {\rocq};
  \node[gray, below right] (verilog) at ({$(cuttlec)!0.5!(verilog_desc)$} |- dashy) {Verilog};




  \node (koika_sem2)   [color3, above=0.1cm of koika_sem] {Alternative semantics};
  \node (pl)           [color3, left=2cm of prop] {Program logic};
  \node (spec)         [color3] at (pl |- koika) {\noc \\ Specification};
  \node (config)       [] at (spec |- extr1) {\noc \\ configuration};

  \draw[->, dashed, color3] (prop.west) -- (pl.east);
  \draw[->, dashed, color3] (koika_sem2.west) -| (pl.north);
  \draw[->, dashed, color3] (koika.west) -- (spec.east)
  node (re) [midway, above] {Refinement};

  \draw[->] (config.east) -- ([xshift=-4mm]extr1);

  \begin{pgfonlayer}{background}
    \draw[rounded corners,colornew,draw=none] (spec.south west) |- (koika.north east) -- (koika_ocaml.south east) -| (koika.south west) -- cycle;
    \node[text=gray] [left= 0cm of koika_ocaml] {\sys \\ (\Cref{sec:noc})};
    \draw[rounded corners,colornew,draw=none] (pl.south west) |- (koika_sem2.north east) |- (koika_sem.north west) |- cycle;
    \node[text=gray] [above left= 0cm and 0cm of koika_sem2] {\koika extensions (\Cref{sec:extensions})};
  \end{pgfonlayer}
\end{tikzpicture}
  \caption{
    The \koika compile(r) and the \koika programs are
    extracted from Coq into OCaml.
    We highlight \textcolor{blue}{the contributions of
      this paper} and \textcolor{green}{the existing formal
      definitions of \koika}.
  }
  \label{fig:overview}
  \vspace*{-0.2cm}
\end{figure}

We use \koika which has an accompanying compiler with a proof that
it preserves the semantics in the translation to Verilog.
\Cref{fig:overview} presents an overview of \koika and the
twofold contributions of this paper.
We first extend \koika to then define \sys.
Specifically, we contribute:
\begin{description}
   \item[Read-After-Write (RAW) semantics for \koika]
        In the design of \sys, we noticed that the semantics of
        \koika are too strict when it comes to concurrency.
        This prevents an efficient \noc design (\Cref{sec:challenge2}).
        Hence, we implemented Read-After-Write semantics
        to increase concurrent register access.
        To show that this does not jeopardize \koika{}'s data race freedom,
        we reestablished all the guarantees of \koika including compiler correctness.
        Our RAW semantics apply to hardware designs beyond \noc{}s.
  \item[A program logic for \koika]
        We created a program logic for \koika to write specifications
        as Hoare triples and reason about \koika programs in a modular
        fashion with support for proof automation. 
        To the best of our knowledge, this is the first program logic
        for a rule-based hardware description language.
 \item[A parametric \noc library]
        These two extensions to \koika enabled the design of \sys.
        \sys includes a formal specification of $k$-dimensional on chip
        communication networks.
        In order to prove refinement only once for all $k$-dimensional
        \noc{}s, the implementation in \koika is non-trivial.
        It is parametric in the number of \noc dimensions and their respective
        sizes.
        We are unaware of such a highly parameterized hardware design that
        is formal yet practical.
\end{description}
The rest of the paper starts with an introduction into the design
of our \noc library show-casing how it goes beyond what is possible
in hardware design languages such as SystemVerilog.
In \cref{sec:problem}, we explain the challenges that we faced
in the design of \noc{s} in \koika.
Afterwards, \cref{sec:extensions} defines our extensions to
\koika that enable the design of \sys in \cref{sec:noc}.
To study the practicality of \sys, we synthesize several \noc{s}
in \cref{sec:eval}.
We find that \textbf{practical and scalable \noc
design in rule-based hardware requires strong reasoning, as in
\sys, to lift data-race-freedom guarantees from runtime to
design time}.
Finally, we review related work and point to future directions
in \cref{sec:related}.
Our \emph{admit-free} \rocq development for \sys is available as an artifact
of this paper.

\section{A library for multi-dimensional NoC design}
\label{sec:libdesign}

Support for higher-order programming in hardware
design languages is limited which makes the design
of a library for \noc{s} challenging.
At the same time, higher-order reasoning does not exist
and hence it is impossible to formally-verify such a
library itself.
The most popular language for hardware design,
SystemVerilog, supports a C-like macro system and
a more principled approach to generate code based
on module parameters and \texttt{generate} blocks.
Module parameters are wires, i.e.,
\texttt{logic} in SystemVerilog terminology.
For a library that generates \noc{s} of size $N\!\times\!N$,
we would respectively define a
module with \texttt{parameter} \texttt{N}.
\begin{lstlisting}[style=verilog, style=frame]
module noc
  #(parameter int §4:WIDTH§ = §32§, §4:N§)
   (input logic clk, rst);
\end{lstlisting}

Additionally, we parameterize the module over the
\texttt{WIDTH} of the data to be routed.
In the \texttt{noc} module, we then define the
input and output interfaces for all routers and the
channels.
\begin{lstlisting}[style=verilog, style=frame]
  §1:channel_if§ #(WIDTH) channels [§4§][§4:N§ * §4:N§] ();
  §1:channel_if§ #(WIDTH) src_if [§4:N§ * §4:N§] ();
  §1:channel_if§ #(WIDTH) dst_if [§4:N§ * §4:N§] ();
\end{lstlisting}

The \texttt{channel\_if} is the data type for data transfer
both inside the network and outside the network, i.e., data
enters the \noc via a source interface (\texttt{src\_if})
and leaves it via a destination interface (\texttt{dst\_if}).
A \texttt{generate} block would generate the $N\!\times\!N$
routers and connect them respectively.
\begin{lstlisting}[style=verilog, style=frame]
  generate
    for (genvar i = §0§; i < §4:N§ * §4:N§; i++) begin
      router #(.§4:WIDTH§, .§4:N§, .X(i / §4:N§), .Y(i % §4:N§))
        rtr (/* router connections */);
    end
  endgenerate
endmodule
\end{lstlisting}

We omit the complicated code for setting up the proper
router connections.
The router also needs to be parameterized
by its index (\texttt{X:Y}), which is necessary
for routing decisions.
\begin{lstlisting}[style=verilog, style=frame]
module router
  #(parameter int §4:WIDTH§ = §32§, §4:N§, §4:X§, §4:Y§)
   (input logic clk, rst,
    §1:channel_if.in§ dir_in [§4§],  src_if,
    §1:channel_if.out§ dir_out [§4§], dst_if);
  <|\vspace{-5pt}|>
  always_ff <|@|>(posedge clk) begin
    if (rst) begin /* reset logic */ end
    else begin /* routing / arbitration / switching */ end
  end
endmodule
\end{lstlisting}

Once more, we omit the complicated code that generalizes
the router design, i.e., routing, arbitration and switching,
over \texttt{N}, \texttt{X} and \texttt{Y}. We just assume
that it should be feasible to implement these parts in a
generic fashion.
Instead, we focus on stating properties
about the \noc as SystemVerilogAssertions (SVA),
an extension of SystemVerilog with a form of temporal logic~\cite{sva}.
An interesting liveness property would generalize
over the state of the whole \noc, the data that is
routed and the packets that are injected into
and retrieved from the \noc.
\begin{theorem}[Single \noc liveness]\label{theo:noc:liveness:base}
  Any message $m$ that managed to enter into the \noc
  at any router $i$ will \textbf{eventually} be delivered to its destination.
\end{theorem}

Since SVA does not support quantification, this liveness property
must be constructed from smaller properties,
one for each of the $N\!\times\!N$ routers.
\begin{lstlisting}[style=verilog, style=frame]
generate
  for (genvar i = 0; i < §4:N§ * §4:N§; i++) begin
    assert property (liveness #(i));
  end
endgenerate
\end{lstlisting}

For every router, we assert that if a packet entered
at this router then it will eventually arrive at its
destination and the transmitted \texttt{data} did not change.
\begin{lstlisting}[style=verilog,style=frame]
property liveness #(int §4:SRC§);
  logic [WIDTH-1:0] data;
  logic [ADDR_W-1:0] dest;
  <|\vspace{-5pt}|>
  (§3:src_if§[§4:SRC§].valid, dest = §3:src_if§[§4:SRC§].dest, data = §3:src_if§[§4:SRC§].data)
  <||-> \#\#[0:\$]|>
  (§3:dst_if§[dest].valid &&  §3:dst_if§[dest].data == data);
endproperty
\end{lstlisting}

The connective \pre{$A$} \texttt{|-> \#\#[0:\$]} \post{$B$} denotes that
if assertion \pre{$A$} holds in the current clock cycle
then \post{$B$} must hold after an arbitrary
number (\texttt{\$}) of clock cycles%
\footnote{A hardware design is essentially an infinite loop
where each loop iteration is referred to as a \emph{clock cycle}.}.

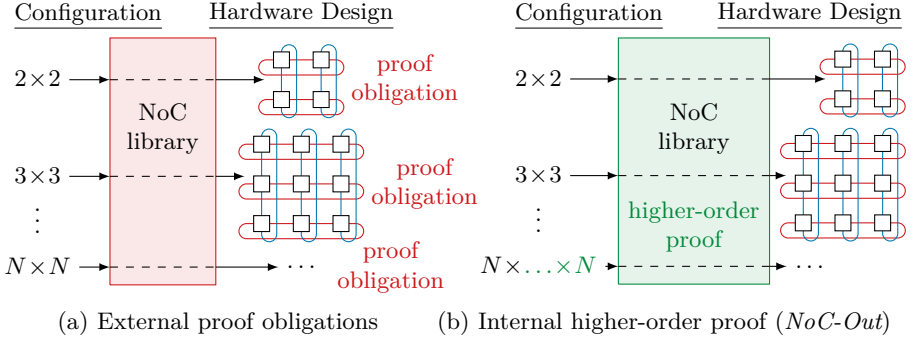
\begin{figure}
  \centering
  \begin{subfigure}[t]{0.53\textwidth}
    \centering
    \vspace{0pt}
    \newcommand{\spcr}{0.2cm}
\newcommand{\spcc}{0.2cm}
\begin{tikzpicture}[
  font=\footnotesize,
  node distance=0.01\textwidth
  ]

   \node[] (c) at (0,0) { \underline{Configuration} };
   \node[] (hd) [right = 0.4cm of c] { \underline{Hardware Design} };

   \matrix (2x2net) [
    column sep=0.3cm,
    row sep=0.3cm,
    matrix of nodes,
    nodes in empty cells,
    nodes={outer sep=0pt,draw,fill=white},
    ]
    [below = of hd]
  {
     & \\
     & \\
  };

  \begin{pgfonlayer}{background}
  \begin{scope}[rounded corners=0.1cm, red]
    \draw (2x2net-1-1) -- (2x2net-1-2);
    \draw (2x2net-2-1) -- (2x2net-2-2);
    \draw (2x2net-1-2) -| ($(2x2net-1-2) + (1.5*\spcr,-\spcr)$) -- ($(2x2net-1-1) + (-1.5*\spcr,-\spcr)$) |- (2x2net-1-1);
    \draw (2x2net-2-2) -| ($(2x2net-2-2) + (1.5*\spcr,-\spcr)$) -- ($(2x2net-2-1) + (-1.5*\spcr,-\spcr)$) |- (2x2net-2-1);
  \end{scope}
  \end{pgfonlayer}

  \begin{pgfonlayer}{background}
  \begin{scope}[rounded corners=0.1cm, blue]
    \draw (2x2net-1-1) -- (2x2net-2-1);
    \draw (2x2net-1-2) -- (2x2net-2-2);
    \draw (2x2net-1-1) |- ($(2x2net-1-1) + (\spcc,\spcc)$) -- ($(2x2net-2-1) + (\spcc,-1.2*\spcc)$) -| (2x2net-2-1);
    \draw (2x2net-1-2) |- ($(2x2net-1-2) + (\spcc,\spcc)$) -- ($(2x2net-2-2) + (\spcc,-1.2*\spcc)$) -| (2x2net-2-2);
  \end{scope}
  \end{pgfonlayer}

  \matrix (3x3net) [
    column sep=0.3cm,
    row sep=0.3cm,
    matrix of nodes,
    nodes in empty cells,
    nodes={outer sep=0pt,draw,fill=white},
    ]
    [below =4pt of 2x2net]
  {
     & & \\
     & & \\
     & & \\
  };

  \begin{pgfonlayer}{background}
  \begin{scope}[rounded corners=0.1cm, red]
    \draw (3x3net-1-1) -- (3x3net-1-2);
    \draw (3x3net-2-1) -- (3x3net-2-2);
    \draw (3x3net-3-1) -- (3x3net-3-2);
    \draw (3x3net-1-2) -- (3x3net-1-3);
    \draw (3x3net-2-2) -- (3x3net-2-3);
    \draw (3x3net-3-2) -- (3x3net-3-3);
    \draw (3x3net-1-3) -| ($(3x3net-1-3) + (1.5*\spcr,-\spcr)$) -- ($(3x3net-1-1) + (-1.5*\spcr,-\spcr)$) |- (3x3net-1-1);
    \draw (3x3net-2-3) -| ($(3x3net-2-3) + (1.5*\spcr,-\spcr)$) -- ($(3x3net-2-1) + (-1.5*\spcr,-\spcr)$) |- (3x3net-2-1);
    \draw (3x3net-3-3) -| ($(3x3net-3-3) + (1.5*\spcr,-\spcr)$) -- ($(3x3net-3-1) + (-1.5*\spcr,-\spcr)$) |- (3x3net-3-1);
  \end{scope}

  \begin{scope}[rounded corners=0.1cm, blue]
    \draw (3x3net-1-1) -- (3x3net-2-1);
    \draw (3x3net-1-2) -- (3x3net-2-2);
    \draw (3x3net-1-3) -- (3x3net-2-3);
    \draw (3x3net-2-1) -- (3x3net-3-1);
    \draw (3x3net-2-2) -- (3x3net-3-2);
    \draw (3x3net-2-3) -- (3x3net-3-3);
     \draw (3x3net-1-1) |- ($(3x3net-1-1) + (\spcc,\spcc)$) -- ($(3x3net-3-1) + (\spcc,-1.2*\spcc)$) -| (3x3net-3-1);
    \draw (3x3net-1-2) |- ($(3x3net-1-2) + (\spcc,\spcc)$) -- ($(3x3net-3-2) + (\spcc,-1.2*\spcc)$) -| (3x3net-3-2);
    \draw (3x3net-1-3) |- ($(3x3net-1-3) + (\spcc,\spcc)$) -- ($(3x3net-3-3) + (\spcc,-1.2*\spcc)$) -| (3x3net-3-3);
  \end{scope}
  \end{pgfonlayer}

  \node[anchor=west] (c1) at (c.west |- 2x2net) {$2\!\times\!2$};
  \coordinate (3x3net_) at ($(3x3net.west)+(0,0.1cm)$);

  \node[anchor=west] (c2) at (c.west |- 3x3net_) {$3\!\times\!3$};
  \node[] (cn) at ($(c2) + (0,-1.2cm)$) {$N\!\times\!N$};
  \node[] (dots) at ($(c2)!0.4!(cn)$) {$\vdots$};

  \node[] (nxnnet) at (hd.south |- cn) {$\ldots$};

  \coordinate (boxl) at ($(c.south east) - (0.8cm,0)$);
  \coordinate (boxr) at ($(hd.south west) + (0.2cm,0)$);

  \draw[fill=red!10, draw=red] (boxl |- cn.south east) |- (boxr) |- (boxl |- cn.south east);

  \node[align=center] (lib) at ($(boxl)!0.5!(boxr) + (0,-1.2cm)$) {\noc\\library};

  \draw[-latex] (c1.east) -- (boxl |- c1.east);
  \draw[dashed] (boxl |- c1.east) -- (boxr |- 2x2net.east);
  \draw[-latex] (boxr |- 2x2net.east) -- (2x2net.west);

  \draw[-latex] (c2.east) -- (boxl |- c2.east);
  \draw[dashed] (boxl |- c2.east) -- (boxr |- 3x3net_);
  \draw[-latex] (boxr |- 3x3net_) -- (3x3net_);

  \draw[-latex] (cn.east) -- (boxl |- cn.east);
  \draw[dashed] (boxl |- cn.east) -- (boxr |- nxnnet.east);
  \draw[-latex] (boxr |- nxnnet.east) -- (nxnnet.west);

  \node[align=center] (po1) [right=of 2x2net] {\textcolor{red}{proof} \\ \textcolor{red}{obligation}};
  \node[align=center] (po2) [right=of 3x3net] {\textcolor{red}{proof} \\ \textcolor{red}{obligation}};
  \node[align=center] (po3) [right=of nxnnet] {\textcolor{red}{proof} \\ \textcolor{red}{obligation}};

 \end{tikzpicture}
  \end{subfigure}
  \begin{subfigure}[t]{0.45\textwidth}
    \centering
    \vspace{0pt}
    \hspace*{-0.5cm}
    \newcommand{\spcr}{0.2cm}
\newcommand{\spcc}{0.2cm}
\begin{tikzpicture}[
  font=\footnotesize,
  node distance=0.01\textwidth
  ]

  \node[] (c) at (0,0) { \underline{Configuration} };
  \node[] (hd) [right = 0.5cm of c] { \underline{Hardware Design} };

  \matrix (2x2net) [
    column sep=0.3cm,
    row sep=0.3cm,
    matrix of nodes,
    nodes in empty cells,
    nodes={outer sep=0pt,draw,fill=white},
  ] [xshift=0.7cm, below = of hd]
  {
     & \\
     & \\
  };

  \begin{pgfonlayer}{background}
  \begin{scope}[rounded corners=0.1cm, red]
    \draw (2x2net-1-1) -- (2x2net-1-2);
    \draw (2x2net-2-1) -- (2x2net-2-2);
    \draw (2x2net-1-2) -| ($(2x2net-1-2) + (1.5*\spcr,-\spcr)$) -- ($(2x2net-1-1) + (-1.5*\spcr,-\spcr)$) |- (2x2net-1-1);
    \draw (2x2net-2-2) -| ($(2x2net-2-2) + (1.5*\spcr,-\spcr)$) -- ($(2x2net-2-1) + (-1.5*\spcr,-\spcr)$) |- (2x2net-2-1);
  \end{scope}

  \begin{scope}[rounded corners=0.1cm, blue]
    \draw (2x2net-1-1) -- (2x2net-2-1);
    \draw (2x2net-1-2) -- (2x2net-2-2);
    \draw (2x2net-1-1) |- ($(2x2net-1-1) + (\spcc,\spcc)$) -- ($(2x2net-2-1) + (\spcc,-1.2*\spcc)$) -| (2x2net-2-1);
    \draw (2x2net-1-2) |- ($(2x2net-1-2) + (\spcc,\spcc)$) -- ($(2x2net-2-2) + (\spcc,-1.2*\spcc)$) -| (2x2net-2-2);
  \end{scope}
  \end{pgfonlayer}

  \matrix (3x3net) [
    column sep=0.3cm,
    row sep=0.3cm,
    matrix of nodes,
    nodes in empty cells,
    nodes={outer sep=0pt,draw,fill=white},
    ] [xshift=-0.26cm, below =4pt of 2x2net]
  {
     & & \\
     & & \\
     & & \\
  };

  \begin{pgfonlayer}{background}
  \begin{scope}[rounded corners=0.1cm, red]
    \draw (3x3net-1-1) -- (3x3net-1-2);
    \draw (3x3net-2-1) -- (3x3net-2-2);
    \draw (3x3net-3-1) -- (3x3net-3-2);
    \draw (3x3net-1-2) -- (3x3net-1-3);
    \draw (3x3net-2-2) -- (3x3net-2-3);
    \draw (3x3net-3-2) -- (3x3net-3-3);
    \draw (3x3net-1-3) -| ($(3x3net-1-3) + (1.5*\spcr,-\spcr)$) -- ($(3x3net-1-1) + (-1.5*\spcr,-\spcr)$) |- (3x3net-1-1);
    \draw (3x3net-2-3) -| ($(3x3net-2-3) + (1.5*\spcr,-\spcr)$) -- ($(3x3net-2-1) + (-1.5*\spcr,-\spcr)$) |- (3x3net-2-1);
    \draw (3x3net-3-3) -| ($(3x3net-3-3) + (1.5*\spcr,-\spcr)$) -- ($(3x3net-3-1) + (-1.5*\spcr,-\spcr)$) |- (3x3net-3-1);
  \end{scope}

  \begin{scope}[rounded corners=0.1cm, blue]
    \draw (3x3net-1-1) -- (3x3net-2-1);
    \draw (3x3net-1-2) -- (3x3net-2-2);
    \draw (3x3net-1-3) -- (3x3net-2-3);
    \draw (3x3net-2-1) -- (3x3net-3-1);
    \draw (3x3net-2-2) -- (3x3net-3-2);
    \draw (3x3net-2-3) -- (3x3net-3-3);
     \draw (3x3net-1-1) |- ($(3x3net-1-1) + (\spcc,\spcc)$) -- ($(3x3net-3-1) + (\spcc,-1.2*\spcc)$) -| (3x3net-3-1);
    \draw (3x3net-1-2) |- ($(3x3net-1-2) + (\spcc,\spcc)$) -- ($(3x3net-3-2) + (\spcc,-1.2*\spcc)$) -| (3x3net-3-2);
    \draw (3x3net-1-3) |- ($(3x3net-1-3) + (\spcc,\spcc)$) -- ($(3x3net-3-3) + (\spcc,-1.2*\spcc)$) -| (3x3net-3-3);
  \end{scope}
  \end{pgfonlayer}

  \node[anchor=west] (c1) at (c.west |- 2x2net) {$2\!\times\!2$};
  \coordinate (3x3net_) at ($(3x3net.west)+(0,0.1cm)$);

  \node[anchor=west] (c2) at (c.west |- 3x3net_) {$3\!\times\!3$};
  \node[] (cn) at ($(c2) + (0,-1.2cm)$) {$N\!\times\! \textcolor{green}{\ldots \!\times\! N}$};
  \node[] (dots) at ($(c2)!0.4!(cn)$) {$\vdots$};

  \node[] (nxnnet) at (hd.south |- cn) {$\ldots$};

  \coordinate (boxl) at ($(c.south east) - (0.7cm,0)$);
  \coordinate (boxr) at ($(hd.south west) + (0.8cm,0)$);

  \draw[fill=green!10, draw=green] (boxl |- cn.south east) |- (boxr) |- (boxl |- cn.south east);

  \node[align=center] (lib) at ($(boxl)!0.5!(boxr) + (0,-1.2cm)$) {\noc\\library};

  \node[align=center] (po) [below = 0.4cm of lib] {\textcolor{green}{higher-order} \\ \textcolor{green}{proof}};

  \draw[-latex] (c1.east) -- (boxl |- c1.east);
  \draw[dashed] (boxl |- c1.east) -- (boxr |- 2x2net.east);
  \draw[-latex] (boxr |- 2x2net.east) -- (2x2net.west);

  \draw[-latex] (c2.east) -- (boxl |- c2.east);
  \draw[dashed] (boxl |- c2.east) -- (boxr |- 3x3net_);
  \draw[-latex] (boxr |- 3x3net_) -- (3x3net_);

  \draw[-latex] (cn.east) -- (boxl |- cn.east);
  \draw[dashed] (boxl |- cn.east) -- (boxr |- nxnnet.east);
  \draw[-latex] (boxr |- nxnnet.east) -- (nxnnet.west);

 \end{tikzpicture}
  \end{subfigure}
\\
  \begin{subfigure}[t]{0.45\textwidth}
    \vspace{-\baselineskip}
    \caption{External proof obligations}
    \label{fig:lib:verilog}
  \end{subfigure}
  \begin{subfigure}[t]{0.50\textwidth}
    \vspace{-\baselineskip}
    \caption{Internal higher-order proof (\emph{\sys})}
    \label{fig:lib:ours}
  \end{subfigure}

  \caption{Two different \noc library designs:
    In languages such as SystemVerilog, it is possible to emit
    proof obligations with every instantiated \noc, but the state
    space quickly explodes making them unprovable for model checkers.
    In \sys, we conduct a single higher-order
    proof about all the \noc{s} that the library generates.
    No additional verification is necessary anymore while \sys scales
    to \noc{s} with arbitrary dimensions.
  }
\end{figure}

\paragraph{Implications}
At this point, we generate a proof obligation for every instantiated
\noc as visualized in \cref{fig:lib:verilog}.
Our stated property is valid with respect to the SVA standard, but
it is unclear whether it is supported by the various proprietary
formal tools.
Despite that, the proof effort is substantial, even for
small instances of the \noc.
Formal tools for SVA usually use a combination of model checking and
SAT solving and as such depend heavily on the state space of the design.
Assume a benevolent approximation, omitting all the additional complexity
of a \noc, that wants to fulfill the above liveness property (which
we implement and verify in the rest of the paper) and assumes a
data \texttt{WIDTH} of 1 bit.
Then even a small $2\!\times\!2$ \noc already contains $2^{36}$ states.
Recent work in modelling such \noc{s} in a model checker reports
verification times ranging from hours to days~\cite{9858921}.
Even worse, their property is weaker than the liveness property that we
stated above. Specifically, they only used a local liveness property over each individual router.
That is, the \texttt{liveness} property -- even for this small $2\!\times\!2$
instance of our \texttt{noc} module -- would be unbearable to check.

\paragraph{Library reasoning}
At this point, we would like to hide all these proof obligations
inside the library as visualized in \cref{fig:lib:ours}.
That is, we would like to reason about all the \noc{s} that
the library could ever generate.
\begin{theorem}[\noc liveness]\label{theo:noc:liveness:lib}
  For every generated \noc, \cref{theo:noc:liveness:base} holds.
\end{theorem}
If we could do this then there would be no further proof obligation
for the generated \noc{s}.
A -- hypothetical -- property in SVA would therefore have to reason about
the \texttt{noc} module itself.
\begin{lstlisting}[style=verilog,style=frame]
property liveness_lib;
  <|$\forall$|> §4:N§, noc #(.§4:N§).§3:liveness§
endproperty
\end{lstlisting}

But this is impossible to state in SVA.
First, SVA does not have a notion of universal quantification ($\forall$).
Second, the \texttt{N} which is quantified over is a module parameter that
gets elaborated even before any formal framework sees it.
That is, it is impossible to even state such a property let alone
reasoning about it.
And as a consequence, the expressivity of SystemVerilog is clearly insufficient
for the goals of this paper.
Instead, we decided to use \koika, a language for hardware design
embedded in the Rocq prover, to construct our \noc generator along
with the higher-order proof for the liveness property
(\Cref{theo:noc:liveness:lib}).
Beyond that, \sys does not only allow for 2-dimensional \noc designs but
generalizes to arbitrary-sized $k$-dimensional tori.

\section{Hardware design in \koika}
\label{sec:problem}

To introduce the challenges and shortcomings of \koika that
we encountered in the design of \sys,
we start from a simple \noc and gradually add complexity.
To introduce \koika's programming model, we define a simple
3-staged \noc pipeline with unidirectional messaging.
This pipeline suffices to show that stating properties about
the execution is neither modular nor idiomatic and makes proof
challenging.
Afterwards, we add bidirectional messaging, i.e., messages may also
travel the dimensions in reverse.
This increases concurrency inside the routers and shows that
\koika's semantics would enforce inefficient \noc designs.
The \sys library in \cref{sec:noc} then generalizes
over many aspects that are left specific in this section.

\subsection{A 3-stage pipeline \noc}
\begin{figure}
  \centering
  \begin{tikzpicture}[
  node distance=0.3cm,
  font=\footnotesize,
]
  \node[draw] (r0) at (0,0) { \ilstRocq{router}$_0$ };
  \node[draw] (r1) [right = 0.15\textwidth of r0] { \ilstRocq{router}$_1$ };
  \node[draw] (r2) [right = 0.15\textwidth of r1] { \ilstRocq{router}$_2$ };

  \draw[-latex] (r0.east) --node[midway,fill=white]{\ilstRocq{up}$_0$} (r1.west);
  \draw[-latex] (r1.east) --node[midway,fill=white]{\ilstRocq{up}$_1$} (r2.west);

  \node[draw,rounded corners] (e1) [left = 0.05\textwidth of r0] { \ilstRocq{noc_enter} };
  \node[draw,rounded corners] (e2) [right = 0.05\textwidth of r2] { \ilstRocq{noc_exit} };

  \draw[densely dotted,->] (e1.east) -- (r0.west);
  \draw[densely dotted,->] (r2.east) -- (e2.west);

\end{tikzpicture}
  \caption{The 3-stage \noc pipeline.}
  \label{fig:router:noc}
\end{figure}
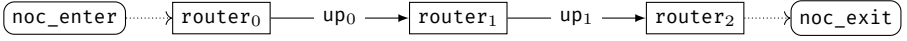
\begin{figure}
%
%
\begin{subfigure}[b]{0.2\textwidth}
\begin{lstlisting}
Variant reg_t :=

| up<|$_0$|>
| up<|$_1$|>
.
\end{lstlisting}
\captionsetup{justification=raggedright,singlelinecheck=false}
\caption{Definition}
\label{fig:registers:def}
\end{subfigure}
\hfill
\begin{subfigure}[b]{0.32\textwidth}
  \centering
\begin{lstlisting}
Definition R reg :=
match reg with
| up<|$_0$|> <|=>|> §3:bits_t§ flit_size
| up<|$_1$|> <|=>|> §3:bits_t§ flit_size
end.
\end{lstlisting}
\captionsetup{justification=raggedright,singlelinecheck=false}
\caption{Data types}
\label{fig:registers:types}
\end{subfigure}
\hfill
\begin{subfigure}[b]{0.34\textwidth}
  \centering
\begin{lstlisting}
Definition r reg : R reg :=
match reg with
| up<|$_0$|> <|=>|> §3:empty_flit§
| up<|$_1$|> <|=>|> §3:empty_flit§
end.
\end{lstlisting}
\captionsetup{justification=raggedright,singlelinecheck=false}
\caption{Initialization}
\label{fig:registers:init}
\end{subfigure}
%
\caption{The \koika definition registers for the 2 channels in a 3-stage \noc pipeline.}
\label{fig:registers}
\end{figure}

To introduce \koika's programming model, we develop the
simple unidirectional 3-stage pipeline from \cref{fig:router:noc}.
Messages enter the \noc at the router with the lowest
dimensional index (\texttt{router}$_{0}$),
travel the dimension \texttt{up}-wards (from left to right) and
exit the \noc at the router with the highest
dimensional index (\texttt{router}$_{2}$).
To exchange messages, two routers are connected via a \emph{channel},
i.e., a queue that preserves FIFO ordering of messages.
As such, for every channel there is a router that sends/enqueues
messages and another router that receives/dequeues these messages.
For sake of simplicity, we degenerate the channels of our \noc
in this paper to queues of size 1.
We leave the extension to channels that store more than a single
message as future work.
\koika's programming model abstracts over hardware
concepts such as wires and clocks thereby making it
amenable to non-hardware engineers.
A \koika program consists of three parts:
the register definition,
the functional specification, and
an execution schedule.
\paragraph{Registers}
Hardware programs execute in cycles.
Registers are the state of a hardware program, i.e.,
a value stored to a register in one cycle is available
in the next cycle.
The channels capture the state, i.e., the in-flight messages,
of our \noc and as such need to be defined as registers.
\Cref{fig:registers} presents the three steps to
define the set of registers (\ref{fig:registers:def}),
specify their types (\ref{fig:registers:types}) and
initialize them (\ref{fig:registers:init}).
A channel \texttt{up}$_{n}$ sends a message
from a router at index $n$
to a router at index $n+1$.
Similarly, in a bidirectional \noc,
a channel \texttt{down}$_{n}$ sends a message
from a router at index $n$
to a router at index $n-1$.
The type of the channels is \texttt{bits\_t flit\_size}, i.e.,
a bit vector of the size of a single message,
also called a \emph{flit}.
As such, our channels can store only a single message.
Since every hardware state needs to be initialized, we
start from an \texttt{empty\_flit}, which zeroes all bits.
\paragraph{Actions}
Expressions in \koika are called \emph{actions}.
An action may perform operations on bit vectors,
incorporate condition control flow,
call external functions, and
read from or write to registers.
That is, actions can only communicate with other
actions via registers, e.g., our routers communicate
via channel registers.
For our \noc, we implement three actions, one for
each variant of the routers in the pipeline.
In \texttt{router}$_0$ (\ref{fig:router:left}), a
message enters via the (external) interface of the \noc (Line~2)
and is sent \texttt{up}wards in the dimension (Line~3).
The intermediate \texttt{router}$_1$ (\ref{fig:router:mid})
retrieves the message from channel \texttt{up}$_0$ (Line~2) and
pipes it \texttt{up}wards via channel \texttt{up}$_1$ (Line~3).
Lastly, \texttt{router}$_2$ (\ref{fig:router:right}) retrieves
the message from channel \texttt{up}$_1$ (Line~2) and emits it via
the (external) \noc interface.
\paragraph{Schedule}
To define a \koika program, we register the three actions
as \emph{rules} in a schedule (see \cref{fig:scheduler}).
The schedule represents one hardware execution
cycle where every rule, i.e., every action, runs exactly once.%
\footnote{From this point on, we will use ``rule'' to mean a scheduled action.}

\begin{figure}
\begin{subfigure}[b]{0.30\textwidth}
\begin{lstlisting}
Definition router<|$_0$|> := <{
  let m = §3:noc_enter§()
  in write<|$_0$|>(up<|$_0$|>, m)
}>.
\end{lstlisting}
    \captionsetup{justification=raggedright,singlelinecheck=false}
    \caption{Leftmost router}
    \label{fig:router:left}
  \end{subfigure}
  \hfill
  \begin{subfigure}[b]{0.30\textwidth}
\begin{lstlisting}
Definition router<|$_1$|> := <{
  let m = read<|$_0$|>(up<|$_0$|>)
  in write<|$_0$|>(up<|$_1$|>, m)
}>.
\end{lstlisting}
    \captionsetup{justification=raggedright,singlelinecheck=false}
    \caption{Intermediate router}
    \label{fig:router:mid}
 \end{subfigure}
  \hfill
  \begin{subfigure}[b]{0.30\textwidth}
\begin{lstlisting}
Definition router<|$_2$|> := <{
  let m = read<|$_0$|>(up<|$_1$|>)
  in §3:noc_exit§(m)
}>.
\end{lstlisting}
    \captionsetup{justification=raggedright,singlelinecheck=false}
    \caption{Rightmost router}
    \label{fig:router:right}
  \end{subfigure}
  \\
\begin{subfigure}{1\textwidth}
  \centering
  \vspace{7pt}
\begin{minipage}{0.76\linewidth}
\begin{lstlisting}[escapeinside={§<}{>§}]
Definition schedule := §3:router$_0$§ §<|>>§ §3:router$_1$§ §<|>>§ §3:router$_2$§ §<|>>§ §4:done§.
\end{lstlisting}
\end{minipage}
\caption{Scheduling the rules, i.e., the different routers, in the \noc.}
\label{fig:scheduler}
\end{subfigure}
  \caption{
    The three router
    implementations for our simple unidirectional
    \noc pipeline.}
  \label{fig:routers}
\end{figure}

\subsection{Parallel execution but sequential reasoning}

\koika executes actions in parallel
but reasoning about them remains sequential.
All actions of a schedule execute in parallel and
thus \koika needs to prevent data races, i.e., two
actions that access the same register.
In order to resolve such concurrent accesses,
the execution is transactional.
That is, per-register-logs record reads and writes.
The final value of a register is only \emph{committed}
at the end of the cycle when all actions were executed.
If a data race occurred then the first action in the schedule
wins and the succeeding conflicting actions are aborted, i.e.,
\koika discards their effects to the registers.
As such, reasoning about a program either during the implementation
or in a proof is sequential: one rule at a time.
For a single action, it follows the sequential control flow
of its definition and at the level of a whole cycle it
follows the order stated in the schedule.

\subsection{Challenge 1: Lack of modular reasoning}
\label{sec:challenge1}
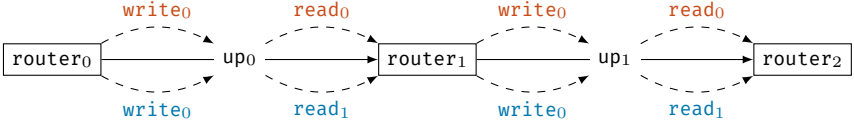
\begin{figure}
  \centering
  \begin{tikzpicture}[
  node distance=0.3cm,
  font=\footnotesize,
  color1/.style={orange},
  color2/.style={blue},
]
  \node[draw] (r0) at (0,0) { \lstinline{router}$_0$ };
  \node[draw] (r1) [right = 0.3\textwidth of r0] { \lstinline{router}$_1$ };
  \node[draw] (r2) [right = 0.3\textwidth of r1] { \lstinline{router}$_2$ };

  \draw[-latex] (r0.east) -- (r1.west)
  node (reg0) [midway,fill=white]{\lstinline{up}$_0$};
  \draw[-latex] (r1.east) -- (r2.west)
  node (reg1) [midway,fill=white]{\lstinline{up}$_1$};

 \draw (r0.north east)
  edge [dashed,bend left,-latex]
  node[midway, above, color1]{\texttt{write$_{0}$}}
  (reg0.north west)
  ;
 \draw (reg0.north east)
  edge [dashed,bend left,-latex]
  node[midway, above, color1]{\texttt{read$_{0}$}}
  (r1.north west)
  ;
 \draw (r1.north east)
  edge [dashed,bend left,-latex]
  node[midway, above,color1]{\texttt{write$_{0}$}}
  (reg1.north west)
  ;
 \draw (reg1.north east)
  edge [dashed,bend left,-latex]
  node[midway, above, color1]{\texttt{read$_{0}$}}
  (r2.north west)
  ;

 \draw (r0.south east)
  edge [dashed,bend right,-latex]
  node[midway, below, color2]{\texttt{write$_{0}$}}
  (reg0.south west)
  ;
 \draw (reg0.south east)
  edge [dashed,bend right,-latex]
  node[midway, below, color2]{\texttt{read$_{1}$}}
  (r1.south west)
  ;
 \draw (r1.south east)
  edge [dashed,bend right,-latex]
  node[midway, below,color2]{\texttt{write$_{0}$}}
  (reg1.south west)
  ;
 \draw (reg1.south east)
  edge [dashed,bend right,-latex]
  node[midway, below, color2]{\texttt{read$_{1}$}}
  (r2.south west)
  ;

\end{tikzpicture}
  \caption{
    \koika's ephemeral history registers allow a \textcolor{orange}{throughput-} and
    a \textcolor{blue}{latency-}oriented \noc design.
  }
  \label{fig:noc:sem}
\end{figure}
\koika actions can exchange data even in the same cycle via
\emph{ephemeral history registers (EHRs)}. Such registers remember
and allow access to some of their previous values, i.e. their \emph{history}.
In \koika, a call to \kread{0}{\texttt{up}} retrieves the initial value,
i.e., the value stored in register \texttt{up} at the beginning
of the cycle.
Similarly, \kwrite{0}{\texttt{up}}{v} stores a new value $v$ in \texttt{up} but the
initial value is still retrievable (in the same action) via \kread{0}{\texttt{up}}.
To make the new value available to another action
in the same cycle, \koika provides the \kread{1}{\texttt{up}} expression.
And lastly, \kwrite{1}{\texttt{up}}{v'} writes the final value $v'$ to
register \texttt{up} in the current cycle.
As an example, we visualize the execution semantics
of our \noc in \cref{fig:noc:sem}.
Our \textcolor{orange}{current implementation of the routers} uses
solely \kread{0}{\cdot} and \kwrite{0}{\cdot}{\cdot}%
\footnote{\kwrite{0}{\cdot}{\cdot} denotes a write of some value to some register.
  \kread{0}{\cdot} is analogous.}
.
As such, a message needs two cycles to arrive at \texttt{router}$_{2}$.
An \textcolor{blue}{alternative implementation} uses \kread{1}{\cdot}
to retrieve the message written by the previous router in the same cycle.
Respectively, a message only takes a single cycle to arrive
at \texttt{router}$_{2}$.
There are pros and cons to each of the two \noc designs.
The previous design is throughput-oriented because it routes
messages in a pipeline parallel fashion, i.e., 2 messages arrive
per cycle:
one at \texttt{router}$_{1}$ and
the other at \texttt{router}$_{2}$.%
\footnote{
  This holds starting from the 2nd cycle when
  the channels are filled.
}
The new single-cycle design favors latency because the message that entered the
\noc at \texttt{router}$_{0}$ arrives
at \texttt{router}$_{2}$ in the very same cycle.
But, in the synthesized hardware, the wall-clock time of a cycle actually increases, because
the shared registers serialize the execution of the routers.
Formal reasoning about the execution of these two different
\noc designs also differs.
\koika's operational semantics are a big-step evaluation function,
which directly computes the value of a given action.
\begin{definition}[Koika Semantics for Actions]
  For every
  register environment \R,
  scheduler log $L$,
  action log $l$, and
  variable context \Ctx
  the action (i.e. expression) $a$
  big-step evaluates $\downarrow$
  to a result value $v$, a new variable assignment \Ctx' and log $l'$, denoted
  $(l,\Ctx,a) \koikasemfun (l',\Ctx',v)$, or it aborts, denoted
  $(l,\Ctx,a) \koikasemfun \texttt{abort}$.
\end{definition}
%
%
The register environment \R stores the register values at the
beginning of the cycle, while the context \Ctx maintains the
current values of local variables.
The logs track prior reads and writes of the ongoing cycle with a crucial
distinction.
The action log $l$ only captures register accesses of
the current action, while the scheduler log $L$ aggregates
the accesses of all preceding actions in the schedule.
This distinction is necessary to discard the register
accesses of the current action in case it is aborted.
Note that, the evaluation only effects the action log $l$
and the context \Ctx.
The register environment \R and scheduler log $L$
are not altered.
After every execution of an action, in case there were no
conflicts, the scheduler log
gets concatenated $L' := L \mdoubleplus l'$.
Then, at the end of the cycle, the scheduler only persists the
changes accumulated in $L'$ by committing them into the
register environment $\R' := \texttt{update} ~ \R ~ L'$.

But reasoning about an action needs to take this into
account.
Consider this Hoare-like statement with
a \pre{precondition $P$} and
a \post{postcondition $Q$}
for reasoning about the action of \texttt{router}$_{2}$:

\noindent\begin{minipage}{\linewidth}
\begin{lstlisting}
<|$\forall$|> <|\R{}|> L <|Γ|>' l' v,
  <|\pre{$P$ (update \R L)}|> -> ([],{},router<|$_{\texttt{2}}$|>) <|$\downarrow_{\R,L}$|> (l',<|Γ|>',v) -> <|\post{$Q$ (update \R (L$\mdoubleplus$l'))}|>.
\end{lstlisting}\vspace{2pt}
\end{minipage}
Here, both \pre{$P$} and \post{$Q$} reason about an
already committed register environment \R.
This is intuitive, since the state at the end of the cycle
is the only one visible to other components and thus should
provide the important properties.
However, reasoning about the register environment
would require a schedule of a single action. Because,
the \texttt{update} occurs only at the end of a cycle when
all the actions have run and not after each individual action.
Hence, this encoding only allows running a single action
per cycle.
This is clearly not what most of the \koika programs
look like and would prevent all parallel processing.
Instead, to reason about multiple actions, the \pre{pre-} and \post{post}conditions need
to take the logs into account.
For example, if \pre{$P$} would test for a specific value \texttt{v}
in register $r$ then this translates into:
\[
\R[r] = \texttt{v} ~\land~ \kwrite{0}{r}{\cdot} \notin L ~\land~ \kwrite{1}{r}{\cdot} \notin L
\]

This condition covers the cases where the register $r$ had the value \texttt{v} at the beginning
of the cycle and no prior action overrode it. In other words, having this condition is enough
to guarantee that a \kread{0}{r} or a \kread{1}{r} would return \texttt{v}. However, even this fairly
explicit assumption leaves ambiguity about whether a \kwrite{0}{r}{\cdot} succeeds or fails. In case
$\kread{1}{r} \in L$ it fails, otherwise it succeeds. As this example demonstrates, \koika's multi-log
semantics result in rather complex register state descriptions. In particular, the two logs $l$ and $L$
in combination with the EHR operations \kread{0}{\cdot}, \kread{1}{\cdot}, \kwrite{0}{\cdot}{\cdot}, and
\kwrite{1}{\cdot}{\cdot} yield a total of $2^{8} = 256$ ways to describe the state of a single register.
Although many of these characterizations are semantically equivalent, a subset differs%
\footnote{
  This subset actually contains 7 unique states and is visualized later in \cref{fig:states:original}.
}%
, giving rise to two practical difficulties.
First, selecting the right one as well as comprehending them is
cumbersome. Second, when the selected descriptions do not match syntactically across proofs,
translating them results in significant overhead.
\begin{figure}
  \centering
  \begin{tikzpicture}[
  node distance=0.3cm,
  font=\footnotesize,
]
  \node[draw] (r0) at (0,0) { \lstinline{router}$_n$ };
  \node[draw] (r1) [right = 0.5\textwidth of r0] { \lstinline{router}$_{n+1}$ };

  \draw[-latex] ([yshift=0.1cm]r0.east) -- ([yshift=0.1cm]r1.west)
  node (reg0) [midway,fill=white,yshift=0.1cm]{\lstinline{up}$_{n}$};
  \draw[latex-] ([yshift=-0.1cm]r0.east) -- ([yshift=-0.1cm]r1.west)
  node (reg1) [midway,fill=white,yshift=-0.1cm]{\lstinline{down}$_{n+1}$};

 \draw (r0.north east)
  edge [dashed,bend left=20,-latex]
  node[midway, above]{\textcolor{orange}{\texttt{write$_{0}$}}/\textcolor{blue}{\texttt{write$_{0}$}}}
  (reg0.north west)
  ;
 \draw (reg0.north east)
  edge [dashed,bend left=20,-latex]
  node[midway, above]{\textcolor{orange}{\texttt{read$_{0}$}}/\textcolor{blue}{\texttt{read$_{1}$}}}
  (r1.north west)
  ;

  \draw (r0.south east)
  edge [dashed,bend right=20,latex-]
  node[midway, below]{\textcolor{orange}{\texttt{read$_{0}$}}/\textcolor{blue}{\texttt{read$_{1}$}}}
  (reg1.south west)
  ;
 \draw (reg1.south east)
  edge [dashed,bend right=20,latex-]
  node[midway, below]{\textcolor{orange}{\texttt{write$_{0}$}}/\textcolor{blue}{\texttt{write$_{0}$}}}
  (r1.south west)
  ;

  \draw[draw=gray,-latex] ([xshift=-1cm,yshift=0.1cm]r0.west) -- ([yshift=0.1cm]r0.west);
  \draw[draw=gray,-latex] ([yshift=-0.1cm]r0.west) -- ([xshift=-1cm,yshift=-0.1cm]r0.west);

  \draw[draw=gray,-latex] ([yshift=0.1cm]r1.east) -- ([xshift=1cm,yshift=0.1cm]r1.east);
  \draw[draw=gray,-latex] ([xshift=1cm,yshift=-0.1cm]r1.east) -- ([yshift=-0.1cm]r1.east);

\end{tikzpicture}
  \caption{Two intermediate routers of a bidirectional NoC.}
  \label{fig:intermediate_routers}
\end{figure}
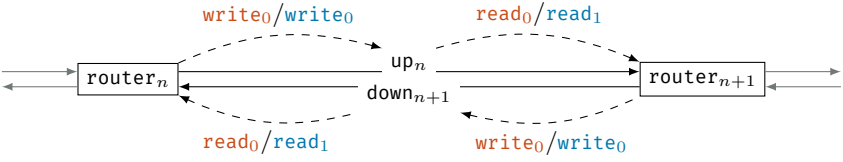
%
%

\subsection{Challenge 2: \koika's semantics foster inefficient \noc design}
\label{sec:challenge2}

Increasing the communication inside the \noc leads to more concurrency and in turn
to a trade-off between efficient execution and efficient verification.
In \cref{fig:intermediate_routers}, we integrate bidirectional messaging with
a dedicated channel for each direction and
again depict a \textcolor{orange}{throughput}- and a \textcolor{blue}{latency}-favored
\noc design.
In both cases, either of the two router actions aborts. That is because, for a given register $r$, \koika
neither allows \kread{0}{r} nor \kread{1}{r} after a \kwrite{0}{r}{\cdot} occurred in a different action.
The log $\{
\kread{0}{\texttt{down}_{n+1}}, \kwrite{0}{\texttt{up}_n}{\cdot}\} \mdoubleplus \{
\textcolor{red}{\kread{0}{\texttt{up}_n}}, \kwrite{0}{\texttt{down}_{n+1}}{\cdot}
\}$
aborts \verb|router|$_{n+1}$.
Likewise, the log $\{
\kread{0}{\texttt{up}_n}, \kwrite{0}{\texttt{down}_{n+1}}{\cdot}\} \mdoubleplus \{
\textcolor{red}{\kread{0}{\texttt{down}_{n+1}}}, \kwrite{0}{\texttt{up}_n}{\cdot},
\}$
aborts \verb|router|$_{n}$.
The cases for the latency-oriented \noc design are analogous.

The only way to circumvent this limitation is by splitting each router into separate
actions for retrieving and sending using additional router-local registers.
Then it is possible to first issue all $\texttt{read}_0$'s, save the messages into
the router-local intermediate registers, and then issue all $\texttt{write}_0$'s
in a final set of actions.
But the additional router-local registers bloat the generated Verilog and
the increased number of actions unnecessarily complicates formal reasoning.
The root cause of this problem lies in the design of \koika's transactional
register semantics.
Specifically, the decision to prevent read access to the initial value
(\kread{0}{r}) after an earlier action wrote a new value (\kwrite{p}{r}{\cdot}).
This makes two \emph{dependent} actions in a schedule
\ilstrocq{a}$_{1}$ \texttt{|>} \ilstrocq{a}$_{2}$ \texttt{|>} \ilstrocq{§4:done§},
connected via \kwrite{0}{r}{\cdot} and \kread{1}{r}, appear as a single fused action.

\noindent
\hspace{0.2cm}
\begin{minipage}[t]{0.31\linewidth}
\captionof*{figure}{Action $a_{1}$}
\vspace{-\baselineskip}
\begin{lstlisting}[style=coq]
let initial = read$_{0}$(r) in
write$_{0}$(r, §'b110§);
  $\ldots$
write$_{0}$(s, initial);
\end{lstlisting}
\end{minipage}
\hfill
\begin{minipage}[t]{0.54\linewidth}
\captionof*{figure}{Action $a_{2}$}
\vspace{-\baselineskip}
\begin{lstlisting}[style=coq]
let initial = read$_{0}$(r) in      (* aborts *)
let last_written = read$_{1}$(r) in
assert (last_written == §'b110§);
  $\ldots$
\end{lstlisting}
\end{minipage}
\hspace{0.2cm}
\\[3pt]

Aborting a consecutive \kread{0}{r} and thereby forcing the developer
into using \kread{1}{r}, facilitates monotonic evolution of the registers.
That is, a read access always sees the latest value written.
But in a single (fused) action, it would always be possible to store the initial
value from the beginning of the cycle in a local variable to access it even after
updating the register (see action \ilstrocq{a}$_{1}$).
But the monotonic semantics prevent action \ilstrocq{a}$_{2}$ from reading
the initial value of register \ilstrocq{r}.
We argue that this restriction is not necessary and accessing the initial value
from the beginning of a cycle is a typical hardware design pattern.
Therefore, we make the initial value available by default to all actions in
a \koika schedule.

\section{\koika extensions}
\label{sec:extensions}

In this section, we describe our contributions to \koika that help
us to overcome the limitations outlined in the previous section.
At first, we present the \oursem (RAW) for \koika and then
construct our program logic on top of that.
Our \oursem removes some of the limitations for reading
previously written values (see \cref{sec:challenge2}).
Later, this also simplifies the design of the program logic.

\subsection{\oursem}
To make our \noc design, but also other hardware
designs, more efficient in terms of the generated Verilog and
the number of actions to formally reasoning about,
we propose a semantics that allows for reads after writes.
\begin{figure}
  \begin{tikzpicture}[node distance = 0cm]
    \node[] (a1) at (0,0) {\textsf{Action} $a$};
    \node[] (a) [right=of a1] {$::=$};
    \node[] (b) [right=of a] {
      $b$ $\mid$ $x$  $\mid $ \ilstrocq{§4:skip§} $\mid$ \ilstrocq{let} $x$ \ilstrocq{:=} $a_1$ \ilstrocq{in} $a_2$ $\mid$ \ilstrocq{§4:abort§} $\mid$
    };
    \node[] (c) [below right=-0.15cm and 0cm of b.south west] {
      \ilstrocq{if} $a_1$ \ilstrocq{then} $a_2$ \ilstrocq{else} $a_3$ $\mid$ \kread{p}{r} $\mid$ \kwrite{p}{r}{a} $\mid$ \texttt{$f$($a_1$,$\ldots$,$a_n$)}
    };
    \node[] (d) [below = 0.55cm of a] {$::=$};
    \node[] (d1) [left=of d] {\textsf{Ports} $p$};
    \node[] (d2) [right=of d] {$0$ $\mid$ $1$};

    \node[] (e) [right=0.4cm of d2] {\textsf{Variables} $\: x$};
    \node[] (f) [right=0.4cm of e] {\textsf{Registers} $\: r$};
    \node[] (g) [right=0.4cm of f] {\textsf{Externals} $\: f$};
    \node[] (h) [right=0.4cm of g] {\textsf{Bitstrings} $\: b$};

    \node[] (i1) [above=0.2cm of a] {$::=$};
    \node[] (i2) [left=of i1] {\textsf{Schedule} $s$};
    \node[] (j) [right=of i1] {
      \ilstrocq{§4:done§} $\mid$ $rule$ \texttt{|>} $s$
    };

  \end{tikzpicture}

  \caption{Syntax of \koika actions.}
  \label{fig:syntax}

\end{figure}
\begin{definition}[Syntax of \koika]\label{def:koika:syn}
  We define the syntax of \koika in \cref{fig:syntax}.
  A \koika programs $P:=(\mathbb{R},\mathbb{A},s)$ consists of
  a set of registers $\mathbb{R}$,
  a set of actions $\mathbb{A}$ and
  a schedule $s$, which defines their order of execution.
  An action is an expression composed of
  (constant) bit strings $b$,
  variables $x$,
  conditionals,
  binds,
  reads and writes to registers $r$,
  no-ops (\texttt{skip}),
  calls to functions $f$, and
  explicit action-\texttt{abort}s.
\end{definition}
Instead of recreating a new semantics, we integrate
ours as a switch into the existing ones.
We do so because other designs might in turn benefit from
the more restricted original ones.
Therefore, we add a Read-After-Write (\textcolor{blue}{\texttt{raw}}) switch to
\koika's semantics.
This switch only effects the evaluation of reads, i.e.,
\kread{0}{\cdot} and \kread{1}{\cdot}.
%
The semantics for all other language constructs of \koika
remain unaffected.
\begin{figure}
\begin{tikzpicture}
  \matrix (m) [
  matrix of nodes,
  ampersand replacement=\&,
  row sep=0.3cm
  ] {
  {$
  \inference
  []
  {{\color{c41} ( \kwrite{0}{r}{\cdot} \notin L ~ \land ~
    \kwrite{1}{r}{\cdot} \notin L )
  } \color{blue} ~ \lor ~ \texttt{raw}}
  { (l,\Ctx,\kread{0}{r}) \oursemfun (l,\Ctx,\R[r]) }
  [\textsc{Read0}\textcolor{blue}{$_{\text{raw}}$}]
  $}

  \\

  {$
  \inference
  []
  {\textcolor{c41}{\kwrite{1}{r}{\cdot} \notin L} ~ \textcolor{blue}{\lor} ~
  \textcolor{blue}{\texttt{raw}}
  }
  { (l,\Ctx,\kread{1}{r}) \oursemfun (l,\Ctx,v) ~
    \text{where} ~
    v ~ := ~
    \left\{
    {\begin{array}{@{}l l@{}l@{}}
    \R[r] & \text{if} ~ \kwrite{0}{r}{\cdot} ~ & \notin L \mdoubleplus l \\
    v_{0} & \text{if} ~ \kwrite{0}{r}{v_{0}} ~ & \in L \mdoubleplus l
    \end{array}}
    \right.
  }
  [\textsc{Read1}\textcolor{blue}{$_{\text{raw}}$}]
  $}

  \\

  {$
  \inference
  []
  { (l,\Ctx,a) \koikasemfun (l',\Ctx',v)   }
  { (l,\Ctx,a) \oursemfun (l',\Ctx',v)  }
  [\textsc{Original}]
  $}
  \\
   };
 \end{tikzpicture}

 \caption{
   Support for reads after writes in \koika unfolds into 4 rules.
   \textsc{Read0} and \textsc{Read1} are the \textcolor{c41}{original} rules
   that place preconditions on the occurrence of writes in
   the log $L$.
   Our \textcolor{blue}{Read-After-Write} Semantics adds rules \textsc{Read0}\textcolor{blue}{$_{\text{raw}}$}
   and \textsc{Read1}\textcolor{blue}{$_{\text{raw}}$} where these
   preconditions are omitted.
   The remainder of the rules stays unchanged.
 }
 \label{fig:sem_raw}
\end{figure}
\begin{definition}[\oursem]\label{def:koika:sem}
  We define \oursem as \oursemfun{} in \cref{fig:sem_raw}.
  Rules~\textsc{Read0}\textcolor{blue}{$_{\text{raw}}$} and
  ~\textsc{Read1}\textcolor{blue}{$_{\text{raw}}$} define the new
  semantics for reading values from registers.
  Rule~\textsc{Original} reestablishes the existing
  semantics \koikasemfun{} for the rest of the
  syntactic constructs for actions $a$.
\end{definition}
When \verb|raw| is
\textcolor{c41}{false} then the
original semantics are restored and reads are allowed
only when the corresponding preconditions hold.
For Rule~\textsc{Read0}, no previous writes from other
actions to the same register are allowed (to occur
in the scheduler log $L$).
For Rule~\textsc{Read1}, the final value must
not have been written yet ($\kwrite{1}{r}{\cdot} \notin L$),
then the returned value $v$
is either the initial value from the register environment or
the current value in the logs $L$ and $l$.%
\footnote{
  There can only be one written value for a register in the
  log because \koika prevents Write-After-Writes.
}
When \texttt{raw} is \textcolor{blue}{true} then
reads are allowed without any precondition, i.e.,
even to a previously written register.
The respective rules
\textsc{Read0}\textcolor{blue}{$_{\text{raw}}$}
\textsc{Read1}\textcolor{blue}{$_{\text{raw}}$} just vacuously
satisfy the preconditions (premises) but leave the conclusion
untouched.
With this Read-After-Write support, we had to adapt
the compiler, the generated Verilog code and respectively
reestablish the compiler correctness theorem.
The key property of \koika programs, data race freedom,
was not effected by our changes and remains valid.

\paragraph{Original vs. \texttt{raw}-semantics}

Compared to the original semantics of \koika, our new semantics resolve
both of the previously outlined challenges. They remove the unnecessary
\texttt{abort} mentioned in \cref{sec:challenge2}, allowing for pipeline parallel
designs. But, they also simplify the state model of registers mentioned
in \cref{sec:challenge1}. The latter is presented in full detail in the
following section.

\subsection{Program logic}
\label{sec:programlogic}

In order to make reasoning about \koika programs tractable,
we implement a program logic, in our case a \emph{Hoare logic}.
We define assertions to reason about registers and logs first
and then state the structural and the term rules of the Hoare logic.
We conclude this section with a proof of soundness and
our initial design for the automation of proofs in our
Hoare logic.

\subsubsection{Foundations}

Informally, a specification in our program logic is the usual
Hoare triple \hoare{P}{a}{Q}.
Precondition \pre{$P$} is a proposition that needs to hold before the (symbolic)
execution of action $a$.
Postcondition \post{$Q$} is a proposition that needs to hold after the
execution of $a$.
Both, $P$ and $Q$ reason about the whole state of a \koika action.
Thus, this reasoning needs access to the register environment $\R$, the
logs for the scheduler $L$ and the current action $l$, and the
variable context \Ctx (\Cref{sec:challenge1}).
A convenient side effect of our new semantics is that the abort handling
no longer depends on a distinction between scheduler and action logs.%
\footnote{
  Though not necessary for the semantics, the logs are still kept separate for
  an easier rewind on abort.
}\textsuperscript{,}%
\footnote{
  Later in \cref{fig:states} this can be observed as the action transitions (%
\begin{tikzpicture}
  \draw[->, gray, densely dashed] (0,0) -- +(0.3,0);
\end{tikzpicture}) only loop on
  their state instead of leading to new ones.
}
Thus, reasoning can be simplified to a single combined log $\cL := L \mdoubleplus l$.
Pre- and postconditions are then assertions defined as follows.
\begin{definition}[Assertions]
  Assertions are propositions that reason
  over the register environment $\R$,
  the combined log $\cL$, and
  the variable context \Ctx
  \[
    \{ P \} := \lambda ~ \R ~ \cL ~ \Ctx. ~ P ~ (\R, \cL, \Ctx)
  \]
\end{definition}
%
According to \cref{def:koika:syn,def:koika:sem},
\koika actions are expressions that return a value.
If a postcondition \post{$Q$} would also like to reason about
this return value then we write \post{$v.~Q$}.
For a particular value $v$, we write \post{$Q~(v)$}.
For a particular log $l$, we write \post{$Q~(l,~v)$}.
The reasoning about the (combined) log needs to
distinguish between values written via \kwrite{0}{\cdot}{\cdot} and \kwrite{1}{\cdot}{\cdot}.
To access a written value $v$ at port $p$ from a register $r$,
we define operator $\cL[r]_{p}$.
And, since this access is partial, we also define an operator \texttt{?\!:}
for chaining accesses.
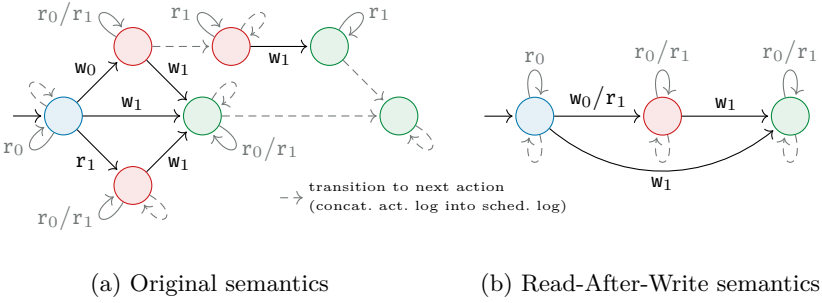
\begin{figure}
  \centering
  \begin{subfigure}[t]{0.45\textwidth}
    \pgfmathsetlengthmacro{\nodedist}{1.3cm}
\begin{tikzpicture}[shorten >=1pt,
  node distance=\nodedist, on grid,
  auto,
  initial text=,
  loop belowl/.style={in=210,out=240,loop},
  loop belowr/.style={in=300,out=330,loop},
  loop abover/.style={in=30,out=60,loop},
  loop abovel/.style={in=120,out=150,loop},
  cycle/.style={densely dashed, gray},
  noop/.style={gray},
  state0/.style={fill=blue!10,draw=blue},
  state1/.style={fill=red!10,draw=red},
  state2/.style={fill=green!10,draw=green},
  every state/.style={
    minimum size=5mm,
    inner sep=0pt
  },
  every node/.style={ inner xsep=0pt, inner ysep=2pt },
  ]
  \pgfmathsetlengthmacro{\dist}{\nodedist}
  \node[state, state0, initial] (q_0)                      {};
  \node[state, state1]          (q_1) [above right=of q_0] {};
  \node[state, state1]          (q_2) [below right=of q_0] {};
  \node[state, state2]          (q_3) [below right=of q_1] {};
  \node[state, state1]          (q_4) [      right=\dist of q_1] {};
  \node[state, state2]          (q_5) [      right=\dist of q_4] {};
  \node[state, state2]          (q_6) [below right=of q_5] {};

  \path[->] (q_0) edge               node        {\texttt{w$_0$}}    (q_1)
                  edge               node [swap] {\texttt{r$_1$}}    (q_2)
                  edge               node        {\texttt{w$_1$}}    (q_3)
                  edge [loop belowl,noop] node [left=2pt] {\texttt{r$_0$}} ()
                  edge [loop abovel,cycle]               ()
            (q_1) edge               node        {\texttt{w$_1$}}    (q_3)
                  edge [cycle]                           (q_4)
                  edge [loop abovel,noop] node [left=2pt] {\texttt{r$_0/$r$_1$}} ()
            (q_2) edge               node [swap] {\texttt{w$_1$}}    (q_3)
                  edge [loop belowl,noop] node [left=2pt] {\texttt{r$_0/$r$_1$}} ()
                  edge [loop belowr,cycle]               ()
            (q_3) edge [cycle]                           (q_6)
                  edge [loop belowr,noop] node [right=2pt] {\texttt{r$_0/$r$_1$}} ()
                  edge [loop abover,cycle]               ()
            (q_4) edge               node [swap] {\texttt{w$_1$}}    (q_5)
                  edge [loop abovel,noop] node [left=2pt] {\texttt{r$_1$}}    ()
                  edge [loop abover,cycle]               ()
            (q_5) edge [cycle]                           (q_6)
                  edge [loop abover,noop] node [right=2pt] {\texttt{r$_1$}}    ()
            (q_6) edge [loop belowr,cycle]               ();
\end{tikzpicture}
    \caption{Original semantics}
    \label{fig:states:original}
  \end{subfigure}
  \begin{subfigure}[t]{0.05\textwidth}
    \centering
    \vspace{-0.7cm}
    \hspace*{-2cm}
    \begin{tikzpicture}[
  font=\footnotesize
]
  \matrix[
    fill=white,
    row sep=5pt,
    column sep=2pt,
    inner sep=0pt,
    column 1/.style={anchor=center},
  ] at (0,0) {
      \node {
      \begin{tikzpicture}
        \draw[->, gray, densely dashed] (0,0) -- +(0.3,0);
      \end{tikzpicture}
      }; &
      \node[anchor=west] {\tiny \shortstack[l]{
        transition to next action \\
        (concat. act. log into sched. log)
      }}; \\
  };
\end{tikzpicture}
  \end{subfigure}
  \begin{subfigure}[t]{0.37\textwidth}
    \centering
    \pgfmathsetlengthmacro{\nodedist}{1.3cm}
\begin{tikzpicture}[shorten >=1pt,
  node distance=\nodedist, on grid,
  auto,
  initial text=,
  loop belowl/.style={in=210,out=240,loop},
  loop belowr/.style={in=300,out=330,loop},
  loop abover/.style={in=30,out=60,loop},
  loop abovel/.style={in=120,out=150,loop},
  cycle/.style={densely dashed, gray},
  noop/.style={gray},
  state0/.style={fill=blue!10,draw=blue},
  state1/.style={fill=red!10,draw=red},
  state2/.style={fill=green!10,draw=green},
  every state/.style={
    minimum size=5mm,
    inner sep=0pt
  },
  every node/.style={ inner xsep=0pt, inner ysep=2pt },
  ]
  \pgfmathsetlengthmacro{\dist}{\nodedist * 1.3}
  \node[state, state0, initial] (q_0)                      {};
  \node[state, state1]          (q_1) [right=\dist of q_0] {};
  \node[state, state2]          (q_2) [right=\dist of q_1] {};

  \path[->] (q_0) edge                   node {\texttt{w$_0/$r$_1$}} (q_1)
                  edge [bend right=40] node [swap] {\texttt{w$_1$}} (q_2)
                  edge [loop above,noop] node {\texttt{r$_0$}} ()
                  edge [loop below,cycle] ()
            (q_1) edge                   node {\texttt{w$_1$}} (q_2)
                  edge [loop above,noop] node {\texttt{r$_0/$r$_1$}} ()
                  edge [loop below,cycle] ()
            (q_2) edge [loop above,noop] node {\texttt{r$_0/$r$_1$}} ()
                  edge [loop below,cycle] ();

  \node[state,draw=none]         (q_3) [below right=of q_0] {}; 
  \path (q_3) edge [loop belowl,noop,draw=none] node [left=2pt] {\phantom{r0/r1}} (); 
\end{tikzpicture}
    \caption{\oursem}
    \label{fig:states:ours}
  \end{subfigure}
  \caption{
    Unique states of a single register which can be distinguished solely by their abort semantics.
    (All transitions which are left out lead to an \texttt{abort}, and
      every register is reset to the initial state at the beginning of each cycle)
  }
  \label{fig:states}
\end{figure}
\[
\begin{array}{c@{\hspace{2cm}}c}
{\cL[r]_p :=
\left\{
\begin{array}{@{}l l@{}l@{}l@{}}
\texttt{Some} \; v & \text{if} \;\; \kwrite{p}{r}{v}     & \; \in    & \; \cL \\
\texttt{None}      & \text{if} \;\; \kwrite{p}{r}{\cdot} & \; \notin & \; \cL
\end{array}
\right.} &
{a \,\texttt{?\!:}\, b ~ :=
\left\{
\begin{array}{@{}l l@{}}
v & \text{if} \;\; a = \texttt{Some} \; v \\
b & \text{if} \;\; a = \texttt{None}
\end{array}
\right.}
\end{array}
\]
With these operators, we define a relation that resolves some of
the complexity of reasoning about the register environment \R and
the logs \cL.
In \cref{sec:challenge1}, we calculated $2^{8}$ potential
states per register exist.
\cref{fig:states} shows that the 7 distinct states for the original
semantics reduce to 3 distinct states for our RAW semantics.
Note their difference in complexity, i.e. transitions, which directly
translates into mental overhead during the design and additional proof
burden during the verification of a hardware description.
Next, we capture the 3 new states in a relation for concise reasoning.
\begin{figure}
  \colorlet{cond1}{blue}
  \colorlet{cond2}{red}
  \colorlet{cond3}{green}
  \centering
  \begin{tikzpicture}
  \matrix (m) [
  matrix of math nodes,
  ampersand replacement=\&,
  row sep=0.3cm
  ] {
  \inference
  { \textcolor{cond1}{P_0 (\R[r])} & \textcolor{cond2}{\kread{1}{r} \notin \cL} & \textcolor{cond2}{\kwrite{0}{r}{\cdot} \notin \cL} & \textcolor{cond3}{\kwrite{1}{r}{\cdot} \notin \cL} }
  {r \Mapsto_{\R,\cL} \texttt{[\textcolor{cond1}{$P_0$},\textcolor{cond2}{/},\textcolor{cond3}{/}]}}
  [\textsc{RegEnv}]
  \\
  \inference
  { \textcolor{cond1}{P_0 (\R[r])} & \textcolor{cond2}{P_1 (\cL[r]_0 ~ \texttt{?\!:} ~ \R[r])} & \textcolor{cond3}{\kwrite{1}{r}{\cdot} \notin \cL} }
  {r \Mapsto_{\R,\cL} \texttt{[\textcolor{cond1}{$P_0$},\textcolor{cond2}{$P_1$},\textcolor{cond3}{/}]}}
  [\textsc{RegState0}]
  \\
  \inference
  { \textcolor{cond1}{P_0 (\R[r])} & \textcolor{cond2}{P_1 (\cL[r]_0 ~ \texttt{?\!:} ~ \R[r])} & \textcolor{cond3}{P_2 (\cL[r]_1 ~ \texttt{?\!:} ~ \cL[r]_0 ~ \texttt{?\!:} ~ \R[r])}   }
  {r \Mapsto_{\R,\cL} \texttt{[\textcolor{cond1}{$P_0$},\textcolor{cond2}{$P_1$},\textcolor{cond3}{$P_2$}]}}
  [\textsc{RegState1}]
  \\
  };
  \end{tikzpicture}

  \caption{
    The relation to reason about the values of a single
    register $r$ across the register environment and the logs.
  }
  \label{fig:logstate}

\end{figure}
\begin{definition}[Log properties]
  We define the relation to reason about a register $r$
  in the context of a register environment \R and a
  combined log \cL in \cref{fig:logstate} as
  \mbox{$r \Mapsto_{\R,\cL}$ \texttt{[$P_0$,$P_1$,$P_2$]}}. In this relation
  $P_{0}$, $P_{1}$, and $P_{2}$ are
  propositions on the value of the register.
\end{definition}
The notation \mbox{\texttt{[$P_0$,$P_1$,$P_2$]}} refers to the three
possible values that a register could have at the same time.
$P_{0}$ is a proposition about the initial value from the
register environment $\R[r]$.
$P_{1}$ and $P_{2}$ reason about the values written with \kwrite{0}{r}{\cdot}
and \kwrite{1}{r}{\cdot}, respectively.
For example, to assert an initial value $v_{0}$ and a
written value $v_{1}$, we simply write \mbox{\texttt{[$=\!v_0$,$=\!v_1$,/]}}.
Note the use of our operators, for example
rule~\textsc{RegState1} covers the case where there is no
\kwrite{0}{r}{\cdot} in the log such that $P_{1}$ needs to hold
for the initial value in \R.
We also support a notation such as \mbox{\texttt{[?,$P_1$,/]}} to reason
exclusively about a previously written value by setting $P_0 := \top$.
Note how this register relation makes register conditions significantly easier to comprehend.
The value of a \kread{0/1}{r} can directly be deduced from $P_{0/1}$
and likewise the failure of a \kwrite{0/1}{r}{\cdot} from $P_{1/2}$.

\subsubsection{Rules}
Many of the rules for reasoning in our program logic
are standard.
Hence, we focus on the ones that are unique to \koika and
define the full set of rules in the \cref{sec:appendix:hoare}.
\begin{figure}
\centering
\begin{tikzpicture}[node distance=0cm]

\matrix (m1) at (0,0) [
  inner sep=0pt,
  matrix of math nodes,
  ampersand replacement=\&,
  column sep=0.25cm
  ] {
    \inference
    {  }
    { \hoare{P}{\texttt{abort}}{\bot}~~ }
    [\scriptsize\hspace{-1.8em}\raisebox{-0.8em}{\textsc{Abort}}]
    \&
    \inference
    {  }
    { \hoare
      { Q~(\cL[r]_0 ~ \texttt{?\!:} ~ \R[r]) }
      { \kread{1}{r} }
      { Q }~~
    }
    [\scriptsize\hspace{-1.8em}\raisebox{-0.8em}{\textsc{HRead1}}]
    \\
};

\matrix (m1a) [below =1em of m1] [
  inner sep=0pt,
  matrix of math nodes,
  ampersand replacement=\&,
  column sep=0.25cm
  ] {
    \inference
    { \hoare
      { P }
      { a }
      { v. ~ Q~(\cL \mdoubleplus [\kwrite{p}{r}{v}], ~ \epsilon) }
    }
    { \hoare{ P }{ \kwrite{p}{r}{a} }{ Q } }
    [\scriptsize\hspace{-1.8em}\raisebox{-0.8em}{\textsc{HWrite}}]
    \&
    \inference
    {  }
    { \hoare
      { Q~(\R[r]) }
      { \kread{0}{r} }
      { Q }~~
    }
    [\scriptsize\hspace{-1.8em}\raisebox{-0.8em}{\textsc{HRead0}}]
  \\
};

\end{tikzpicture}

  \caption{Reasoning rules for selected \koika actions.}
  \label{fig:rules:action}

\end{figure}
\begin{definition}[Action rules (excerpt)]
  The term rules of our program logic to read from and
  write to a register, and aborting an action
  are defined in \cref{fig:rules:action}.
\end{definition}
For \kread{0}{\cdot}, we retrieve the value $v$ from the register
environment \R.
And for \kread{1}{\cdot}, we either find \kwrite{0}{r}{v} in the
log \cL or load it again from \R.
For a write to a register, we require that $Q$ holds when
passed the log \cL with the appended \kwrite{p}{r}{v}.
Finally, the \textsc{Abort} rule maps an \texttt{abort}
to $\bot$ and thereby specifies our Hoare triples as partial.%
\footnote{
  As the language does not support any form of looping by design, the
  usual definition of partiality would not apply to \koika.
}
For total Hoare triples denoted as \thoare{P}{a}{Q},
we drop this rule and thereby require that the action never
aborts.
If it happens in the course of a proof then the proof is simply
stuck.

\subsubsection{Soundness}

To prove the soundness of our program logic, we define
the Hoare triples formally.
\koika's big-step evaluation $\downarrow_{\R,L}$ is partial and
so is our \oursem. The only unevaluable (stuck) terms are the ones which result in
an \texttt{abort}.
Based on this partial evaluation, we define partial and total
weakest-preconditions (WPs).
\begin{definition}[Partial WPs for \koika actions]\label{def:wpa:partial}
\begin{equation*}
\wpp{a}{a}{Q} := (l,\Ctx,a) \downarrow_{\R,L}^{\text{raw}} (l',\Ctx',v) \Rightarrow \post{Q~(\R ,L \mdoubleplus l',\Ctx',v)}
\end{equation*}
\end{definition}
That is, if an action $a$ aborts the precondition reduces to $\bot$ and the postcondition
follows vacuously.
\begin{definition}[Total WPs for \koika actions]\label{def:wpa:total}
\begin{equation*}
  \wpt{a}{a}{Q} := (l,\Ctx,a) \downarrow_{\R,L}^{\text{raw}} (l',\Ctx',v) \land \post{Q~(\R ,L \mdoubleplus l',\Ctx',v)}
\end{equation*}
\end{definition}
Drawing a $\bot$ conclusion for a total WP is not possible.
That is, for a proven total WP, the execution never aborts.
From these WP definitions, we derive Hoare triples.
\begin{definition}[Partial Hoare Triples for \koika actions]
\begin{equation*}
  \hoare{P}{a}{Q} := \forall ~ \R ~ L ~ l ~ \Ctx. ~ \pre{P~(\R ,L \mdoubleplus l,\Ctx)} \vdash \wpp{a}{a}{Q}
\end{equation*}
\end{definition}
Note, that both the precondition and postcondition reason over the combined log.
The definition for total Hoare triples is analogous.
\begin{definition}[Total Hoare Triples for \koika actions]\label{def:hoare:total}
\begin{equation*}
  \thoare{P}{a}{Q} := \forall ~ \R ~ L ~ l ~ \Ctx. ~ \pre{P~(\R ,L \mdoubleplus l,\Ctx)} \vdash \wpt{a}{a}{Q}
\end{equation*}
\end{definition}
For the scheduler, we define only total WPs because a
\koika program can not fail.
\begin{definition}[WPs for the schedules]\label{def:wps}
\[
\wpt{s}{s}{R} ~ := ~
\left\{
\begin{array}{@{}l l@{}}
R                          & \text{if} ~ s = \texttt{done} \\
\wpt{a}{a}{\wpt{s}{s'}{R}} & \text{if} ~ s = a~\texttt{|>}~s'
\end{array}
\right.
\]
\end{definition}
This definition builds upon the WPs for actions.
For the time being, we only reason about schedules where
actions do not abort.
It is straightforward to add reasoning support for schedules
with failing actions.
With these definitions in place, we can define the
soundness of our program logic for \koika.
Since our Hoare Triples directly unfold to WPs, we use those to state
the adequacy, first for actions and then for schedules.
\begin{theorem}[Adequacy of \koika actions]\label{theo:adequate:action}
\begin{equation*}
  \inference
  {(l,\Ctx,a) \textnormal{\oursemfun} (l',\Ctx',v) & \wpp{a}{a}{Q} }
  {\post{Q~(\R ,L \mdoubleplus l',\Ctx',v)}}
  []
\end{equation*}
\end{theorem}
\begin{proof}
  The proof is immediate by unfolding
  $\wpp{a}{a}{Q}$ (\Cref{def:wpa:partial}).
\end{proof}
\begin{theorem}[Adequacy of \koika schedules]\label{theo:adequate:schedule}
\begin{equation*}
  \inference
  {(\R,L,s) \downarrow_{\textnormal{\text{Sched}}}^{\textnormal{\text{raw}}} (\R,L',\epsilon) & \wpp{s}{s}{Q} }
  {\post{Q~(\R ,L')}}
  []
\end{equation*}
\end{theorem}
\begin{proof}
  The proof is by induction on $s$.
  The base case, where $s=\ilstrocq{§4:done§}$ is immediate because
  $L=L'$ and, by \cref{def:wps}, $\wpp{s}{\ilstrocq{§4:done§}}{Q}$
  unfolds to $Q$.
  The induction step, where $s=a\,\texttt{|>}\,s'$ unfolds
  $\wpp{s}{a\,\texttt{|>}\,s'}{Q}$ to
  $\wpt{a}{a}{\wpt{s}{s'}{R}}$ (\Cref{def:wps}).
  The WP for the action follows from the adequacy for actions
  (\Cref{theo:adequate:action}) and then we apply the
  induction hypothesis to conclude the proof.
\end{proof}

\paragraph*{Proof automation}
Based on this WP-encoding, we provide the usual tactic support
to symbolically execute an action in the \rocq proof mode.

\section{A parameterized \noc library in \koika}
\label{sec:noc}

Equipped with the \oursem and the program logic, we can implement
a library that allows for $k$-dimensional \noc designs.
We start with an informal overview of our design and define the
input to our \noc library, i.e., the parameters that the
user has to provide and that our design is parametric in.
Then, we craft a specification for Network-on-Chips and
implement it in \coq.
Afterwards, we use our \oursem to implement an efficient
\noc library in \koika.
Finally, we take advantage of our program logic to prove
that this implementation refines the specification.

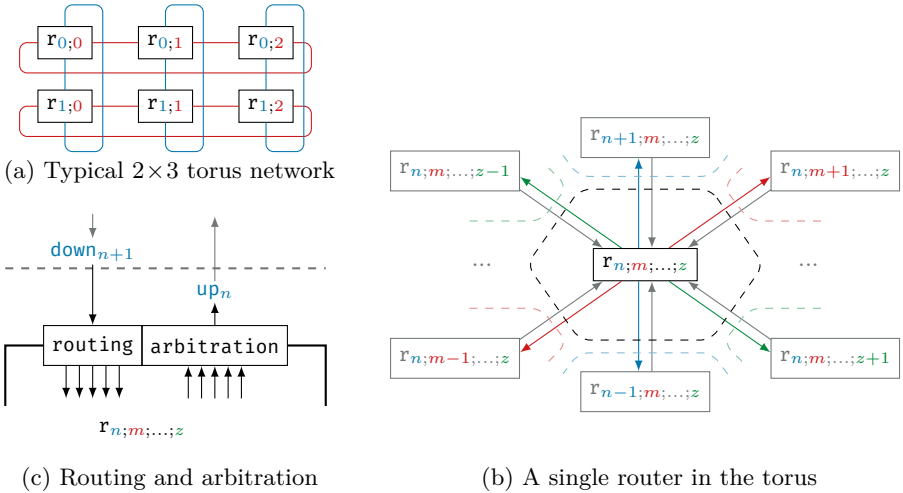
\begin{figure}
  \centering
  \begin{minipage}[b]{0.38\textwidth}
    \vspace{0pt}
      \centering
      \begin{subfigure}[t]{\linewidth}
          \centering
          \newcommand{\spcr}{0.4cm}
\newcommand{\spcc}{0.5cm}
\begin{tikzpicture}[
  font=\footnotesize
]
  \matrix (net) [
    column sep=0.6cm,
    row sep=0.4cm,
    matrix of nodes,
    nodes in empty cells,
    nodes={outer sep=0pt,draw},
  ]
  {
    \texttt{\texttt{r}$_{\textcolor{blue}{0};\textcolor{red}{0}}$} & \texttt{\texttt{r}$_{\textcolor{blue}{0};\textcolor{red}{1}}$} & \texttt{\texttt{r}$_{\textcolor{blue}{0};\textcolor{red}{2}}$} \\
    \texttt{\texttt{r}$_{\textcolor{blue}{1};\textcolor{red}{0}}$} & \texttt{\texttt{r}$_{\textcolor{blue}{1};\textcolor{red}{1}}$} & \texttt{\texttt{r}$_{\textcolor{blue}{1};\textcolor{red}{2}}$} \\
  };

  \begin{scope}[rounded corners=0.1cm, red]
    \draw (net-1-1) -- (net-1-2);
    \draw (net-2-1) -- (net-2-2);
    \draw (net-1-2) -- (net-1-3);
    \draw (net-2-2) -- (net-2-3);
    \draw (net-1-3) -| ($(net-1-3) + (1.5*\spcr,-\spcr)$) -- ($(net-1-1) + (-1.5*\spcr,-\spcr)$) |- (net-1-1);
    \draw (net-2-3) -| ($(net-2-3) + (1.5*\spcr,-\spcr)$) -- ($(net-2-1) + (-1.5*\spcr,-\spcr)$) |- (net-2-1);
  \end{scope}

  \begin{scope}[rounded corners=0.1cm, blue]
    \draw (net-1-1) -- (net-2-1);
    \draw (net-1-2) -- (net-2-2);
    \draw (net-1-3) -- (net-2-3);

    \draw (net-1-1) |- ($(net-1-1) + (\spcc,\spcc)$) -- ($(net-2-1) + (\spcc,-1.2*\spcc)$) -| (net-2-1);
    \draw (net-1-2) |- ($(net-1-2) + (\spcc,\spcc)$) -- ($(net-2-2) + (\spcc,-1.2*\spcc)$) -| (net-2-2);
    \draw (net-1-3) |- ($(net-1-3) + (\spcc,\spcc)$) -- ($(net-2-3) + (\spcc,-1.2*\spcc)$) -| (net-2-3);
  \end{scope}


\end{tikzpicture}
          \caption{Typical $2\!\times\!3$ torus network}
          \label{fig:2x3torus}
      \end{subfigure}

      \vspace{1em}

      \begin{subfigure}[t]{\linewidth}
          \setcounter{subfigure}{2}
          \centering
\newcommand{\rspace}{1cm}
\newcommand{\spacing}{2.3cm}
\newcommand{\len}{2.2cm}
\newcommand{\debug}{none}

\begin{tikzpicture}[every node/.style={draw,rectangle}]

  \node[draw=none, minimum height=1.5cm, minimum width=3cm] (r) at (0,0) {};
  \node[fill=white, minimum height=15pt, left=-0.1pt of r.north] (rout) {\footnotesize{\texttt{routing}}};
  \node[fill=white, minimum height=15pt, right=-0.1pt of r.north] (arb) {\footnotesize{\texttt{arbitration}}};

  \draw[thick] (rout.west) -- ++(left:0.5cm) -- ++(down:0.8cm);
  \draw[thick] (arb.east) -- ++(right:0.5cm) -- ++(down:0.8cm);

  \newcommand{\spc}{0.75cm}
  \draw[dashed, thick, gray] ($(rout.north west) + (-0.5cm,\spc)$) -- ($(arb.north east) + (0.5cm,\spc)$);

  \node[draw=none, text=blue, inner sep=1pt, above=0.8cm of rout] (down) {\texttt{down$_{n+1}$}};
  \draw[latex-] (rout) -- (down);
  \node[draw=none, text=blue, inner sep=1pt, above=0.3cm of arb] (up) {\texttt{up$_n$}};
  \draw[-latex] (arb) -- (up);

  \draw[-latex, gray] (up) -- +(up:1cm);
  \draw[latex-, gray] (down) -- +(up:0.5cm);

  \foreach \a in {0,-5,-10,5,10} {
    \draw[latex-, transform canvas={xshift=  \a pt}] (arb)  -- ++(down:0.7cm);
    \draw[-latex, transform canvas={xshift=  \a pt}] (rout) -- ++(down:0.7cm);
  }
  \node[draw=none, above=5pt of r.south, fill=white] {\texttt{r}$_{\textcolor{blue}{n};\textcolor{red}{m};...;\textcolor{green}{z}}$};

\end{tikzpicture}
          \caption{Routing and arbitration}
          \label{fig:routingarbitration}
      \end{subfigure}

  \end{minipage}
  \hfill
  \begin{minipage}[b]{0.58\textwidth}
    \vspace{0pt}
    \begin{subfigure}[t]{\linewidth}
      \setcounter{subfigure}{1}
      \centering
\newcommand{\rspace}{1.5cm}
\newcommand{\spacing}{1.5cm}
\newcommand{\len}{2.2cm}
\newcommand{\debug}{none}

\begin{tikzpicture}[every node/.style={draw,rectangle}]
  \begin{scope}[node distance=20pt]
    \node (a) at (0,0) {\texttt{r}$_{\textcolor{blue}{n};\textcolor{red}{m};...;\textcolor{green}{z}}$};

    \node[draw=none, minimum height=15pt] (nxt1) at ($(a)!0.7 * (\spacing + 0.9cm)!   0:(up:1cm)$) {};
    \node[draw=none, minimum height=15pt] (nxt2) at ($(a)!0.9 * (\spacing + 0.9cm)! -55:(up:1cm)$) {};
    \node[draw=none, minimum height=15pt] (nxtn) at ($(a)!0.9 * (\spacing + 0.9cm)!-125:(up:1cm)$) {};
    \node[draw=none, minimum height=15pt] (prv1) at ($(a)!0.7 * (\spacing + 0.9cm)!-180:(up:1cm)$) {};
    \node[draw=none, minimum height=15pt] (prv2) at ($(a)!0.9 * (\spacing + 0.9cm)!-235:(up:1cm)$) {};
    \node[draw=none, minimum height=15pt] (prvn) at ($(a)!0.9 * (\spacing + 0.9cm)!-305:(up:1cm)$) {};
    \node[draw=none, minimum height=15pt, text=gray] (nxtd) at ($(a)!0.9 * (\spacing + 0.9cm)! -90:(up:1cm)$) {...};
    \node[draw=none, minimum height=15pt, text=gray] (prvd) at ($(a)!0.9 * (\spacing + 0.9cm)!-270:(up:1cm)$) {...};

    \begin{scope}[transform canvas={xshift=2.5pt}]
      \draw[-latex, draw=gray] (nxt1) -- (a);
      \draw[-latex, draw=gray] (prv1) -- (a);
    \end{scope}
    \begin{scope}[transform canvas={xshift=-7pt}]
      \draw[-latex, draw=gray] (prv2) -- (a);
      \draw[-latex, draw=gray] (prvn) -- (a);
    \end{scope}
    \begin{scope}[transform canvas={xshift=7pt}]
      \draw[-latex, draw=gray] (nxt2) -- (a);
      \draw[-latex, draw=gray] (nxtn) -- (a);
    \end{scope}

    \node[draw=gray, text=gray, fill=white, minimum height=15pt, above=0pt of nxt1.south] {\texttt{r}$_{\textcolor{blue}{n+1};\textcolor{red}{m};...;\textcolor{green}{z}}$};
    \node[draw=gray, text=gray, fill=white, minimum height=15pt, right=0pt of nxt2.west] {\texttt{r}$_{\textcolor{blue}{n};\textcolor{red}{m+1};...;\textcolor{green}{z}}$};
    \node[draw=gray, text=gray, fill=white, minimum height=15pt, right=0pt of nxtn.west] {\texttt{r}$_{\textcolor{blue}{n};\textcolor{red}{m};...;\textcolor{green}{z+1}}$};
    \node[draw=gray, text=gray, fill=white, minimum height=15pt, below=0pt of prv1.north] {\texttt{r}$_{\textcolor{blue}{n-1};\textcolor{red}{m};...;\textcolor{green}{z}}$};
    \node[draw=gray, text=gray, fill=white, minimum height=15pt, left=0pt of prv2.east] {\texttt{r}$_{\textcolor{blue}{n};\textcolor{red}{m-1};...;\textcolor{green}{z}}$};
    \node[draw=gray, text=gray, fill=white, minimum height=15pt, left=0pt of prvn.east] {\texttt{r}$_{\textcolor{blue}{n};\textcolor{red}{m};...;\textcolor{green}{z-1}}$};

    \begin{scope}[every node/.style={
      draw=none,fill=white,
      inner sep=1pt
    }]
    \node (up1)   at ($(a)!0.7 * \rspace!   0:(up:1cm)$) {
    };
    \node (up2)   at ($(a)!0.9 * \rspace! -55:(up:1cm)$) {
    };
    \node (upd)   at ($(a)!0.9 * \rspace! -90:(up:1cm)$) {
    };
    \node (upn)   at ($(a)!0.9 * \rspace!-125:(up:1cm)$) {
    };
    \node (down1) at ($(a)!0.7 * \rspace!-180:(up:1cm)$) {
    };
    \node (down2) at ($(a)!0.9 * \rspace!-235:(up:1cm)$) {
    };
    \node (downd) at ($(a)!0.9 * \rspace!-270:(up:1cm)$) {
    };
    \node (downn) at ($(a)!0.9 * \rspace!-305:(up:1cm)$) {
    };
    \end{scope}

    \begin{scope}[transform canvas={xshift=-2.5pt}]
      \draw[latex-, draw=blue] (nxt1) -- (a);
      \draw[latex-, draw=blue] (prv1) -- (a);
    \end{scope}


    \draw[latex-, draw=  red] (nxt2) -- (a);
    \draw[latex-, draw=green] (nxtn) -- (a);
    \draw[latex-, draw=  red] (prv2) -- (a);
    \draw[latex-, draw=green] (prvn) -- (a);

    \draw[draw=\debug, name path=l1] ($(a)!0.7 * \spacing!   0:(up:1cm)$) -- ([turn]90:\len);
    \draw[draw=\debug, name path=r1] ($(a)!0.7 * \spacing!   0:(up:1cm)$) -- ([turn]-90:\len);
    \draw[draw=\debug, name path=l2] ($(a)!0.9 * \spacing! -55:(up:1cm)$) -- ([turn]90:\len);
    \draw[draw=\debug, name path=r2] ($(a)!0.9 * \spacing! -55:(up:1cm)$) -- ([turn]-90:\len);
    \draw[draw=\debug, name path=l3] ($(a)!0.9 * \spacing!-125:(up:1cm)$) -- ([turn]90:\len);
    \draw[draw=\debug, name path=r3] ($(a)!0.9 * \spacing!-125:(up:1cm)$) -- ([turn]-90:\len);
    \draw[draw=\debug, name path=l4] ($(a)!0.7 * \spacing!-180:(up:1cm)$) -- ([turn]90:\len);
    \draw[draw=\debug, name path=r4] ($(a)!0.7 * \spacing!-180:(up:1cm)$) -- ([turn]-90:\len);
    \draw[draw=\debug, name path=l5] ($(a)!0.9 * \spacing!-235:(up:1cm)$) -- ([turn]90:\len);
    \draw[draw=\debug, name path=r5] ($(a)!0.9 * \spacing!-235:(up:1cm)$) -- ([turn]-90:\len);
    \draw[draw=\debug, name path=l6] ($(a)!0.9 * \spacing!-305:(up:1cm)$) -- ([turn]90:\len);
    \draw[draw=\debug, name path=r6] ($(a)!0.9 * \spacing!-305:(up:1cm)$) -- ([turn]-90:\len);

    \begin{scope}[
      name intersections={of=r1 and l2, by={pa}},
      name intersections={of=r2 and l3, by={pb}},
      name intersections={of=r3 and l4, by={pc}},
      name intersections={of=r4 and l5, by={pd}},
      name intersections={of=r5 and l6, by={pe}},
      name intersections={of=r6 and l1, by={pf}},
    ]
      \draw[dashed, rounded corners]
        ($(a)!0.95!(pa)$) --
        ($(a)!0.95!(pb)$) --
        ($(a)!0.95!(pc)$) --
        ($(a)!0.95!(pd)$) --
        ($(a)!0.95!(pe)$) --
        ($(a)!0.95!(pf)$) -- cycle;

      \foreach \a/\b in {pf/pa,pc/pd}
      \draw[dashed, rounded corners, draw=blue!50!white]
        ($(a)!1.30!-4:(\a)$) --
        ($(a)!1.05!-4:(\a)$) --
        ($(a)!1.05! 4:(\b)$) --
        ($(a)!1.30! 4:(\b)$);

      \coordinate (paa) at ($(a)!1.05!-4:(pa)$);
      \coordinate (pbb1) at ($(a)!1.05! 4:(pb)$);
      \coordinate (pbb2) at ($(a)!1.05!-4:(pb)$);
      \coordinate (pcc) at ($(a)!1.05! 4:(pc)$);
      \coordinate (pdd) at ($(a)!1.05!-4:(pd)$);
      \coordinate (pee1) at ($(a)!1.05! 4:(pe)$);
      \coordinate (pee2) at ($(a)!1.05!-4:(pe)$);
      \coordinate (pff) at ($(a)!1.05! 4:(pf)$);

      \draw[dashed, rounded corners, draw=red!50!white]
        ($(a)!1.30!-4:(pa)$) -- (paa) --
        ($(paa)!0.5!(pbb1)$) -- +(right:1cm);
      \draw[dashed, rounded corners, draw=green!50!white]
        ($(a)!1.30! 4:(pc)$) -- (pcc) --
        ($(pcc)!0.5!(pbb2)$) -- +(right:1cm);
      \draw[dashed, rounded corners, draw=red!50!white]
        ($(a)!1.30!-4:(pd)$) -- (pdd) --
        ($(pdd)!0.5!(pee1)$) -- +(left:1cm);
      \draw[dashed, rounded corners, draw=green!50!white]
        ($(a)!1.30! 4:(pf)$) -- (pff) --
        ($(pff)!0.5!(pee2)$) -- +(left:1cm);

    \end{scope}

  \end{scope}
\end{tikzpicture}
      \caption{A single router in the torus}
      \label{fig:routerND}
    \end{subfigure}
  \end{minipage}
  \caption{Foundations of our formal specification for $k$-dimensional \noc{}s.}
  \label{fig:combined}
\end{figure}

\subsection{Overview}

Our library provides formally-verified $k$-dimensional \noc{}s.
Naturally, our library and the respective specification abstracts over
\texttt{data\_t}, i.e., the payload, of a single packet that is transmitted
via the network.
More importantly, our development is parametric in the number of dimensions
$|\texttt{dims}|$, and the length of each individual dimension.
\begin{lstlisting}
Context data_t : §2:type§.
Context dims : §2:list nat§.
\end{lstlisting}

Based on these inputs, our library creates $k$-dimensional torus networks, i.e.,
each dimension has a wrap-around link.
\Cref{fig:2x3torus} shows a \texttt{[2;3]} torus.
Hence, the shortest path to transmit a message from router
\texttt{r}$_{\textcolor{blue}{0};\textcolor{red}{0}}$ to router
\texttt{r}$_{\textcolor{blue}{0};\textcolor{red}{2}}$
is not via \texttt{r}$_{\textcolor{blue}{0};\textcolor{red}{1}}$ but the
\textcolor{red}{direct connection} that wraps around.
This wrap-around simplifies our design because it removes the start and end such that
every router contains the same logic.
Routers only differ in their index, i.e., the location in the network, and respectively
the registers they are connected to.
Consequently, the whole network can be understood by examining the design of
a single router.
\Cref{fig:routerND} visualizes router
\texttt{r}$_{\textcolor{blue}{n};\textcolor{red}{m};...;\textcolor{green}{z}}$
in a $k$-dimensional \noc.
The index of the router, i.e., its location in the \noc{}, is
$\textcolor{blue}{n};\textcolor{red}{m};...;\textcolor{green}{z}$
with $|\textcolor{blue}{n};\textcolor{red}{m};...;\textcolor{green}{z}| = k$
dimensions.
The router has two neighbors per dimension, e.g.,
\texttt{r}$_{\textcolor{blue}{n-1};\textcolor{red}{m};...;\textcolor{green}{z}}$ and
\texttt{r}$_{\textcolor{blue}{n+1};\textcolor{red}{m};...;\textcolor{green}{z}}$ for
\textcolor{blue}{dimension 1}.
As in the pipeline design in \cref{fig:intermediate_routers}, two shared
channels \texttt{up} and \texttt{down} facilitate a bidirectional channel
per neighbor.
\Cref{fig:routingarbitration} shows the router logic that composes
of two algorithms, \texttt{routing} and \texttt{arbitration}.
For each incoming channel, the \texttt{routing} decides which output a message
should be forwarded to.
In our design, we consider only channels that can store at most a single message.
Hence, we send at most a single message along a channel per cycle.
For each outgoing channel, the \texttt{arbitration} selects which message to
prioritize in case of a conflict, i.e., when multiple
incoming channels are routed towards the same outgoing channel.

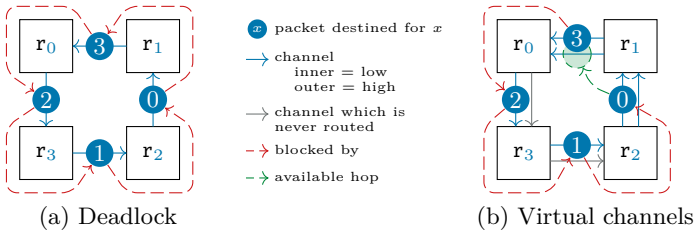
\begin{figure}
  \hspace{0.5cm}
  \begin{subfigure}[t]{0.40\textwidth}
    \centering
    \newcommand{\arrspc}{0pt}
\begin{tikzpicture}[
  font=\footnotesize
]
  \matrix (net) [
    column sep=0.7cm,
    row sep=0.7cm,
    matrix of nodes,
    nodes in empty cells,
    nodes={outer sep=0pt,draw,minimum size=0.7cm, inner sep=0pt},
  ]
  {
    \texttt{\texttt{r}$_{\textcolor{blue}{0}}$} & \texttt{\texttt{r}$_{\textcolor{blue}{1}}$} \\
    \texttt{\texttt{r}$_{\textcolor{blue}{3}}$} & \texttt{\texttt{r}$_{\textcolor{blue}{2}}$} \\
  };

  \begin{scope}[rounded corners=0.1cm, blue]
    \draw[<-] ($(net-1-1.east)  + (0, \arrspc)$) -- ($(net-1-2.west)  + (0, \arrspc)$); 
    \draw[<-] ($(net-2-2.west)  + (0,-\arrspc)$) -- ($(net-2-1.east)  + (0,-\arrspc)$); 

    \draw[<-] ($(net-2-1.north) + (-\arrspc,0)$) -- ($(net-1-1.south) + (-\arrspc,0)$); 
    \draw[<-] ($(net-1-2.south) + ( \arrspc,0)$) -- ($(net-2-2.north) + ( \arrspc,0)$); 
  \end{scope}



  \coordinate (p3) at ($ ([yshift=\arrspc]net-1-1.east) !0.5! ([yshift=\arrspc]net-1-2.west) $);
  \coordinate (p1) at ($ ([yshift=-\arrspc]net-2-2.west) !0.5! ([yshift=-\arrspc]net-2-1.east) $);
  \coordinate (p2) at ($ ([xshift=-\arrspc]net-2-1.north) !0.5! ([xshift=-\arrspc]net-1-1.south) $);
  \coordinate (p0) at ($ ([xshift=\arrspc]net-1-2.south) !0.5! ([xshift=\arrspc]net-2-2.north) $);

  \newcommand{\idk}{5pt}

  \begin{scope}[rounded corners]
    \node[circle, fill=blue, text=white, inner sep=1pt] (c3) at (p3) {3};
    \node[circle, fill=blue, text=white, inner sep=1pt] (c1) at (p1) {1};
    \node[circle, fill=blue, text=white, inner sep=1pt] (c2) at (p2) {2};
    \node[circle, fill=blue, text=white, inner sep=1pt] (c0) at (p0) {0};

    \draw[->, red, dash pattern=on 5pt off 2pt] (c3) -- ($(net-1-1.north east) + (0,\idk)$) -| ($(net-1-1.south west) + (-\idk,0)$) -- (c2);
    \draw[->, red, dash pattern=on 5pt off 2pt] (c2) -- ($(net-2-1.north west) + (-\idk,0)$) |- ($(net-2-1.south east) + (0,-\idk)$) -- (c1);
    \draw[->, red, dash pattern=on 5pt off 2pt] (c1) -- ($(net-2-2.south west) + (0,-\idk)$) -| ($(net-2-2.north east) + (\idk,0)$) -- (c0);
    \draw[->, red, dash pattern=on 5pt off 2pt] (c0) -- ($(net-1-2.south east) + (\idk,0)$) |- ($(net-1-2.north west) + (0,\idk)$) -- (c3);
  \end{scope}


\end{tikzpicture}
    \vspace{2pt}
    \caption{Deadlock}
    \label{fig:deadlock}
  \end{subfigure}
  \begin{subfigure}[t]{0.1\textwidth}
    \centering
    \vspace{-2.3cm}
    \hspace*{-0.8cm}
    \begin{tikzpicture}[
  font=\footnotesize
]
  \matrix[
    fill=white,
    row sep=5pt,
    column sep=2pt,
    inner sep=0pt,
    column 1/.style={anchor=center},
  ] at (0,0) {
      \node[circle, fill=blue, text=white, inner sep=1pt] {\tiny $x$}; &
      \node[anchor=west] {\tiny packet destined for $x$}; \\

      \node {
      \begin{tikzpicture}
        \draw[->, blue] (0,0) -- +(0.3,0);
      \end{tikzpicture}
      }; &
      \node[anchor=north west, yshift=3pt] {\tiny \shortstack[l]{channel \\[-1pt] \hspace{5pt} inner = low \\[-1pt] \hspace{5pt} outer = high}}; \\

      \node {
      \begin{tikzpicture}
        \draw[->, gray] (0,0) -- +(0.3,0);
      \end{tikzpicture}
      }; &
      \node[anchor=north west, yshift=3pt] {\tiny \shortstack[l]{channel which is\\[-1pt]never routed}}; \\

      \node {
      \begin{tikzpicture}
        \draw[->, red, densely dashed] (0,0) -- +(0.3,0);
      \end{tikzpicture}
      }; &
      \node[anchor=west] {\tiny blocked by}; \\

      \node {
      \begin{tikzpicture}
        \draw[->, green, densely dashed] (0,0) -- +(0.3,0);
      \end{tikzpicture}
      }; &
      \node[anchor=west] {\tiny available hop}; \\
  };
\end{tikzpicture}
  \end{subfigure}
  \begin{subfigure}[t]{0.40\textwidth}
    \centering
    \newcommand{\arrspc}{-3pt}
\newcommand{\arrspcc}{3pt}
\begin{tikzpicture}[
  font=\footnotesize
]
  \matrix (net) [
    column sep=0.7cm,
    row sep=0.7cm,
    matrix of nodes,
    nodes in empty cells,
    nodes={outer sep=0pt,draw,minimum size=0.7cm, inner sep=0pt},
  ]
  {
    \texttt{\texttt{r}$_{\textcolor{blue}{0}}$} & \texttt{\texttt{r}$_{\textcolor{blue}{1}}$} \\
    \texttt{\texttt{r}$_{\textcolor{blue}{3}}$} & \texttt{\texttt{r}$_{\textcolor{blue}{2}}$} \\
  };

  \begin{scope}[rounded corners=0.1cm, blue]
    \draw[<-] ($(net-1-1.east)  + (0, \arrspc)$) -- ($(net-1-2.west)  + (0, \arrspc)$); 
    \draw[<-] ($(net-2-2.west)  + (0,-\arrspc)$) -- ($(net-2-1.east)  + (0,-\arrspc)$); 

    \draw[<-,gray] ($(net-2-1.north) + (-\arrspc,0)$) -- ($(net-1-1.south) + (-\arrspc,0)$); 
    \draw[<-] ($(net-1-2.south) + ( \arrspc,0)$) -- ($(net-2-2.north) + ( \arrspc,0)$); 
  \end{scope}
  \begin{scope}[rounded corners=0.1cm, blue]
    \draw[<-] ($(net-1-1.east)  + (0, \arrspcc)$) -- ($(net-1-2.west)  + (0, \arrspcc)$); 
    \draw[<-,gray] ($(net-2-2.west)  + (0,-\arrspcc)$) -- ($(net-2-1.east)  + (0,-\arrspcc)$); 

    \draw[<-] ($(net-2-1.north) + (-\arrspcc,0)$) -- ($(net-1-1.south) + (-\arrspcc,0)$); 
    \draw[<-] ($(net-1-2.south) + ( \arrspcc,0)$) -- ($(net-2-2.north) + ( \arrspcc,0)$); 
  \end{scope}




  \coordinate (p3) at
    ($ ([yshift=\arrspcc]net-1-1.east) !0.5! ([yshift=\arrspcc]net-1-2.west) $);
  \coordinate (p3l) at
    ($ ([yshift=-\arrspcc]net-1-1.east) !0.5! ([yshift=-\arrspcc]net-1-2.west) $);
  \coordinate (p1) at
    ($ ([yshift=-\arrspc]net-2-2.west) !0.5! ([yshift=-\arrspc]net-2-1.east) $);
  \coordinate (p2) at
    ($ ([xshift=-\arrspcc]net-2-1.north) !0.5! ([xshift=-\arrspcc]net-1-1.south) $);
  \coordinate (p0) at
    ($ ([xshift=\arrspc]net-1-2.south) !0.5! ([xshift=\arrspc]net-2-2.north) $);

  \newcommand{\idk}{5pt}

  \begin{scope}[rounded corners]
    \node[circle, inner sep=1pt, draw=green, densely dashed, fill=green,
    fill opacity=0.2] (c3l) at (p3l) {\phantom{3}};

    \node[circle, fill=blue, text=white, inner sep=1pt] (c3) at (p3) {\hypertarget{virt_channel:3}{3}};
    \node[circle, fill=blue, text=white, inner sep=1pt] (c1) at (p1) {\hypertarget{virt_channel:1}{1}};
    \node[circle, fill=blue, text=white, inner sep=1pt] (c2) at (p2) {\hypertarget{virt_channel:2}{2}};
    \node[circle, fill=blue, text=white, inner sep=1pt] (c0) at (p0) {\hypertarget{virt_channel:0}{0}};

    \draw[->, red, dash pattern=on 5pt off 2pt] (c3) -- ($(net-1-1.north east) + (0,\idk)$) -| ($(net-1-1.south west) + (-\idk,0)$) -- (c2);
    \draw[->, red, dash pattern=on 5pt off 2pt] (c2) -- ($(net-2-1.north west) + (-\idk,0)$) |- ($(net-2-1.south east) + (0,-\idk)$) -- (c1);
    \draw[->, red, dash pattern=on 5pt off 2pt] (c1) -- ($(net-2-2.south west) + (0,-\idk)$) -| ($(net-2-2.north east) + (\idk,0)$) -- (c0);
    \draw[->, green, dash pattern=on 5pt off 2pt] (c0) -- ($(net-1-2.south west) - (\idk,\idk)$) -- (c3l);

  \end{scope}

\end{tikzpicture}
    \vspace{2pt}
    \caption{Virtual channels}
    \label{fig:virt_chan}
  \end{subfigure}
  \caption{
    Because the wrap-around creates a cycle,
    a deadlock might occur in which every packet
    is blocked by another (a).
    Virtual channels resolve this issue by breaking
    the dependency cycle (b).
  }
  \label{fig:virtual_channels}
\end{figure}

\subsection{Specification}

Due to the generalization over the number and size of the dimensions,
most internal definitions are conceptually constructed from two nested recursions.
The outer recursion over all dimensions and
the inner recursion over the routers in a single dimension.

\subsubsection{Virtual Channels}\label{sec:virt_chan}

Our \noc architecture employs wrap-around links, which makes it
vulnerable to the classic deadlock issue. \Cref{fig:deadlock}
illustrates a deadlock scenario in a single-dimensional \noc consisting
of four routers. In this arrangement every channel is occupied, forming a
circular dependency that prevents any packet from advancing.

To address the deadlock problem, we adopted the methodology of
\emph{virtual channels}~\cite{virtualchannels}. For that, each channel
is duplicated into a \emph{high} and a \emph{low} virtual channel.
Then the routing algorithm is adapted such that packets with a destination
higher than their current router are forwarded via the high channels. All other
packets are routed via the low channels. Phrased differently, packets which need
to wrap around are routed on the high channels until they wrap. Then they continue
on the low channels together with the other packets.
Intuitively, this breaks the dependency cycle because low channels never depend on
high channels.

An application of this methodology to the previous deadlock example is shown
in \cref{fig:virt_chan}. The high channels are depicted as the outer ring while
the low channels form the inner ring. Notably, two channels are grayed-out
as they can never be selected by the adapted routing algorithm. In this
specific example the deadlock is resolved as packet
\hyperlink{virt_channel:0}{%
\tikz[baseline=(char.base)]{
  \node[circle, fill=blue, text=white, inner sep=1pt] (char) {\small 0};
}%
}
does not depend on packet
\hyperlink{virt_channel:3}{%
\tikz[baseline=(char.base)]{
  \node[circle, fill=blue, text=white, inner sep=1pt] (char) {\small 3};
}%
}
anymore.

\subsubsection{Routing}

\begin{figure}
  \begin{subfigure}[b]{0.48\textwidth}
    \centering
\newcommand{\rspace}{1cm}
\newcommand{\spacing}{2.3cm}
\newcommand{\len}{2.2cm}
\newcommand{\debug}{none}

\newcommand{\spc}{6.5pt}
\newcommand{\wspc}{5pt}

\begin{tikzpicture}[every node/.style={draw,rectangle}]

  \node[thick, minimum height=2.5cm, minimum width=2.5cm] (r) at (0,0) {};

  \node[fill=white, inner sep=0pt, minimum width=1.0cm, minimum height=9pt, text=gray, right=-0.1pt of r.north] (arbt) {\footnotesize{\texttt{arb.}}};
  \node[fill=white, inner sep=0pt, minimum width=1.0cm, minimum height=9pt, text=blue, fill=blue!5, left=-0.1pt of r.north] (rout1) {\footnotesize{\texttt{rout.}}};

  \node[fill=white, inner sep=0pt, minimum width=1.0cm, minimum height=9pt, text=gray, right=-0.1pt of r.south] (rout2) {\footnotesize{\texttt{rout.}}};
  \node[fill=white, inner sep=0pt, minimum width=1.0cm, minimum height=9pt, text=blue, fill=blue!5, left=-0.1pt of r.south] (arbb) {\footnotesize{\texttt{arb.}}};

  \node[fill=white, inner sep=0pt, minimum width=1.0cm, minimum height=9pt, text=gray, left=0pt of r.west, rotate=90] (routl) {\footnotesize{\texttt{rout.}}};
  \node[fill=white, inner sep=0pt, minimum width=1.0cm, minimum height=9pt, text=gray, right=0pt of r.west, rotate=90] (arbl) {\footnotesize{\texttt{arb.}}};

  \node[fill=white, inner sep=0pt, minimum width=1.0cm, minimum height=9pt, text=gray, right=0pt of r.east, rotate=90] (routr) {\footnotesize{\texttt{rout.}}};
  \node[fill=white, inner sep=0pt, minimum width=1.0cm, minimum height=9pt, text=gray, left=0pt of r.east, rotate=90] (arbr) {\footnotesize{\texttt{arb.}}};

  \node[draw=none, text=blue, inner sep=1pt, above=0.8cm of rout1] (downt) {\texttt{down$_{n+1}$}};
  \draw[latex-, blue, thick] (rout1) -- (downt);
  \draw[dashed, thick, gray] ($(r.north) + (-2cm,0.9cm)$) -- ($(r.north) + (2cm,0.9cm)$);
  \node[draw=none, text=gray, inner sep=1pt, above=0.3cm of arbt] (upt) {\texttt{up$_n$}};
  \draw[-latex, gray] (arbt) -- (upt);

  \node[draw=none, text=blue, inner sep=1pt, below=0.3cm of arbb] (downb) {\texttt{down$_n$}};
  \draw[-latex, blue, thick] (arbb) -- (downb);

  \node[draw=none, text=gray, inner sep=1pt, right=0.3cm of arbr.south] (upr) {\texttt{up$_{m}$}};
  \draw[latex-, gray] (upr) -- (arbr);

  \node[draw=none, text=gray, inner sep=1pt, left=0.3cm of arbl.north] (downl) {\texttt{down$_{m}$}};
  \draw[latex-, gray] (downl) -- (arbl);

  \draw[-latex, gray, rounded corners=2pt] ($(rout1.south) + {-1.5}*(\wspc,0)$) -- ($(rout1.south) + {-1.5}*(\wspc,0) + (0,-\spc)$) -- ($(arbl.south) + {+1.5}*(0,\wspc) + ( \spc,0)$) -- ($(arbl.south) + {+1.5}*(0,\wspc)$);
  \draw[-latex, gray, rounded corners=2pt] ($(rout1.south) + { 0.5}*(\wspc,0)$) -- ($(rout1.south) + { 0.5}*(\wspc,0) + (0,-\spc)$) -- ($(arbr.north) + { 0.5}*(0,\wspc) + (-\spc,0)$) -- ($(arbr.north) + { 0.5}*(0,\wspc)$);
  \draw[-latex, gray, rounded corners=2pt] ($(rout1.south) + { 1.5}*(\wspc,0)$) -- ($(rout1.south) + { 1.5}*(\wspc,0) + (0,-\spc)$) -- ($(arbt.south) + {-1.5}*(\wspc,0) + (0,-\spc)$) -- ($(arbt.south) + {-1.5}*(\wspc,0)$);

  \draw[-latex, gray, rounded corners=2pt] ($(routl.south) + { 1.5}*(0,\wspc)$) -- ($(routl.south) + { 1.5}*(0,\wspc) + ( \spc,0)$) -- ($(arbl.south) + {-1.5}*(0,\wspc) + ( \spc,0)$) -- ($(arbl.south) + {-1.5}*(0,\wspc)$);
  \draw[-latex, gray, rounded corners=2pt] ($(routl.south) + {-1.5}*(0,\wspc)$) -- ($(routl.south) + {-1.5}*(0,\wspc) + ( \spc,0)$) -- ($(arbb.north) + {-1.5}*(\wspc,0) + (0, \spc)$) -- ($(arbb.north) + {-1.5}*(\wspc,0)$);
  \draw[-latex, gray, rounded corners=2pt] ($(routl.south) + {-0.5}*(0,\wspc)$) -- ($(routl.south) + {-0.5}*(0,\wspc) + ( \spc,0)$) -- ($(arbr.north) + {-0.5}*(0,\wspc) + (-\spc,0)$) -- ($(arbr.north) + {-0.5}*(0,\wspc)$);
  \draw[-latex, gray, rounded corners=2pt] ($(routl.south) + { 0.5}*(0,\wspc)$) -- ($(routl.south) + { 0.5}*(0,\wspc) + ( \spc,0)$) -- ($(arbt.south) + {-0.5}*(\wspc,0) + (0,-\spc)$) -- ($(arbt.south) + {-0.5}*(\wspc,0)$);

  \draw[-latex, gray, rounded corners=2pt] ($(rout2.north) + {-0.5}*(\wspc,0)$) -- ($(rout2.north) + {-0.5}*(\wspc,0) + (0, \spc)$) -- ($(arbl.south) + {-0.5}*(0,\wspc) + ( \spc,0)$) -- ($(arbl.south) + {-0.5}*(0,\wspc)$);
  \draw[-latex, gray, rounded corners=2pt] ($(rout2.north) + {-1.5}*(\wspc,0)$) -- ($(rout2.north) + {-1.5}*(\wspc,0) + (0, \spc)$) -- ($(arbb.north) + { 1.5}*(\wspc,0) + (0, \spc)$) -- ($(arbb.north) + { 1.5}*(\wspc,0)$);
  \draw[-latex, gray, rounded corners=2pt] ($(rout2.north) + { 1.5}*(\wspc,0)$) -- ($(rout2.north) + { 1.5}*(\wspc,0) + (0, \spc)$) -- ($(arbr.north) + {-1.5}*(0,\wspc) + (-\spc,0)$) -- ($(arbr.north) + {-1.5}*(0,\wspc)$);
  \draw[-latex, gray, rounded corners=2pt] ($(rout2.north) + { 0.5}*(\wspc,0)$) -- ($(rout2.north) + { 0.5}*(\wspc,0) + (0, \spc)$) -- ($(arbt.south) + {+0.5}*(\wspc,0) + (0,-\spc)$) -- ($(arbt.south) + {+0.5}*(\wspc,0)$);

  \draw[-latex, gray, rounded corners=2pt] ($(routr.north) + { 0.5}*(0,\wspc)$) -- ($(routr.north) + { 0.5}*(0,\wspc) + (-\spc,0)$) -- ($(arbl.south) + { 0.5}*(0,\wspc) + ( \spc,0)$) -- ($(arbl.south) + { 0.5}*(0,\wspc)$);
  \draw[-latex, gray, rounded corners=2pt] ($(routr.north) + {-0.5}*(0,\wspc)$) -- ($(routr.north) + {-0.5}*(0,\wspc) + (-\spc,0)$) -- ($(arbb.north) + { 0.5}*(\wspc,0) + (0, \spc)$) -- ($(arbb.north) + { 0.5}*(\wspc,0)$);
  \draw[-latex, gray, rounded corners=2pt] ($(routr.north) + {-1.5}*(0,\wspc)$) -- ($(routr.north) + {-1.5}*(0,\wspc) + (-\spc,0)$) -- ($(arbr.north) + { 1.5}*(0,\wspc) + (-\spc,0)$) -- ($(arbr.north) + { 1.5}*(0,\wspc)$);
  \draw[-latex, gray, rounded corners=2pt] ($(routr.north) + { 1.5}*(0,\wspc)$) -- ($(routr.north) + { 1.5}*(0,\wspc) + (-\spc,0)$) -- ($(arbt.south) + {+1.5}*(\wspc,0) + (0,-\spc)$) -- ($(arbt.south) + {+1.5}*(\wspc,0)$);

  \draw[-latex, blue, thick, rounded corners=2pt] ($(rout1.south) + {-0.5}*(\wspc,0)$) -- ($(rout1.south) + {-0.5}*(\wspc,0) + (0,-\spc)$) -- ($(arbb.north) + {-0.5}*(\wspc,0) + (0, \spc)$) -- ($(arbb.north) + {-0.5}*(\wspc,0)$);


  %

  %


  \node[draw=none, inner sep=1pt, above right=0pt of r.north east, fill=white] {\texttt{r}$_{\textcolor{blue}{n};\textcolor{red}{m}}$};

\end{tikzpicture}
    \caption{routing through a single dimension}
    \label{fig:routing_1}
  \end{subfigure}
~
  \begin{subfigure}[b]{0.45\textwidth}
    \centering
\newcommand{\rspace}{1cm}
\newcommand{\spacing}{2.3cm}
\newcommand{\len}{2.2cm}
\newcommand{\debug}{none}

\newcommand{\spc}{6.5pt}
\newcommand{\wspc}{5pt}

\begin{tikzpicture}[every node/.style={draw,rectangle}]

  \node[thick, minimum height=2.5cm, minimum width=2.5cm] (r) at (0,0) {};

  \node[fill=white, inner sep=0pt, minimum width=1.0cm, minimum height=9pt, text=gray, right=-0.1pt of r.north] (arbt) {\footnotesize{\texttt{arb.}}};
  \node[fill=white, inner sep=0pt, minimum width=1.0cm, minimum height=9pt, text=blue, fill=blue!5, left=-0.1pt of r.north] (rout1) {\footnotesize{\texttt{rout.}}};

  \node[fill=white, inner sep=0pt, minimum width=1.0cm, minimum height=9pt, text=gray, right=-0.1pt of r.south] (rout2) {\footnotesize{\texttt{rout.}}};
  \node[fill=white, inner sep=0pt, minimum width=1.0cm, minimum height=9pt, text=gray, left=-0.1pt of r.south] (arbb) {\footnotesize{\texttt{arb.}}};

  \node[fill=white, inner sep=0pt, minimum width=1.0cm, minimum height=9pt, text=gray, left=0pt of r.west, rotate=90] (routl) {\footnotesize{\texttt{rout.}}};
  \node[fill=white, inner sep=0pt, minimum width=1.0cm, minimum height=9pt, text=gray, right=0pt of r.west, rotate=90] (arbl) {\footnotesize{\texttt{arb.}}};

  \node[fill=white, inner sep=0pt, minimum width=1.0cm, minimum height=9pt, text=gray, right=0pt of r.east, rotate=90] (routr) {\footnotesize{\texttt{rout.}}};
  \node[fill=white, inner sep=0pt, minimum width=1.0cm, minimum height=9pt, text=red, fill=red!5, left=0pt of r.east, rotate=90] (arbr) {\footnotesize{\texttt{arb.}}};

  \node[draw=none, text=blue, inner sep=1pt, above=0.8cm of rout1] (downt) {\texttt{down$_{n+1}$}};
  \draw[latex-, blue, thick] (rout1) -- (downt);
  \draw[dashed, thick, gray] ($(r.north) + (-2cm,0.9cm)$) -- ($(r.north) + (2cm,0.9cm)$);
  \node[draw=none, text=gray, inner sep=1pt, above=0.3cm of arbt] (upt) {\texttt{up$_n$}};
  \draw[-latex, gray] (arbt) -- (upt);

  \node[draw=none, text=gray, inner sep=1pt, below=0.3cm of arbb] (downb) {\texttt{down$_n$}};
  \draw[-latex, gray] (arbb) -- (downb);

  \node[draw=none, text=red, inner sep=1pt, right=0.3cm of arbr.south] (upr) {\texttt{up$_{m}$}};
  \draw[latex-, red, thick] (upr) -- (arbr);

  \node[draw=none, text=gray, inner sep=1pt, left=0.3cm of arbl.north] (downl) {\texttt{down$_{m}$}};
  \draw[latex-, gray] (downl) -- (arbl);

  \draw[-latex, gray, rounded corners=2pt] ($(rout1.south) + {-1.5}*(\wspc,0)$) -- ($(rout1.south) + {-1.5}*(\wspc,0) + (0,-\spc)$) -- ($(arbl.south) + {+1.5}*(0,\wspc) + ( \spc,0)$) -- ($(arbl.south) + {+1.5}*(0,\wspc)$);
  \draw[-latex, gray, rounded corners=2pt] ($(rout1.south) + {-0.5}*(\wspc,0)$) -- ($(rout1.south) + {-0.5}*(\wspc,0) + (0,-\spc)$) -- ($(arbb.north) + {-0.5}*(\wspc,0) + (0, \spc)$) -- ($(arbb.north) + {-0.5}*(\wspc,0)$);
  \draw[-latex, gray, rounded corners=2pt] ($(rout1.south) + { 1.5}*(\wspc,0)$) -- ($(rout1.south) + { 1.5}*(\wspc,0) + (0,-\spc)$) -- ($(arbt.south) + {-1.5}*(\wspc,0) + (0,-\spc)$) -- ($(arbt.south) + {-1.5}*(\wspc,0)$);

  \draw[-latex, gray, rounded corners=2pt] ($(routl.south) + { 1.5}*(0,\wspc)$) -- ($(routl.south) + { 1.5}*(0,\wspc) + ( \spc,0)$) -- ($(arbl.south) + {-1.5}*(0,\wspc) + ( \spc,0)$) -- ($(arbl.south) + {-1.5}*(0,\wspc)$);
  \draw[-latex, gray, rounded corners=2pt] ($(routl.south) + {-1.5}*(0,\wspc)$) -- ($(routl.south) + {-1.5}*(0,\wspc) + ( \spc,0)$) -- ($(arbb.north) + {-1.5}*(\wspc,0) + (0, \spc)$) -- ($(arbb.north) + {-1.5}*(\wspc,0)$);
  \draw[-latex, gray, rounded corners=2pt] ($(routl.south) + {-0.5}*(0,\wspc)$) -- ($(routl.south) + {-0.5}*(0,\wspc) + ( \spc,0)$) -- ($(arbr.north) + {-0.5}*(0,\wspc) + (-\spc,0)$) -- ($(arbr.north) + {-0.5}*(0,\wspc)$);
  \draw[-latex, gray, rounded corners=2pt] ($(routl.south) + { 0.5}*(0,\wspc)$) -- ($(routl.south) + { 0.5}*(0,\wspc) + ( \spc,0)$) -- ($(arbt.south) + {-0.5}*(\wspc,0) + (0,-\spc)$) -- ($(arbt.south) + {-0.5}*(\wspc,0)$);

  \draw[-latex, gray, rounded corners=2pt] ($(rout2.north) + {-0.5}*(\wspc,0)$) -- ($(rout2.north) + {-0.5}*(\wspc,0) + (0, \spc)$) -- ($(arbl.south) + {-0.5}*(0,\wspc) + ( \spc,0)$) -- ($(arbl.south) + {-0.5}*(0,\wspc)$);
  \draw[-latex, gray, rounded corners=2pt] ($(rout2.north) + {-1.5}*(\wspc,0)$) -- ($(rout2.north) + {-1.5}*(\wspc,0) + (0, \spc)$) -- ($(arbb.north) + { 1.5}*(\wspc,0) + (0, \spc)$) -- ($(arbb.north) + { 1.5}*(\wspc,0)$);
  \draw[-latex, gray, rounded corners=2pt] ($(rout2.north) + { 1.5}*(\wspc,0)$) -- ($(rout2.north) + { 1.5}*(\wspc,0) + (0, \spc)$) -- ($(arbr.north) + {-1.5}*(0,\wspc) + (-\spc,0)$) -- ($(arbr.north) + {-1.5}*(0,\wspc)$);
  \draw[-latex, gray, rounded corners=2pt] ($(rout2.north) + { 0.5}*(\wspc,0)$) -- ($(rout2.north) + { 0.5}*(\wspc,0) + (0, \spc)$) -- ($(arbt.south) + {+0.5}*(\wspc,0) + (0,-\spc)$) -- ($(arbt.south) + {+0.5}*(\wspc,0)$);

  \draw[-latex, gray, rounded corners=2pt] ($(routr.north) + { 0.5}*(0,\wspc)$) -- ($(routr.north) + { 0.5}*(0,\wspc) + (-\spc,0)$) -- ($(arbl.south) + { 0.5}*(0,\wspc) + ( \spc,0)$) -- ($(arbl.south) + { 0.5}*(0,\wspc)$);
  \draw[-latex, gray, rounded corners=2pt] ($(routr.north) + {-0.5}*(0,\wspc)$) -- ($(routr.north) + {-0.5}*(0,\wspc) + (-\spc,0)$) -- ($(arbb.north) + { 0.5}*(\wspc,0) + (0, \spc)$) -- ($(arbb.north) + { 0.5}*(\wspc,0)$);
  \draw[-latex, gray, rounded corners=2pt] ($(routr.north) + {-1.5}*(0,\wspc)$) -- ($(routr.north) + {-1.5}*(0,\wspc) + (-\spc,0)$) -- ($(arbr.north) + { 1.5}*(0,\wspc) + (-\spc,0)$) -- ($(arbr.north) + { 1.5}*(0,\wspc)$);
  \draw[-latex, gray, rounded corners=2pt] ($(routr.north) + { 1.5}*(0,\wspc)$) -- ($(routr.north) + { 1.5}*(0,\wspc) + (-\spc,0)$) -- ($(arbt.south) + {+1.5}*(\wspc,0) + (0,-\spc)$) -- ($(arbt.south) + {+1.5}*(\wspc,0)$);

  \coordinate (pstart) at ($(rout1.south) + {0.5}*(\wspc,0) + (0,-\spc)$);
  \coordinate (pend)   at ($(arbr.north) + {0.5}*(0,\wspc) + (-\spc,0)$);
  \coordinate (midA) at ($(pstart)!0.2cm!(pend)$);
  \coordinate (midB) at ($(pend)!0.2cm!(pstart)$);
  \coordinate (midAA) at ($(pstart)!0.21cm!(pend)$);
  \coordinate (midBB) at ($(pend)!0.21cm!(pstart)$);

  \draw[thick, blue, rounded corners=2pt] ($(rout1.south) + {0.5}*(\wspc,0)$) -- (pstart) -- (midAA);
  \fadeline{midA}{midB}{top color=blue, bottom color=red, minimum height=0.8pt}
  \draw[-latex, thick, red, rounded corners=2pt] (midBB) -- (pend) -- ($(arbr.north) + {0.5}*(0,\wspc)$);

  %

  %


  \node[draw=none, inner sep=1pt, above right=0pt of r.north east, fill=white] {\texttt{r}$_{\textcolor{blue}{n};\textcolor{red}{m}}$};

\end{tikzpicture}
    \caption{routing into the next dimension}
    \label{fig:routing_2}
  \end{subfigure}
  \caption{Multi-dimensional routing.}
  \label{fig:routing}
\end{figure}
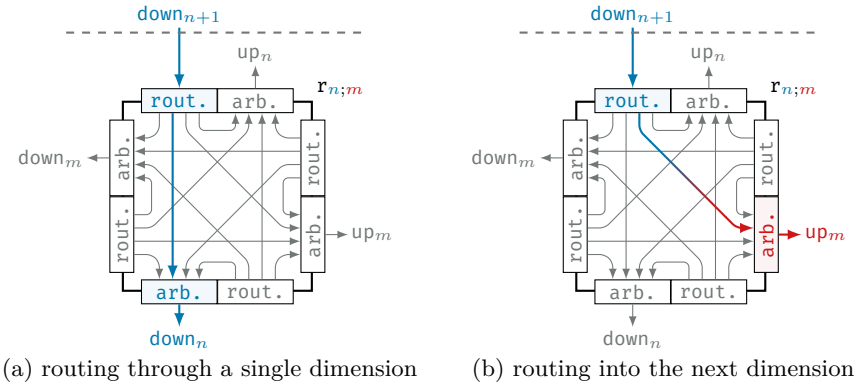

The \texttt{routing} strategy of our design is a generalization of
XY-routing to an arbitrary number of dimensions, called
dimension-order routing~\cite{dally2004principles}.
At first, we route a message to its destination address within a
single dimension (see \cref{fig:routing_1}).
\begin{lstlisting}[style=frame,mathescape]
Definition routing_spec$_1$ lcl dst : §3:RoutingDecision§ :=
  if lcl <|=?|> dst then §4:Arrived§
  else if dist lcl dst <? dist dst lcl
    then if dst <|<?|> lcl then §4:UpHigh§   else §4:UpLow§
    else if lcl <|<?|> dst then §4:DownHigh§ else §4:DownLow§.
\end{lstlisting}

A message only has two choices, either it travels \ilstrocq{§4:Up§}- or
the \ilstrocq{§4:Down§}-wards.
Since the last router of a dimension wraps around, i.e., connects to
the first, both directions would ultimately lead to the message's destination.
However, to minimize the number of steps, we select the shorter path.
Additionally, to prevent the deadlock discussed in the previous section (\ref{sec:virt_chan}),
the routing either selects the \ilstrocq{§4:High§} or the \ilstrocq{§4:Low} virtual channel based on the index.

Then for multidimensional routing, we route the dimensions one after another.
\begin{lstlisting}[style=frame]
Fixpoint routing_spec {dims} lcl dst : §3:vect§ §3:RoutingDecision§ (|dims| + 1) :=
  match dims with
  | []           => §1:λ§  _          _        , [§4:Core§]
  | dim :: dims' => §1:λ§ '(l, lcl') '(d, dst'),
    let r1 := §3:routing_spec$_1$§ l d in
    if r1 =?b §4:Arrived§
      then r1 :: (§3:routing_spec§ lcl' dst')
      else r1 :: (§3:vect_const§ §4:Irrelevant§)
  end lcl dst.
\end{lstlisting}

We start routing at dimension 0.
Once \ilstrocq{§4:Arrived§} at the dimension's destination, we advance to the
next dimension until we arrive at the final destination in dimension $k$,
where the packet is directly forwarded to the connected \ilstrocq{§4:Core§}.
Since only a single dimension can be routed at a time, all lower-priority
dimensions are considered \ilstrocq{§4:Irrelevant§}.
In \cref{fig:routing_2}, router
\texttt{r}$_{\textcolor{blue}{n};\textcolor{red}{m};...;\textcolor{green}{z}}$
forwards the message from the \textcolor{blue}{current} dimension
into the \textcolor{red}{next}.
To make sure that this \ilstrocq{§3:routing_spec§} is sensible, we prove the following
progress lemma.
\begin{lemma}[Routing progress]\label{lem:route:progress}
  Given any local index $lcl$ and destination index $dst$ in the dimensional space $c$ (of the \noc)
  the \ilstrocq{§3:routing_spec§} either selects the local output (in case $lcl = dst$) or it selects
  a next index $nxt$ to route a message to,
  such that:
  $$(\texttt{dist}_{man} ~ nxt ~ dst) + 1 = \texttt{dist}_{man} ~ lcl ~  dst$$
  Here \texttt{dist}$_{man}$ denotes the Manhattan Distance.
\end{lemma}
\begin{proof}
  The proof is by induction on $c$, i.e., the configuration (\texttt{dims}) of the \noc.

  \textbf{Base case $c = \texttt{[]}$}. The conclusion holds trivially. Since the type of both $lcl$ and $dst$ reduces to \ilstrocq{§3:unit§}, they must be equal.
  Similarly, by definition \ilstrocq{§3:routing_spec§} selects the local output (\ilstrocq{§4:Core§}).

  \textbf{Induction step $c = dim \texttt{::} c' $}. Since there is at least one dimension we know that we can split the indices $lcl$ and $dst$ into
  a sub-index of the first dimension ($l$/$d$) and their remainder ($lcl'$/$dst'$).
  We then continue by a case distinction on the routing of the first dimension $\ilstrocq{§3:routing_spec§}_1~l~d$.

  \textbf{Case Arrived}. Routing continues in the remaining dimensions, thus the conclusion follows from the induction hypothesis.

  \textbf{Other cases}. Routing chooses the current dimension (all others are ignored). Since only a single dimension is routed
  either \ilstrocq{§4:Up§} or \ilstrocq{§4:Down§}, distance must be exactly 1 less.
\end{proof}

\subsubsection{Arbiter}

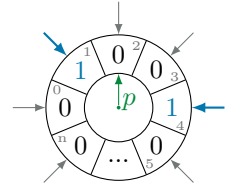
\begin{wrapfigure}{r}{0.25\textwidth}
  \centering
  \vspace{-2.2em}
  \begin{tikzpicture}[every node/.style={draw,rectangle},on grid,
    disabled/.style={text=black},
    enabled/.style={text=blue},
    enabled_arr/.style={blue, thick}]
  \begin{scope}[node distance=0.707106cm]
    \node[circle, fill=green, inner sep=0.5pt, draw=none] (a) at (0,0) {};
    \node[circle, inner sep=0.68cm] at (0,0) {};
    \node[circle, inner sep=0.32cm] at (0,0) {};

    \node[circle, draw=none, minimum height=0.5cm, minimum width=0.5cm, inner sep=0pt, disabled, left=of a       ] (n1) {0}; 
    \node[circle, draw=none, minimum height=0.5cm, minimum width=0.5cm, inner sep=0pt,  enabled, above left=of a ] (n5) {1}; 
    \node[circle, draw=none, minimum height=0.5cm, minimum width=0.5cm, inner sep=0pt, disabled, above=of a      ] (n3) {0}; 
    \node[circle, draw=none, minimum height=0.5cm, minimum width=0.5cm, inner sep=0pt, disabled, above right=of a] (n6) {0}; 
    \node[circle, draw=none, minimum height=0.5cm, minimum width=0.5cm, inner sep=0pt,  enabled, right=of a      ] (n2) {1}; 
    \node[circle, draw=none, minimum height=0.5cm, minimum width=0.5cm, inner sep=0pt, disabled, below right=of a] (n8) {0}; 
    \node[circle, draw=none, minimum height=0.5cm, minimum width=0.5cm, inner sep=0pt, disabled, below=of a      ] (n4) {...}; 
    \node[circle, draw=none, minimum height=0.5cm, minimum width=0.5cm, inner sep=0pt, disabled, below left=of a ] (n7) {0}; 

    \node[draw=none, text=gray, minimum height=0.5cm, minimum width=0.5cm, inner sep=0pt] at ($(a)!0.85cm!-16:(n1)$) {\tiny 0};
    \node[draw=none, text=gray, minimum height=0.5cm, minimum width=0.5cm, inner sep=0pt] at ($(a)!0.85cm!-16:(n5)$) {\tiny 1};
    \node[draw=none, text=gray, minimum height=0.5cm, minimum width=0.5cm, inner sep=0pt] at ($(a)!0.85cm!-16:(n3)$) {\tiny 2};
    \node[draw=none, text=gray, minimum height=0.5cm, minimum width=0.5cm, inner sep=0pt] at ($(a)!0.85cm!-16:(n6)$) {\tiny 3};
    \node[draw=none, text=gray, minimum height=0.5cm, minimum width=0.5cm, inner sep=0pt] at ($(a)!0.85cm!-16:(n2)$) {\tiny 4};
    \node[draw=none, text=gray, minimum height=0.5cm, minimum width=0.5cm, inner sep=0pt] at ($(a)!0.85cm!-16:(n8)$) {\tiny 5};
    \node[draw=none, text=gray, minimum height=0.5cm, minimum width=0.5cm, inner sep=0pt] at ($(a)!0.85cm!-16:(n7)$) {\tiny n};

    \draw[latex-, gray             ] (n1) -- ($(a)!1.4cm!(n1)$);
    \draw[latex-, gray, enabled_arr] (n2) -- ($(a)!1.4cm!(n2)$);
    \draw[latex-, gray             ] (n3) -- ($(a)!1.4cm!(n3)$);
    \draw[latex-, gray, enabled_arr] (n5) -- ($(a)!1.4cm!(n5)$);
    \draw[latex-, gray             ] (n6) -- ($(a)!1.4cm!(n6)$);
    \draw[latex-, gray             ] (n7) -- ($(a)!1.4cm!(n7)$);
    \draw[latex-, gray             ] (n8) -- ($(a)!1.4cm!(n8)$);

    \draw ($(a)!0.4571cm!22.5:(n1)$) -- ($(a)!0.9571cm!22.5:(n1)$);
    \draw ($(a)!0.4571cm!22.5:(n2)$) -- ($(a)!0.9571cm!22.5:(n2)$);
    \draw ($(a)!0.4571cm!22.5:(n3)$) -- ($(a)!0.9571cm!22.5:(n3)$);
    \draw ($(a)!0.4571cm!22.5:(n4)$) -- ($(a)!0.9571cm!22.5:(n4)$);
    \draw ($(a)!0.4571cm!22.5:(n5)$) -- ($(a)!0.9571cm!22.5:(n5)$);
    \draw ($(a)!0.4571cm!22.5:(n6)$) -- ($(a)!0.9571cm!22.5:(n6)$);
    \draw ($(a)!0.4571cm!22.5:(n7)$) -- ($(a)!0.9571cm!22.5:(n7)$);
    \draw ($(a)!0.4571cm!22.5:(n8)$) -- ($(a)!0.9571cm!22.5:(n8)$);

    \draw[-latex, green] (a) -- (n3);
    \node[draw=none, text=green, above right=2pt and 4pt of a] {$p$};
  \end{scope}
\end{tikzpicture}
  \caption{Arbiter ring}
  \vspace{-1em}
  \label{fig:arbiter}
\end{wrapfigure}
Due to the static nature of hardware, \texttt{arbitration}
gets the same number of requests from the \texttt{routing} units each
cycle.
We implement round-robin arbitration~\cite{dally2004principles}
which, conceptually, orders all requests into a ring
(see \cref{fig:arbiter}).
The \textcolor{green}{priority pointer $p$} determines the order in which
requests are granted and thereby forwarded along the outgoing channel.
For that, \texttt{arbitration} always selects the request that is clockwise-closest
to $p$ and has a message to forward (\textcolor{blue}{1}).
In the case of \cref{fig:arbiter}, the next request to be forwarded is
retrieved from incoming channel at arbiter ring index $4$.
Then, to guarantee progress for all requests, $p$ advances to the slot
after the selected one, here 5.
\begin{lstlisting}[label=fig:arb_spec, style=frame, style=coq]
Definition arbiter_spec n p reqs : §2:option§ (§2:Fin.t§ n) :=
  let selected := §3:first_idx§ (§3:rotate§ p reqs)
  in §3:option_map§ (§1:λ§ idx, p + idx) selected.
\end{lstlisting}

We use \ilstrocq{§3:rotate§} to order the list of requests by their priority
and \ilstrocq{§3:first_idx} to select the highest-priority one.
In terms of \cref{fig:arbiter}, this rotation essentially relabels all requests
s.t. $p$ points to index 0.
Then \ilstrocq{§3:first_idx} performs the actual selection.
Finally, to reestablish the correct label, we add the value of $p$ back to the
calculated index.
Just as with our routing specification, we also prove that
our arbitration algorithm makes progress.
\begin{lemma}[Arbiter progress]\label{lem:arb:progress}
  For any arbiter ring $a$ of size $n$ with priority at $p$,
  given there is a request at $i$
  then the \ilstrocq{§3:arbiter_spec§} selects the message at $i$ for transmission
  or $p$ makes progress towards $i$.
\end{lemma}
\begin{proof}
  Let $i_s$ be the selected index, then by case distinction:

\textbf{Case $i_s = i$}: Conclusion follows immediately.

\textbf{Case $i_s \neq i$}:
We show that $p_{new} := i_s + 1$ is \emph{closer} to $i$ than $p$ (i.e. it has a smaller difference).
As \ilstrocq{§3:first_idx§} always picks the smallest set index
we know that $i_s - p < i - p$. Then
\begin{equation*}
\begin{aligned}
  i_s - p                 & < i - p  & \pmod{n} & \implies \\
  (i - p) - (i_s - p + 1) & < i - p  & \pmod{n} & \iff     \\
  i - (i_s + 1)           & < i - p  & \pmod{n} & \iff     \\
  i - p_{new}             & < i - p
\end{aligned}
\end{equation*}
(All comparisons modulo \(n\) are performed on the canonical residue representatives in \(\{0,\dots,n-1\}\))
\end{proof}

\subsubsection{\noc}

The overall state of the \noc is then a record that stores
the contents of all channels and
the priorities of all arbiters.
For stating the liveness property, we define
\begin{equation*}
  \texttt{eventually} ~ P := \lambda ~ s, ~ \exists ~ n, P ~ (\texttt{ns\_step} ~ n ~ s)
\end{equation*}
where \texttt{ns\_step} $n$ $s$ executes the formal model of the \noc for $n$ cycles starting from state $s$.
\begin{axiom}[Output liveness]\label{ax:noc:output}
  Every core connected to the \noc must handle arriving packets immediately.
\end{axiom}
This axiom makes sure that the network cannot be blocked by a core.
\begin{theorem}[Deadlock freedom]\label{theo:noc:deadlockfree}
  Any port $p_i$ of any router $r_{lcl}$ will \texttt{eventually} become available.
\end{theorem}
\begin{proof}
  This proof is intricate as it consists of multiple nested inductions.
  To ease its comprehension, we only present the underlying intuition
  here and refer to our artifact for more details.
  This intuition proceeds in reverse order of the actual proof. It starts at
  \cref{ax:noc:output} and describes how this guarantee propagates
  through the whole network.

  From the axiom we already know that all output ports are always available.
  Likewise, by \hyperref[lem:arb:progress]{\textsc{Arbiter progress}}~(\ref{lem:arb:progress}) a packet
  whose target port is infinitely often available is guaranteed to be forwarded eventually.
  Thus, any packet whose next hop is an output port will
  always eventually take a step, thereby freeing its own port.

  From that we derive that all packets of the lowest dimension also release their port eventually.
  This follows from an induction on the position of a packet in the \emph{channel dependency order}. This order
  describes if a channel could be blocked by another. The base case consists of
  the last low virtual channel, which only ever routes to its direct output port, thus following
  from the argument above.
  The induction step instead utilizes the induction hypothesis together with \hyperref[lem:arb:progress]{\textsc{Arbiter progress}}.
  Any packet could only ever be blocked by a packet within a lower-rank channel, thus its next hop must
  become available, and it must eventually be selected by the arbiter.

  Lastly, we propagate this guarantee to higher dimensions by strong induction on the number of dimensions. The base
  case of the lowest dimension was already described above. The induction step follows the same structure for
  reasoning about packets within a single dimension, except that the last virtual channel may forward either
  to an output or just a lower dimension, which then requires the induction hypothesis.

\end{proof}

\begin{theorem}[\noc liveness]\label{theo:noc:liveness}
  Any message $m$ in the \noc at any input port $i$ of any router $r_{lcl}$ will \texttt{eventually} be delivered to its destination $m.dst$.
\end{theorem}
\begin{proof}
  Let $d := \texttt{dist}~lcl~m.dst$ denote the number of hops between $lcl$ and $m.dst$.
  Then by induction over $d$:

  \textbf{Base case: $d = 0$}. It follows that $lcl = m.dst$ satisfying the conclusion trivially.

  \textbf{Induction step: $d = 1 + d'$}. We know that there exists a next hop $nxt$.
  Thus, we show that eventually $m$ reaches router $r_{nxt}$ and conclude by the induction hypothesis.
  To establish the single step we do a strong induction over the priority of input $i$ on the port that leads from $src$ to $nxt$.
  Let $p$ be the priority pointer at that port.

  \textbf{Base case: $i - p = 0$}. As $i$ has the highest priority, it will be forwarded immediately once the output port is available,
  which must happen eventually by \hyperref[theo:noc:deadlockfree]{\textsc{Deadlock freedom}}~(\ref{theo:noc:deadlockfree}).

  \textbf{Induction step: $i - p = 1 + prio'$}. Consider the arbiter's decision. Either the arbiter
  selects input $i$, then our goal follows just as in the base case. Or the arbiter selects a different input $j$.
  Then the packet at $j$ will eventually be forwarded by \hyperref[theo:noc:deadlockfree]{\textsc{Deadlock freedom}} causing the priority to get closer to $i$
  by \hyperref[lem:arb:progress]{\textsc{Arbiter progress}}~(\ref{lem:arb:progress}).
  With the smaller priority the conclusion then follows from the inner induction hypothesis.
\end{proof}

\subsection{Implementation}

For our \koika implementation, we only show the key parts in our
dependent hardware design.
For the full details, we refer to our \coq development.
In order to implement a $k$-dimensional \noc, we needed to create
the corresponding following depending on $k$ and the length of the
individual dimensions.
\koika supports such a dependent design using parameters to the
registers.
\begin{lstlisting}[style=frame,style=coq]
Inductive reg_t : Type :=
  | §4:routing_up_r§   : §3:idx_t§ -> §3:dim_t§ -> reg_t
  | §4:routing_down_r§ : §3:idx_t§ -> §3:dim_t§ -> reg_t
  | §4:arbiter_up_r§   : §3:idx_t§ -> §3:dim_t§ -> reg_t
  | §4:arbiter_down_r§ : §3:idx_t§ -> §3:dim_t§ -> reg_t
  | §4:arbiter_cpu_r§  : §3:idx_t§          -> reg_t.
\end{lstlisting}

Every channel has an index \texttt{idx\_t} in its dimension
\texttt{dim\_t}.
That is, every channel exists in exactly one dimension.
To make sure the set of registers for a design is finite and thus
can be synthesized to valid hardware, \koika requires parameters to be
finite too.
Hence, we worked with \texttt{Fin.t}, i.e., the natural numbers modulo
$k$, for \texttt{dim\_t} and modulo $m$ for \texttt{idx\_t} where
$m$ is the length of a single dimension.
Respectively, we used \coq as the meta programming language to generate \koika
functions and actions.
\definecolor{hlYellow}{rgb}{0.99,0.78,0.07}
\sethlcolor{hlYellow}
\begin{lstlisting}[style=frame,style=coq]
Fixpoint routing {dims} (lcl : §3:idx_t§ dims) : §3:function§ R Sigma := <|\hl{\texttt{<\{}}|>
fun routing (dst_a : addr_t dims) : §3:array_t§ (§3:bits_t§ §2§) (§3:length§ dims) =>
  <|\hl{`}|>match dims with
  | []           => fun _   => <|\hl{\texttt{<\{}}|> [ ] <|\hl{\texttt{\}>}}|>
  | dim :: dims' => fun lcl => <|\hl{\texttt{<\{}}|>
    let r1 := {§3:routing$_1$§ dim (fst lcl)}(§3:addr_hd§(dst_a)) in
    if r1 == §'b00§
      then §3:array_cons§(r1, {<|@|>§3:routing§ dims' (snd lcl)}(§3:addr_tail§(dst_a)) )
      else §3:array_cons§(r1, <|\$|>(§3:vect_const§ §'b00§) )
    <|\hl{\texttt{\}>}}|>
  end lcl<|\hl{`}|>
<|\hl{\texttt{\}>}}|>.
\end{lstlisting}

Note how the \texttt{routing} implementation enters into \koika via
\hl{\texttt{<\{}} (Line~1), then escapes into \coq via \hl{\texttt{`}} at Line~3 just
  to generate \koika code again at Line~4.
During our development, we found that \koika has powerful support for
meta programming that allows for very flexible hardware designs.

\subsection{Refinement}
The following refinements for the \texttt{routing},
\texttt{arbitration} and the overall \noc connect the execution of a single hardware cycle
to the formal specification.
Functions \texttt{repr} and \texttt{is\_repr} lower indexes and states from the
specification level into the implementation level.
\begin{lemma}[Routing refinement]\label{lem:route:refine}
  For any local index $l$ and destination index $d$ in the dimensional space $c$ (of the \noc),
  the following total Hoare triple holds:
\end{lemma}
\begin{center}
\thoare{\top}
{\texttt{<\{ routing s lcl (repr dst) \}>} }
{\texttt{v, v = repr (routing\_spec lcl dst)}}
\end{center}

\begin{proof}
  By applying our program logic for symbolic evaluation,
  we can reduce the goal to the following sub-goals with minimal effort.

\noindent
\hspace{0.5cm}
\begin{minipage}[t]{0.40\linewidth}
\begin{lstlisting}
Hif: §3:routing_spec$_1$§ l d = §4:Arrived§
@(goal 1) <|\color{gray}\rule[0.5ex]{\dimexpr\linewidth-1.4cm\relax}{0.4pt}|>@
§3:routing_spec§ (l,lcl') (d,dst') =
§3:routing_spec$_1$§ l d ::
  §3:routing_spec§ lcl' dst'
\end{lstlisting}
\end{minipage}
\hfill
\begin{minipage}[t]{0.40\linewidth}
\begin{lstlisting}
Hif: §3:routing_spec$_1$§ l d ≠ §4:Arrived§
@(goal 2) <|\color{gray}\rule[0.5ex]{\dimexpr\linewidth-1.4cm\relax}{0.4pt}|>@
§3:routing_spec§ (l,lcl') (d,dst') =
§3:routing_spec$_1$§ l d ::
  §3:vect_const§ §4:Irrelevant§
\end{lstlisting}
\end{minipage}
\hspace{0.5cm}
\\[0.3cm]
These goals are then finished by manual reasoning over the definitions of the specs.
\end{proof}

\begin{lemma}[Arbiter refinement]\label{lem:arb:refine}
  For any arbiter ring $a$ of size $n$ with priority at $p$,
  the following total Hoare triple holds:
\end{lemma}
\begin{center}
\thoare{\top}
{\texttt{<\{ arbiter n (repr p, repr a) \}>}}
{\texttt{v, v = repr (arbiter\_spec p a)}}
\end{center}
\begin{theorem}[\noc refinement]\label{theo:noc:refine}
  For any dimensional space $c$ (\texttt{dims}),
  the generated schedule \texttt{noc\_schedule} takes the \noc
  from state $s$ into state $s'$ and
  satisfies the following total triple:
\end{theorem}
\begin{center}
\thoare
{\texttt{ is\_repr s }}
{\texttt{ noc\_schedule dims }}
{\texttt{ is\_repr (ns\_{step} 1 s) }}
\end{center}

Due to our program logic in combination with the \oursem, the proofs for
these refinements are straight-forward by automated symbolic execution.

\section{Synthesis}
\label{sec:eval}

In this section, we evaluate \sys and show
that \ \textbf{our formal verification approach is necessary
to make \noc{s} in \koika practical}.
In order to study the practicality of \sys, we compile
a set of \noc{s} to Verilog and synthesize them with the
open-source toolchain from the original \koika paper~\cite{koika}.
The toolchain consists of Yosys~\cite{yosys}, for synthesis,
a process development kit for 45~nm~\cite{freepdk45}
and ABC~\cite{abc} for technology mapping.
In this synthesis flow, the high-level \noc design is mapped to
concrete gates and wires, yielding estimates on circuit size
and timing constraints.
However, as this workflow does not include floorplanning,
these estimates represent best-case bounds.
In particular, for \noc{s} that span more than two dimensions,
the actual performance is likely to degrade, since their topologies
are difficult to map onto a planar 2D silicon.
But this evaluation is outside the scope of this paper.

\begin{figure}
  \centering
  \begin{tikzpicture}
\begin{axis}[
    scatter,
    only marks,
    bar width=18pt,
    ymin=0,
    xmin=0,
    ylabel={delay [ps]},
    xlabel={\# of routers},
    xlabel style={
        at={(axis description cs:1,0)},
        anchor=south west,
    },
    xtick={2,8,16,32,64},
    enlarge y limits={upper, value=0.2},
    enlarge x limits={upper, value=0.1},
    height=5cm,
    width=10cm,
    point meta=explicit symbolic,
    nodes near coords,
    grid=major,
]
\addplot+[
    mark=*,
    blue,
    mark options={fill=blue, fill opacity=0.2},
    nodes near coords style={below, fill opacity=1}
] coordinates {
    ( 2,  358.53)
    ( 3,  494.83)
    ( 4,  541.15)
    ( 8,  757.67) [8]
    (16, 1132.60)
    (32, 1914.81) [32]
    (64, 3553.02) [64]
};
\node[above, blue, xshift=-0.1cm] at (axis cs: 2,  358.53) {2};
\node[below=-0.02cm, blue, xshift=0.05cm] at (axis cs: 3,  494.83) {3};
\node[below=-0.1cm, blue, xshift=0.13cm] at (axis cs: 4,  541.15) {4};
\node[above, blue] at (axis cs:16, 1132.60) {16};

\addplot+[
    mark=square*,
    red,
    mark options={fill=red, fill opacity=0.2},
    nodes near coords style={above, fill opacity=1}
] coordinates {
    ( 4,  605.29) [$2\!\!\times\!\!2$]
    ( 9,  894.38)
    (16, 1031.04)
};
\node[above, red, xshift=0.3cm] at (axis cs:9,  894.38) {$3\!\!\times\!\!3$};
\node[below, red] at (axis cs:16, 1031.04) {$4\!\!\times\!\!4$};

\addplot+[
    mark=triangle*,
    green,
    mark options={fill=green, fill opacity=0.2},
    nodes near coords style={left, fill opacity=1}
] coordinates {
    ( 8, 882.08)
};
\node[left=-0.2cm, green, yshift=0.3cm] at (axis cs:8, 882.08) {$2\!\!\times\!\!2\!\!\times\!\!2$};

\end{axis}
\end{tikzpicture}
  \caption{Critical path length (with transaction checks)}
  \label{fig:critical_path_guards}
\end{figure}
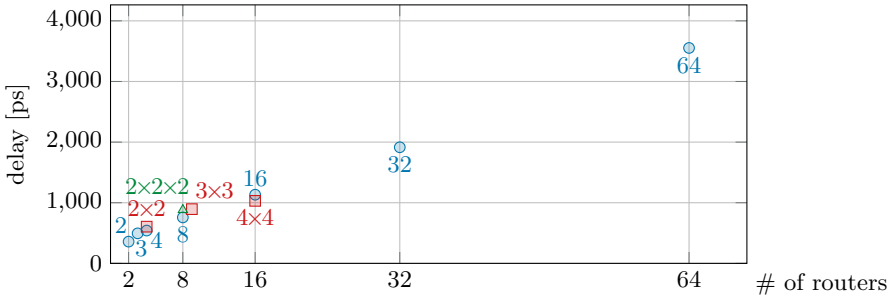

\Cref{fig:critical_path_guards} shows the length of the
\emph{critical path} for different instantiations of
\textcolor{blue}{one-dimensional},
\textcolor{red}{two-dimensional} and a
\textcolor{green}{three-dimensional} \noc{s}.
The critical path is the maximum time that the combinational
logic might take to compute the state of the next cycle.
This delay is particularly important because it directly limits
the achievable clock frequency and thereby the chip's overall
performance.
The plot shows that the critical path length grows linearly with
the number of nodes in the \noc.
This scaling is undesirable; as in an ideal architecture each
router operates independently and in parallel.
Thus, the critical path should only cross a single router and
remain essentially constant regardless of the network's size.
Surprisingly, in the case of our \noc, the scaling is not caused
by the \koika design itself.
Instead, it is introduced by the additional circuitry that the
compiler automatically generates to enforce the transactional
semantics of the language.
For \sys, this transactional logic is unnecessary runtime overhead
because our refinement proof already uses total Hoare triples.
As such, it guarantees that actions never fail (see
\cref{def:wpa:total} and
\cref{def:hoare:total} respectively).
Thus, we implemented a flag for the \koika compiler to disable
the generation of the transactional circuitry.
The plot in \cref{fig:critical_path} shows that this improves the
critical path length by an order of magnitude.
Most importantly, it caps the delay even when increasing the \noc
size thereby making \sys scalable and practical.

\begin{figure}
  \centering
  \begin{tikzpicture}
\begin{axis}[
    scatter,
    only marks,
    bar width=18pt,
    ymin=0,
    xmin=0,
    ylabel={delay [ps]},
    xlabel={\# of routers},
    xlabel style={
        at={(axis description cs:1,0)},
        anchor=south west,
    },
    xtick={2,8,16,32,64},
    enlarge y limits={upper, value=0.2},
    enlarge x limits={upper, value=0.1},
    height=5cm,
    width=10cm,
    point meta=explicit symbolic,
    nodes near coords,
    grid=major,
]

\addplot+[
    mark=*,
    blue,
    mark options={fill=blue, fill opacity=0.2},
    nodes near coords style={below, fill opacity=1}
] coordinates {
    ( 2, 369.67) [2]
    ( 3, 455.32)
    ( 4, 440.39) [4]
    ( 8, 476.11) [8]
    (16, 508.92) [16]
    (32, 521.23) [32]
    (64, 538.25) [64]
};
\node[left=-2pt, blue, yshift=-0.5em] at (axis cs:3, 455.32) {3};

\addplot+[
    mark=square*,
    red,
    mark options={fill=red, fill opacity=0.2},
    nodes near coords style={above=-0.02cm, fill opacity=1}
] coordinates {
    ( 4, 470.10) [$2\!\!\times\!\!2$]
    ( 9, 572.93)
    (16, 583.22) [$4\!\!\times\!\!4$]
};
\node[above, red, xshift=0.2em] at (axis cs:9, 572.93) {$3\!\!\times\!\!3$};

\addplot+[
    mark=triangle*,
    green,
    mark options={fill=green, fill opacity=0.2},
    nodes near coords style={below, fill opacity=1}
] coordinates {
    ( 8, 530.01)
};
\node[left=-0.1cm, green, yshift=0.23cm] at (axis cs:8, 530.01) {$2\!\!\times\!\!2\!\!\times\!\!2$};

\end{axis}
\end{tikzpicture}
  \caption{Critical path length (without transaction checks)}
  \label{fig:critical_path}
\end{figure}
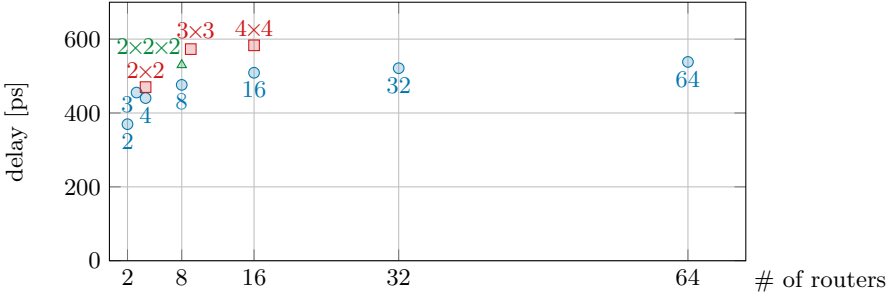

We point out that this compiler flag side-steps the original compiler
correctness theorem.
Re-establishing this guarantee requires to strengthen the correctness
statement of the compiler to preserve not only the semantics but also
the properties, i.e., the weakest pre-conditions, of our Hoare logic,
including the one for data-race freedom~\cite{wp_prisc_2026}.
This requires a considerable effort, in fact a re-design of the
\koika compiler, so we have deferred this to future work.

\section{Related and Future Work}
\label{sec:related}

In this paper, we contributed \oursem and the first program logic to
formally-verify concurrent rule-based hardware designs in
a modular fashion.
This allowed us to design \sys, a highly-parameterized
library that takes a configuration for $k$ dimensions with
varying size as an input and produces a corresponding \noc (as output).
All \noc{}s are formally-verified against a formal specification
and no additional verification effort is required.
We are unaware of such a library.
In these concluding paragraphs, we review these contributions
in the context of related work and identify interesting
directions for future work.

\paragraph*{Verification of concurrent hardware}
Several approaches exist to formally-verified hardware design~\cite{10.1145/307988.307989,braibant2013formal,herklotz21_fvhls}.
Rule-based hardware design was pioneered by the design of Bluespec SystemVerilog~\cite{bluespec}.
The key concept of Bluespec, i.e., the \emph{one-rule-at-a-time (ORAAT)}
semantics, allows designing concurrent hardware that is free of data races.
All three state-of-the-art rule-based HDLs -- \koika, \kami and \fjfj{} --
are derivatives of Bluespec (in \rocq~\cite{the_coq_development_team_2024_14542673}).
\koika~\cite{koika} studies a formally-verified compilation with
guarantees for ORAAT.
\kami~\cite{kami} studies the concept of (hardware) modules which export
methods that can be called by other modules as it is known for example
from object-oriented programming languages.
The work on \kami then investigates how to reason about modules, i.e., on
one module individually and then on their composition.
For that, \kami defines the semantics of their modules as a labeled
transition system such that reasoning can take place on their traces.
In \kami's semantics, a race-free trace then does not contain a particular
method call more than once.
In contrast to \koika, the \kami compiler derives the schedule -- as a
binary relation -- and hands it to the Bluespec SystemVerilog compiler
to enforce ORAAT during parallel execution.
This binary relation quickly explodes for more complex hardware designs
where only some but not all the actions fire, i.e., execute, in a single
hardware cycle.
This also influences the proofs in \kami that need to rule out
certain traces.
Therefore, \fjfj restricts the writes that are possible to the same register
in one rule.
Our approach is different.
Instead of restricting, we provide more freedom.
This is possible because \koika supports ephemeral history registers
while \kami and \fjfj do not.
Yet, the idea of our program logic would also work and benefit
\kami and \fjfj while modules are yet a dearly needed concept in \koika.

\paragraph*{Verification of \noc designs}

The literature even on formally-verified \noc designs is scarce~\cite{nocsverify}.
The most notably formal specification of a \noc is GeNoC~\cite{genoc}.
GeNoC was implemented in ACL2~\cite{chamarthi2011acl2} and used formally verify \noc hardware designs.
GeNoC focused primarily on the routers and their properties.
In our \sys design, we followed some of GeNoC's principles, such
as the separation between routing and arbitration, and left others to future work, such
as the idea of delaying transmission of messages.
Yet, our \noc designs are formally-verified
against a formal specification, while our functional
correctness propagates down to the generated Verilog.
This is far beyond what GeNoC does.
\sys marks a starting point and leaves room for various extensions for
future work, such as abstracting over the size of the channels,
generalizing over different routing algorithms and
performance evaluation of the generated \noc designs.

\begin{credits}
\subsubsection{\ackname}
The authors would like to thank Clément Pit-Claudel and
Thomas Bourgeat for insightful discussions on the Hoare logic and
the design of \koika.

This work was funded by the Agentur für Innovation
in der Cybersicherheit GmbH (Cyberagentur).

\subsubsection{\discintname}
The authors have no competing interests to declare that are relevant to the content of this article.
\end{credits}

\bibliographystyle{LNCS/splncs04}
\bibliography{references2, references}

@book{sva,
  title={Systemverilog assertions and functional coverage},
  author={Mehta, Ashok B},
  year={2020},
  publisher={Springer}
}

@misc{abc,
  author       = {Alan Mishchenko and others},
  title        = {{ABC: System for Sequential Logic Synthesis and Formal Verification}},
  howpublished = {\url{https://github.com/berkeley-abc/abc}},
  note         = {Accessed: 2026-05-11}
}

@misc{yosys,
  author       = {Clifford Wolf},
  title        = {{Yosys Open SYnthesis Suite}},
  howpublished = {\url{http://www.clifford.at/yosys/}},
  note         = {Accessed: 2026-05-11}
}

@misc{freepdk45,
  author       = {{Silvaco}},
  title        = {{45nm FreePDK}},
  howpublished = {\url{https://si2.org/open-cell-and-free-pdk-libraries/}},
  note         = {Accessed: 2026-05-11}
}

@article{virtualchannels,
author = {Dally, W. J. and Seitz, C. L.},
title = {Deadlock-Free Message Routing in Multiprocessor Interconnection Networks},
year = {1987},
issue_date = {May 1987},
publisher = {IEEE Computer Society},
address = {USA},
volume = {36},
number = {5},
issn = {0018-9340},
url = {https://doi.org/10.1109/TC.1987.1676939},
doi = {10.1109/TC.1987.1676939},
journal = {IEEE Trans. Comput.},
month = may,
pages = {547–553},
numpages = {7}
}

@inproceedings{nocsverify,
author="Zerdani, Amina
and Boutekkouk, Fateh",
editor="Laouar, Mohamed Ridda
and Balas, Valentina Emilia
and Piuri, Vincenzo
and Rad, Dana
and Touati Hamad, Zineb
and Cheddad, Abbas",
title="An Overview of Formal Verification of Network-on-Chip (NoC) Methods",
booktitle="13th International Conference on Information Systems and Advanced Technologies ``ICISAT 2023''",
year="2024",
publisher="Springer Nature Switzerland",
address="Cham",
pages="26--34",
isbn="978-3-031-60594-9"
}

@book{dally2004principles,
  title={Principles and practices of interconnection networks},
  author={Dally, William James and Towles, Brian Patrick},
  year={2004},
  publisher={Elsevier}
}

@article{bertozzi2004xpipes,
  title={Xpipes: A network-on-chip architecture for gigascale systems-on-chip},
  author={Bertozzi, Davide and Benini, Luca},
  journal={IEEE circuits and systems magazine},
  volume={4},
  number={2},
  pages={18--31},
  year={2004},
  publisher={IEEE}
}

@article{10.1145/1132952.1132953,
author = {Bjerregaard, Tobias and Mahadevan, Shankar},
title = {A survey of research and practices of Network-on-chip},
year = {2006},
issue_date = {2006},
publisher = {Association for Computing Machinery},
address = {New York, NY, USA},
volume = {38},
number = {1},
issn = {0360-0300},
url = {https://doi.org/10.1145/1132952.1132953},
doi = {10.1145/1132952.1132953},
journal = {ACM Comput. Surv.},
month = jun,
pages = {1–es},
numpages = {51}
}

@inproceedings{kumar2002network,
  title={A network on chip architecture and design methodology},
  author={Kumar, Shashi and Jantsch, Axel and Soininen, J-P and Forsell, Martti and Millberg, Mikael and Oberg, Johny and Tiensyrja, Kari and Hemani, Ahmed},
  booktitle={Proceedings IEEE Computer Society Annual Symposium on VLSI. New Paradigms for VLSI Systems Design. ISVLSI 2002},
  pages={117--124},
  year={2002},
  organization={IEEE}
}

@article{10.1145/3450964,
author = {Charles, Subodha and Mishra, Prabhat},
title = {A Survey of Network-on-Chip Security Attacks and Countermeasures},
year = {2021},
issue_date = {June 2022},
publisher = {Association for Computing Machinery},
address = {New York, NY, USA},
volume = {54},
number = {5},
issn = {0360-0300},
url = {https://doi.org/10.1145/3450964},
doi = {10.1145/3450964},
journal = {ACM Comput. Surv.},
month = may,
articleno = {101},
numpages = {36}
}

@inproceedings{braibant2013formal,
  title={Formal verification of hardware synthesis},
  author={Braibant, Thomas and Chlipala, Adam},
  booktitle={International Conference on Computer Aided Verification},
  pages={213--228},
  year={2013},
  organization={Springer}
}

@article{10.1145/307988.307989,
author = {Kern, Christoph and Greenstreet, Mark R.},
title = {Formal verification in hardware design: a survey},
year = {1999},
issue_date = {April 1999},
publisher = {Association for Computing Machinery},
address = {New York, NY, USA},
volume = {4},
number = {2},
issn = {1084-4309},
url = {https://doi.org/10.1145/307988.307989},
doi = {10.1145/307988.307989},
journal = {ACM Trans. Des. Autom. Electron. Syst.},
month = apr,
pages = {123–193},
numpages = {71}
}

@software{the_coq_development_team_2024_14542673,
  author       = {The Coq Development Team},
  title        = {The Coq Proof Assistant},
  month        = sep,
  year         = 2024,
  publisher    = {Zenodo},
  version      = {8.20},
  doi          = {10.5281/zenodo.14542673},
  url          = {https://doi.org/10.5281/zenodo.14542673},
}

@inproceedings{chamarthi2011acl2,
  title={The ACL2 sedan theorem proving system},
  author={Chamarthi, Harsh Raju and Dillinger, Peter and Manolios, Panagiotis and Vroon, Daron},
  booktitle={International Conference on Tools and Algorithms for the Construction and Analysis of Systems},
  pages={291--295},
  year={2011},
  organization={Springer}
}

@ARTICLE{9858921,
  author={Das, Surajit and Karfa, Chandan and Biswas, Santosh},
  journal={IEEE Access},
  title={Accelerating NoC Verification Using a Complete Model and Active Window},
  year={2022},
  volume={10},
  number={},
  pages={88985-88999},
  doi={10.1109/ACCESS.2022.3199671}
}

@inproceedings{genoc,
  author = {Schmaltz, Julien and Borrione, Dominique},
  title = {A generic network on chip model},
  year = {2005},
  isbn = {3540283722},
  publisher = {Springer-Verlag},
  address = {Berlin, Heidelberg},
  url = {https://doi.org/10.1007/11541868_20},
  doi = {10.1007/11541868_20},
  booktitle = {Proceedings of the 18th International Conference on Theorem Proving in Higher Order Logics},
  pages = {310–325},
  numpages = {16},
  location = {Oxford, UK},
  series = {TPHOLs'05}
}

@inproceedings{koika,
  author = {Bourgeat, Thomas and Pit-Claudel, Cl\'{e}ment and Chlipala, Adam and Arvind},
  title = {The essence of Bluespec: a core language for rule-based hardware design},
  year = {2020},
  isbn = {9781450376136},
  publisher = {Association for Computing Machinery},
  address = {New York, NY, USA},
  url = {https://doi.org/10.1145/3385412.3385965},
  doi = {10.1145/3385412.3385965},
  booktitle = {Proceedings of the 41st ACM SIGPLAN Conference on Programming Language Design and Implementation},
  pages = {243–257},
  numpages = {15},
  location = {London, UK},
  series = {PLDI 2020}
}

@article{kami,
  author = {Choi, Joonwon and Vijayaraghavan, Muralidaran and Sherman, Benjamin and Chlipala, Adam and Arvind},
  title = {Kami: a platform for high-level parametric hardware specification and its modular verification},
  year = {2017},
  issue_date = {September 2017},
  publisher = {Association for Computing Machinery},
  address = {New York, NY, USA},
  volume = {1},
  number = {ICFP},
  url = {https://doi.org/10.1145/3110268},
  doi = {10.1145/3110268},
  journal = {Proc. ACM Program. Lang.},
  month = aug,
  articleno = {24},
  numpages = {30}
}

@article{fjfj,
  author = {Bourgeat, Thomas and Liu, Jiazheng and Chlipala, Adam and Arvind},
  title = {Making Concurrent Hardware Verification Sequential},
  year = {2025},
  issue_date = {June 2025},
  publisher = {Association for Computing Machinery},
  address = {New York, NY, USA},
  volume = {9},
  number = {PLDI},
  url = {https://doi.org/10.1145/3729331},
  doi = {10.1145/3729331},
  journal = {Proc. ACM Program. Lang.},
  month = jun,
  articleno = {228},
  numpages = {25}
}

@Manual{rocq,
  title =        {The Rocq prover reference manual},
  author =       {The Rocq development team},
  organization = {LogiCal Project},
  note =         {Version 9.1.0},
  year =         {2025},
  url =          "https://rocq-prover.org/doc/V9.1.0/refman/index.html"
}

@inproceedings{herklotz21_fvhls,
  author = {Herklotz, Yann and Pollard, James D. and Ramanathan, Nadesh and Wickerson, John},
  title = {Formal Verification of High-Level Synthesis},
  year = {2021},
  number = {OOPSLA},
  numpages = {30},
  month = {11},
  journal = {Proc. ACM Program. Lang.},
  volume = {5},
  publisher = {Association for Computing Machinery},
  address = {New York, NY, USA},
  doi = {10.1145/3485494}
}

@inproceedings{bluespec,
  author={Nikhil, R.},
  booktitle={Proceedings. Second ACM and IEEE International Conference on Formal Methods and Models for Co-Design, 2004. MEMOCODE '04.},
  title={Bluespec System Verilog: efficient, correct RTL from high level specifications},
  year={2004},
  volume={},
  number={},
  pages={69-70},
  doi={10.1109/MEMCOD.2004.1459818}
  }

@inproceedings{wp_prisc_2026,
title = {WP-Preserving Compilation -- Preserving Weakest Preconditions For End-to-End Verification},
author = {Abate, Carmine and Elsheikh, Mohamed and Liotati, Kleio and Farka, Franti\v{s}ek and Ertel, Sebastian},
year = {2026},
booktitle = {10th Workshop on Principles of Secure Compilation},
location = {Rennes, France},
series = {PriSC '26}
}

\appendix

\section{GenAI Declaration}

Neither the code nor the paper was written with any kind of
support from generative AI.

\section{Hoare Rules for \koika}\label{sec:appendix:hoare}
\FloatBarrier
\begin{figure}
  \centering
\begin{tikzpicture}
  \matrix (m) [
  matrix of math nodes,
  ampersand replacement=\&,
  column sep=0.5cm
  ] {
    \inference
    { \pre{P} \rightarrow \pre{P'} & \hoare{P'}{S}{Q} }
    { \hoare{P}{S}{Q} }
    [\textsc{StrenPre}]
    \&
    \inference
    { \hoare{P}{S}{Q'} & \post{Q'} \rightarrow \post{Q} }
    { \hoare{P}{S}{Q} }
    [\textsc{WeakPost}]
    \\
};
\end{tikzpicture}

\caption{Structural rules.}
\label{fig:rules:structural}

\end{figure}
\begin{figure}
\begin{tikzpicture}[node distance=0cm]

\matrix (m) [
  matrix of math nodes,
  ampersand replacement=\&,
  column sep=0.25cm
  ] {
    \inference
    { }
    { \hoare{Q~(b)}{b}{Q} }
    [\textsc{BitStrings}]
    \&
    \inference
    { }
    { \hoare{ Q~(\Ctx[x]) }{ x }{ Q } }
    [\textsc{Var}]
    \\
};

\matrix (m1) [below =of m] [
  matrix of math nodes,
  ampersand replacement=\&,
  column sep=0.25cm
  ] {
    \inference
    {  }
    { \hoare
      { Q~(\R[r]) }
      { \kread{0}{r} }
      { Q }
    }
    [\textsc{HRead0}]
    \\
};

\matrix (m11) [below =of m1] [
  matrix of math nodes,
  ampersand replacement=\&,
  column sep=0.25cm
  ] {
    \inference
    {  }
    { \hoare
      { Q~(\cL[r]_0 ~ \texttt{?\!:} ~ \R[r])) }
      { \kread{1}{r} }
      { Q }
    }
    [\textsc{HRead1}]
    \\
};

\matrix (m1a) [below =of m11] [
  matrix of math nodes,
  ampersand replacement=\&,
  column sep=0.25cm
  ] {
    \inference
    { \hoare
      { P }
      { a }
      { v. ~ Q~(\cL \mdoubleplus [\kwrite{p}{r}{v}], ~ \epsilon) }
    }
    { \hoare{ P }{ \kwrite{p}{r}{a} }{ Q } }
    [\textsc{HWrite}]
    \\
};

\matrix (m2) [below =of m1a] [
  matrix of math nodes,
  ampersand replacement=\&,
  column sep=0.25cm
  ] {
  \inference
  {
    \hoare{ P }{ a_1 }{ v. ~ Q ~ (\Ctx \oplus [x \mapsto v], ~ \epsilon) } \quad
    \hoare{ Q }{ a_2 }{ v. ~ R ~ (\Ctx \setminus \{ x \},~ v) } \quad
    x \in \texttt{Var}
  }
  { \hoare{P}{\texttt{let} ~ x ~ \texttt{:=} ~ a_1 ~ \texttt{in} ~ a_2 }{R} }
  [\textsc{Bind}]
  \\
};

\matrix (m3) [below =of m2] [
  matrix of math nodes,
  ampersand replacement=\&,
  column sep=0.25cm
  ] {
  \inference
  { \hoare{P}{a_{c}}{v. ~ \texttt{if $v$ then} ~ Q_{t} ~ \texttt{else} ~ Q_{f }}
      & \hoare{Q_{t}}{a_{t}}{R}
      & \hoare{Q_{f}}{a_{f}}{R} }
  { \hoare{P}{\texttt{if} ~ a_{c} ~ \texttt{then} ~ a_{t} ~ \texttt{else} ~ a_{f}}{R} }
  [\textsc{If}]
  \\
};

\matrix (m4) [below =of m3] [
  matrix of math nodes,
  ampersand replacement=\&,
  column sep=0.25cm
  ] {
  \inference
  {
    \pre{\{P_0\}} ~ a_1 ~ \post{\{v_1. ~ P_1~(\cL_1, (v1))\}} \\
    \pre{\{P_1\}} ~ a_2 ~ \post{\{v_2. ~ P_2~(\cL_2, (v1,v2))\}} \\
    \ldots \\
    \pre{\{P_{n-1}\}} ~ a_n ~ \post{\{v_n. ~ P_n(\cL_n,~ (v_1,\ldots,v_n))\}} \\
    \hoare{P_n ~ (v_1,\ldots,v_n)}{f}{R}
  }
  { \hoare{P_0}{f\texttt{(}a_1, \ldots, a_n\texttt{)}}{R} }
  [\textsc{Call}]
  \\
};

\matrix (m5) [below =of m4] [
  matrix of math nodes,
  ampersand replacement=\&,
  column sep=0.25cm
  ] {
  \inference
  {  }
  { \hoare{P}{\texttt{abort}}{\bot} }
  [\textsc{Abort}]
  \\
};

\end{tikzpicture}

  \caption{Reasoning rules for \koika actions.}
  \label{fig:rules:term}

\end{figure}
The rules for reasoning in our program logic classically fall
into two categories: structural rules and term rules.
Structural rules are independent of the \koika actions that the
Hoare triples reason about.
\begin{definition}[Structural rules]
  The structural rules to reason about \koika actions are
  defined in \cref{fig:rules:structural}.
\end{definition}
Rules~\textsc{StrenPre} and \textsc{WeakPost} are standard rules
for a Hoare logic and allow exchanging the pre- and postcondition
respectively.
The term rules state how an individual syntactic construct
in a \koika action interacts with the state, i.e., the register
environment \R, the log \cL and the variable context \Ctx.
\begin{definition}[Action rules]
  The term rules for \koika actions of our program logic
  are defined in \cref{fig:rules:term}.
\end{definition}
The rules for bit strings, variables and reads of
register values all follow the same principle.
All of these expressions are side-effect free.
As such, whatever condition \post{$Q$} should hold after
the execution of this construct should also hold before.
Note again that \post{$Q$} here actually means \post{$v.~ Q (v)$}.
For variables, we retrieve $v$ from the context \Ctx.
For \kread{0}{\cdot}, we load $v$ from the register environment \R.
And for \kread{1}{\cdot}, we either find \kwrite{0}{r}{v} in the
log \cL or resort to load again from \R.
For a write to a register, we require that $Q$ holds when
passed the log \cL with the appended \kwrite{p}{r}{v}.
For conditionals, we formulate a join point, i.e.,
a postcondition \post{R} that needs to hold after
execution of either of the two branches.
A call requires a Hoare triple for each of its arguments. These triples
are structured s.t. they accumulate their values, which are then
passed to the triple for the function body. Note that these values
are passed as variable context, which makes all of them
available as local variables in the function body. What's more -- they
even replace the context \Ctx{} -- preventing $f$ from accessing
the caller's local variables.
Finally, the \textsc{Abort} rule maps an \texttt{abort}
to $\bot$ and thereby specifies our Hoare triples as partial.
For total Hoare triples denoted as \thoare{P}{a}{Q},
we drop this rule and thereby require that the action never
aborts.
If it happens in the course of a proof then the proof is simply
stuck.

\end{document}